\documentclass[aps,pra,twocolumn,floatfix,superscriptaddress,longbibliography,10pt,notitlepage,nofootinbib]{revtex4-1}
\usepackage{graphicx, color, graphpap}
\usepackage{float}
\usepackage{enumitem}
\usepackage{amssymb}
\usepackage{amsfonts}
\usepackage{amsthm}
\usepackage{multirow}
\usepackage[dvipsnames]{xcolor}
\usepackage[colorlinks=true,citecolor=MidnightBlue,linkcolor=magenta]{hyperref}
\usepackage[T1]{fontenc}
\usepackage{bbm}
\usepackage{thmtools,thm-restate}
\usepackage{verbatim}
\usepackage{mathtools}
\usepackage{titlesec}
\usepackage{dsfont}
\usepackage{diagbox}
\usepackage[caption = false]{subfig}
\usepackage{subfig}
\usepackage{tikz}
\usetikzlibrary{quantikz}
\usepackage[normalem]{ulem}

\long\def\ca#1\cb{}

\newcommand{\ketbra}[2]{| \hspace{1pt} #1 \rangle\! \langle #2 \hspace{1pt} |}

\newcommand{\dya}[1]{\ket{#1}\!\bra{#1}}

\newcommand{\DC}{\mathcal{D}}
\newcommand{\D}{\mathcal{D}}

\newcommand{\NC}{\mathcal{N}}
\newcommand{\N}{\mathcal{N}}

\newcommand{\PC}{\mathcal{P}}

\newcommand{\SC}{\mathcal{S}}

\newcommand{\UC}{\mathcal{U}}
\newcommand{\U}{\mathcal{U}}

\newcommand{\YC}{\mathcal{Y}}

\newcommand{\Tr}{{\rm Tr}}

\renewcommand{\geq}{\geqslant}
\renewcommand{\leq}{\leqslant}

\renewcommand{\vec}[1]{\boldsymbol{#1}}  

\newcommand*{\id}{\mathbbm{I}}

\def\bmv{ {\vec{v}} }

\newtheorem{theorem}{Theorem}
\newtheorem{lemma}{Lemma}

\newtheorem{corollary}{Corollary}
\newtheorem{proposition}{Proposition}

\newtheorem{definition}{Definition}

\makeatletter
\newcommand{\dbloverline}[1]{\overline{\dbl@overline{#1}}}
\newcommand{\dbl@overline}[1]{\mathpalette\dbl@@overline{#1}}
\newcommand{\dbl@@overline}[2]{%
  \begingroup
  \sbox\z@{$\m@th#1\overline{#2}$}%
  \ht\z@=\dimexpr\ht\z@-2\dbl@adjust{#1}\relax
  \box\z@
  \ifx#1\scriptstyle\kern-\scriptspace\else
  \ifx#1\scriptscriptstyle\kern-\scriptspace\fi\fi
  \endgroup
}
\newcommand{\dbl@adjust}[1]{%
  \fontdimen8
  \ifx#1\displaystyle\textfont\else
  \ifx#1\textstyle\textfont\else
  \ifx#1\scriptstyle\scriptfont\else
  \scriptscriptfont\fi\fi\fi 3
}
\makeatother

\usepackage[usestackEOL]{stackengine}[2013-10-15]

\stackMath

\begin{document}

\title{A framework of partial error correction for intermediate-scale quantum computers}

\author{Nikolaos Koukoulekidis}
\email{nkouk96@gmail.com}
\affiliation{Theoretical Division, Los Alamos National Laboratory, Los Alamos, NM, USA}
\affiliation{Department of Physics, Imperial College London, London, UK}

\author{Samson Wang}
\affiliation{Department of Physics, Imperial College London, London, UK}

\author{Tom O'Leary}
\affiliation{Theoretical Division, Los Alamos National Laboratory, Los Alamos, NM, USA}
\affiliation{Department of Physics, Clarendon Laboratory, University of Oxford, Oxford, UK}

\author{Daniel Bultrini}
\affiliation{Theoretische Chemie, Physikalisch-Chemisches Institut, Universit{\"a}t Heidelberg, Heidelberg, Germany}

\author{Lukasz Cincio}
\affiliation{Theoretical Division, Los Alamos National Laboratory, Los Alamos, NM, USA}
\affiliation{Quantum Science Center, Oak Ridge, TN, USA}

\author{Piotr Czarnik}
\affiliation{Institute of Theoretical Physics, Jagiellonian University, Krak\'ow, Poland.}
\affiliation{Mark Kac Center for Complex Systems Research, Jagiellonian University, Krak\'ow, Poland}

\begin{abstract}
\vspace{5pt}
\begin{center}
    \textbf{ABSTRACT}
\end{center}

As quantum computing hardware steadily increases in qubit count and quality, one important question is how to allocate these resources to mitigate the effects of hardware noise. 
In a transitional era between noisy small-scale and fully fault-tolerant systems, we envisage a scenario in which we are only able to error-correct a fraction of the qubits required to perform an interesting computation. 
In this work, we develop concrete constructions of logical operations on a joint system of a collection of noisy and a collection of error-corrected logical qubits. 
Within this setting and under Pauli noise assumptions, we provide analytic evidence that brick-layered circuits display on average slower concentration to the ``useless'' uniform distribution with increasing circuit depth compared to fully noisy circuits. 
We corroborate these findings by numerical demonstration of slower decoherence with an increasing fraction of error-corrected qubits under depolarizing noise acting at the circuit level. 
We find that this advantage only manifests when the number of error-corrected qubits passes a specified threshold which depends on the number of couplings between error-corrected and noisy registers.
\end{abstract}

\maketitle

\section*{\MakeUppercase{Introduction}}

In the present noisy intermediate-scale quantum (NISQ) computing era we have access to a growing number of devices with hundreds of qubits \cite{preskill2018quantum}. 
Current NISQ algorithm development aims to use physical qubits directly and circumvent noise through various techniques such as  classical compute assistance \cite{bharti2021noisy, cerezo2020variationalreview, peng2020simulating, piveteau2023circuit}, noise-aware circuit compilation~\cite{cincio2018learning,murali2019noise}, and error mitigation methods \cite{temme2017error, cai2022quantum}. 
On the other hand, fault-tolerant implementations of quantum algorithms assume access to a sufficiently large amount of resources so that error detection and correction can occur throughout the computation \cite{gottesman2009introduction, campbell2017roads}. 
All the while, recent progress in hardware has brought exciting milestones in materializing the first generation of error-corrected qubits \cite{google2025quantum,mayer2024benchmarking,bluvstein2023logical,pogorelov2024experimental,postler2022demonstration,egan2021fault,livingston2022experimental,zhao2022realization,krinner2022realizing,acharya2022suppressing,sundaresan2023demonstrating,ryan2021realization,sivak2023real}.
As we move towards a transitional period between the NISQ and fault-tolerant eras, one important question is how to optimally distribute resources between computational and error-correcting tasks in order to materialize any possible quantum advantage in the earliest time frame~\cite{cao2021nisq, holmes2020nisq+, suzuki2022quantum, piveteau2021error, wahl2023zero, self2022protecting,pogorelov2024experimental,yin2025flexion,akahoshi2024partially,dangwal2025variational,akahoshi2024compilation,toshio2025practical}.

One simple framework in the regime of limited quantum error correction (QEC) is to take a collection of error-corrected qubits and couple them to noisy, non-error-corrected, qubits, which from hereon we sometimes refer to as the ``partial error correction'' framework.  Furthermore, we may refer to error-corrected qubits as ``clean qubits''. Since the early QEC implementations only allow for a very limited number of high-quality QEC qubits, this framework has the potential to greatly augment the computational space of a quantum computer utilizing QEC. This framework can also be employed to describe the trade-off between a less robust but more qubit-efficient error correction strategy on all available qubits, or a more robust but less qubit-efficient strategy on a fraction of the qubits. Such strategies can be beneficial if most errors occur only on a small fraction of the logical qubits, making their suppression a priority.

There has yet been no attempt in the literature to develop a concrete framework that considers performing error correction on a fraction of the logical space to obtain a computational advantage.
We note that this framework is distinct from the ``one clean qubit model'' \cite{knill1998power,fujii2016power,morimae2017power}, which considers computation with one single-qubit state coupled to a maximally mixed state, but crucially without noise throughout the circuit.
It is motivated by recent work~\cite{bultrini2022battle} which considers noiseless qubits combined with noisy qubits leading to slower concentration of expectation values in specific architectures.
It is also distinct from recent partially-fault-tolerant error correction schemes that interleave fault-tolerant Clifford gates with non-fault-tolerant Pauli rotations~\cite{akahoshi2024partially,dangwal2025variational,akahoshi2024compilation,toshio2025practical}, thus performing partial error correction on every register.

The present work explores explicit implementations and subsequent analyses of the partial error correction framework. 
Specifically, we propose a concrete implementation of universal two-qubit gates that interact with clean and noisy qubits. 
Our implementation depends on the QEC code characteristics, is optimal in the number of physical two-qubit gates required and provides a framework to explore quantum circuits with clean and noisy registers.

With explicit implementations established, we then explore possible advantages and trade-offs of the partial error correction framework. 
Using a circuit model with Pauli noise applied on the noisy and clean registers, we prove that the output state of an ensemble of brick-layered circuits can, on average, converge more slowly to the maximally mixed state than a fully noisy circuit.
However, this behavior is conditioned on the number of clean qubits passing a non-zero threshold that depends on the number of couplings between clean and noisy registers. 
We note that this fundamentally distinguishes our results from the case where error-corrected qubits are modeled as noiseless qubits~\cite{bultrini2022battle}, where a single noiseless qubit is sufficient for advantage.
Furthermore, we observe these effects in numerical simulations. We use randomized benchmarking to assess the errors of the different types of logical two-qubit gates required in the partial error correction framework. We follow this with demonstration of a threshold on the number of clean qubits above which our framework obtains an advantage for brick-layered circuits that are simulated with depolarizing and device-inspired noise models. 

The results of this manuscript are structured as follows. In Part A, we explicitly outline how to materialize the partially error-corrected framework. In Part B, we present our analytic scaling results for a model for the partially error-corrected setting with a brick-layered circuit. In Part C, we detail our numerical investigations. Finally, in the Discussion, we conclude and provide open questions stemming from our framework.

\section*{\MakeUppercase{Results}}

\subsection*{Part A: Framework for Partial Error Correction: Stabilizer codes with transversal gates}
\label{sec:partial-ec-framework}

Stabilizer codes~\cite{kitaev2003fault,dennis2002topological,raussendorf2007fault,fowler2012surface} form the leading candidate for realizing quantum error correction.
They are defined by a set of Pauli operators $\SC$, called stabilizers, which leave the encoded state unchanged and are usually characterized as $[[N, K, D]]$ codes, where $N$ is the number of physical qubits required to encode $K$ logical qubits with distance $D$.
A Pauli operator in $\SC$ has weight $w$ if it can be expressed as a tensor product of $w$ non-identity single-qubit Pauli operators.
The distance indicates how many physical errors are needed to cause a logical error. 
Formally, it is defined as the minimum weight of any Pauli operator $P$ such that $P \notin \SC$ and the commutation relation $[P,Q]=0$ holds for all $Q \in \SC$.
The surface code $[[N, 1, O(\sqrt{N})]]$ constitutes a prominent family of stabilizer codes, with recent promising experimental demonstrations~\cite{krinner2022realizing,zhao2022realization,acharya2022suppressing,sundaresan2023demonstrating}.

Stabilizer codes admit logical Pauli operators $\overline{X}$ and $\overline{Z}$ that allow logical computation.
For the surface code, they take the form
\begin{equation}\label{eq:logical_pauli}
    \overline{X} = \prod_{i=1}^w X_i; \quad \overline{Z} = \prod_{i=1}^w Z_i \,,
\end{equation}
such that they follow the Pauli algebra, i.e.~they anti-commute with each other $\overline{X}\overline{Z} + \overline{Z}\overline{X} = 0$ and $[\overline{X},Q] = [\overline{Z},Q] = 0$ for all $Q \in \SC$.
In general, the weight of $\overline{X}, \overline{Z}$ is at least equal to the code distance, $w \geq D$, meaning that an undesired logical operator can cause a logical error.
The logical state $\ket{\overline{0}}$ corresponds to all physical qubits being prepared in state $\ket{0}^{\otimes N}$ followed by a projection on the $(+1)$-eigenspace of all stabilizers,
\begin{equation}\label{eq:zerological}
    \ket{\overline{0}} \propto \prod_{P \in \SC} (\id^{\otimes N} + P) \ket{0}^{\otimes N}\,.
\end{equation}
The logical state $\ket{\overline{1}}$ is then simply defined as $\overline{X}\ket{\overline{0}}$.

Another important class of codes is the family of color codes~\cite{bombin2006topological,bombin2013self,bombin2015gauge,kubica2015universal,brown2016fault}, which is seeing significant experimental realizations~\cite{mayer2024benchmarking,bluvstein2023logical,postler2022demonstration,nigg2014quantum,pogorelov2024experimental}. 
On a 2-dimensional lattice, color codes have the characteristic property that Clifford operations, i.e.~the operations that preserve the Pauli group, can be implemented transversally.
An operation is transversal if the $i$'th physical qubit of a code can only interact with the $i$'th qubits of other codes~\cite{gottesman2009introduction}.
This property is important as it ensures that errors propagate only locally.
The Clifford gates on physical qubits are generated by $\langle H, S, CNOT \rangle$, so the logical Clifford gate generators are denoted by $\langle \overline{H}, \overline{S}, \overline{CNOT} \rangle$, where, for color codes, one can choose to express them as

\begin{equation} \label{eq:transgen}
    \begin{split}
    \overline{H} & = \prod_{i=1}^w H_i; \quad
    \overline{S} = \prod_{i=1}^w Z_i S_i; \\
    \overline{CNOT}_{a b} & = \prod_{i=1}^w CNOT_{a(i)\hspace{1pt} b(i)} \; ,
    \end{split}
\end{equation}
where $a,b$ index logical qubits and $a(i),b(i)$ index their constituent $i$'th physical qubit.
The Steane code~\cite{steane1996multiple} is the simplest example of a 2-dimensional color code, with $w=7$.

The planar surface code admits transversal $CNOT$ gates, which facilitates 2-qubit logical operations.
It has further been shown that variants of the surface code are equivalent to color codes up to local unitary transformations~\cite{kubica2015unfolding}, a result that may illuminate the transversal implementation of certain logical gates on the surface code, such as the controlled-controlled-$Z$ gate~\cite{vasmer2019three} for the 3-dimensional surface code.

\subsection*{Clifford gates on clean-noisy qubit pairs}\label{sec:clifford-gates}

We want to explore the interaction between a noisy qubit in state $\rho$ which can be subjected to noise and an encoded qubit in state $\overline{\psi}$ which undergoes some QEC cycle that aims to detect and correct errors according to a QEC code.
In particular, we want to find an expression for the logical $CNOT$ gate between a clean and a noisy register, which we denote as $\widetilde{CNOT}$. 
This gate, along with single-qubit gates, allows for implementing any Clifford operation on a partially error-corrected setup.

The effect of a physical $CNOT$ gate between two qubits can be fully described by its action on the 2-qubit Pauli operators, summarized as
\begin{align} \label{eq:unencoded}
    &\id \otimes X \rightarrow \id \otimes X; \quad &&X \otimes \id \rightarrow X \otimes X; \nonumber\\
    &\id \otimes Z \rightarrow Z \otimes Z; \quad &&Z \otimes \id \rightarrow Z \otimes \id,
\end{align}
where the first qubit acts as the control and the second as the target.
We want to achieve the same logical operation when the control qubit is in a noisy state $\rho$ and the target qubit is in the encoded state $\overline{\psi}$.
Namely, we need to find a sequence of physical operations that achieves
\begin{align} \label{eq:encoded}
    &\id \otimes \overline{X} \rightarrow \id \otimes \overline{X}; \quad &&X \otimes \overline{\id} \rightarrow X \otimes \overline{X}; \nonumber\\
    &\id \otimes \overline{Z} \rightarrow Z \otimes \overline{Z}; \quad &&Z \otimes \overline{\id} \rightarrow Z \otimes \overline{\id}.
\end{align}

As a guiding example, we first demonstrate this on the bit-flip repetition code, which corrects for single-qubit $X$ errors.
The encoding is $\ket{\overline{0}} = \ket{000}$ and $\ket{\overline{1}} = \ket{111}$, so that a general single-qubit state $\ket{\psi} \coloneqq \alpha \ket{0} + \beta \ket{1}$ can be encoded as
\begin{equation}
    \ket{\psi} \rightarrow \ket{\overline{\psi}}_{\rm bit} = \alpha \ket{\overline{0}} + \beta \ket{\overline{1}} = \alpha \ket{000} + \beta \ket{111},
\end{equation}
and $\overline{\psi}_{\rm bit} \coloneqq \ketbra{\overline{\psi}_{\rm bit}}{\overline{\psi}_{\rm bit}}$. 
We assume a pure encoded state for clarity of exposition.
Using index $0$ for the noisy qubit and $1, \dots, 3$ for the physical qubits that make up the clean qubit, as in Fig.~\ref{fig:bit-flip}, the logical Pauli gates are $\overline{X} = X_1X_2X_3$ and $\overline{Z} = Z_1Z_2Z_3$.

\begin{figure}[t]
    \centering
    \includegraphics[width=0.9\linewidth]{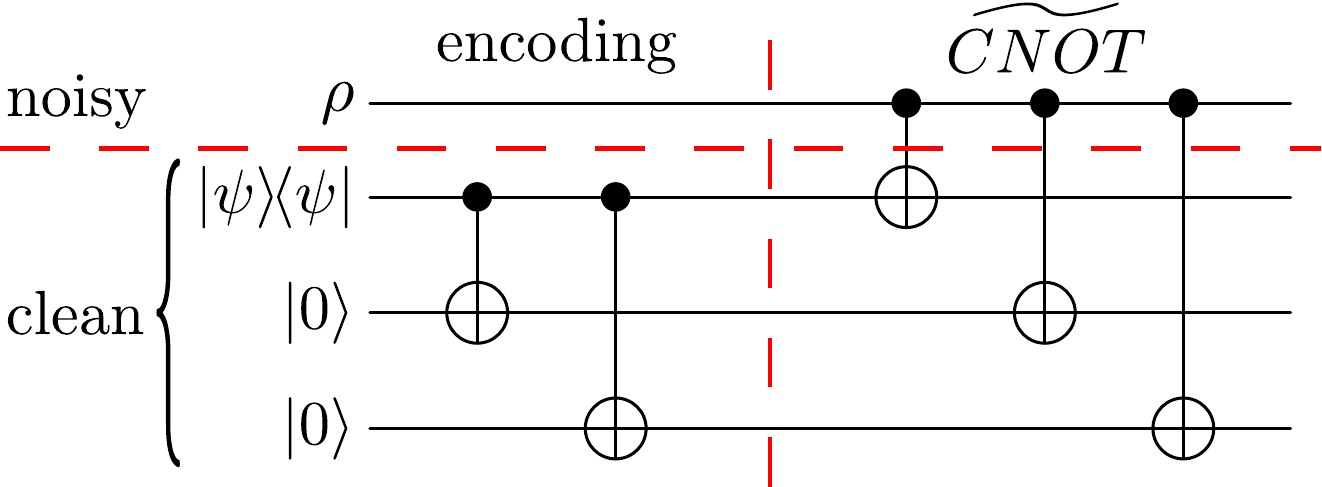} 
    \caption{ {\bf Implementation of a clean-noisy gate.} Diagram illustrating the interaction between a noisy and an encoded qubit, where the code is the bit-flip repetition code.}
    \label{fig:bit-flip}
\end{figure} 

In this example, it turns out we can simply express the gate $\widetilde{CNOT}$ as a sequence of the physical gates $CNOT_{01}, CNOT_{02}, CNOT_{03}$, so that it has the desired effect on $\rho \otimes \overline{\psi}$, as given in Eq.~(\ref{eq:encoded}).
Another instructive example is the phase-flip repetition code, which has logical Pauli $X$ operator given by $\overline{X} = Z_1Z_2Z_3$, so the gate $\widetilde{CNOT}$ must be implemented as a sequence of the physical controlled-Z gates $CZ_{01}, CZ_{02}, CZ_{03}$.

This result can be generalized for any code that admits transversal implementation of the logical Pauli operators, such as surface codes and color codes.
\begin{theorem}\label{thm:cnotboundary}
    Let $C$ be an $n$-qubit circuit, with noisy qubits on registers $1, \dots, n_d$ and clean qubits on registers $n_d+1, \dots, n$. 
    Suppose that the clean qubits are encoded by a given QEC code that admits logical Pauli operators with weight $w$, as expressed in Eq.~(\ref{eq:logical_pauli}).
    
    Then, all Clifford operations on $C$ are generated by 
    \begin{equation} \label{eq:mixed_cliff_gen}
        \langle H_i, S_i, CNOT_{ij}, \widetilde{CNOT}_{i a}, \overline{H}_{a}, \overline{S}_{a}, \overline{CNOT}_{a b} \rangle,
    \end{equation}
    where $1 \leq i,j \leq n_d$,\quad $n_d+1 \leq a,b \leq n$ and the gate
    \begin{equation}
        \widetilde{CNOT}_{i a} \coloneqq \ketbra{0}{0}_i \otimes \id_a + \ketbra{1}{1}_i \otimes \overline{X}_a
    \end{equation}
    can be implemented using exactly $w$ physical two-qubit gates.
\end{theorem}
\begin{proof}
   All Clifford gates between noisy qubits are generated by $\langle H, S, CNOT \rangle$ and all Clifford gates restricted on clean qubits are generated by $\langle \overline{H}, \overline{S}, \overline{CNOT} \rangle$. 

   We can then generate all Clifford gates between a noisy qubit on register $i$ and a clean qubit on register $a$ by combining single-qubit gates 
   \begin{equation}
       \langle H_i \otimes \overline{\id}_a, \id_i \otimes \overline{H}_a, S_i \otimes  \overline{\id}_a, \id_i \otimes \overline{S}_a \rangle
   \end{equation}
    with $\widetilde{CNOT}_{i a}$.
    In particular, we can change the role of control and target, 
    \begin{equation}\label{eq:Hproperty}
        \widetilde{CNOT}_{a i} = (H_i \otimes \overline{H}_a) \widetilde{CNOT}_{i a} (H_i \otimes \overline{H}_a),
    \end{equation}
    hence we can generate any Clifford operation.
    
    Assuming we do not have access to $m$-qubit physical gates for $m>2$, we need at least $w$ gates to implement the logical operator $\overline{X}_a$ on register $a$, controlled on register $i$. 
    Let the logical $X$ operator be 
    \begin{equation}
        \overline{X}_a = \bigotimes_{\kappa=1}^w P_{a(\kappa)}
    \end{equation}    
    where $P_{a(\kappa)}$ denotes the Pauli operator $P$ on the $\kappa$'th physical register of the $a$'th logical qubit.
    It then suffices to use the following sequence of $w$ physical controlled gates to implement a $\widetilde{CNOT}_{i a}$,
    \begin{equation}
        \widetilde{CNOT}_{i a} = \prod_{\kappa=1}^w \Big( \ketbra{0}{0}_i \otimes \id_{a(\kappa)} + \ketbra{1}{1}_i \otimes P_{a(\kappa)} \Big).
    \end{equation}
    Similarly, the most efficient implementation of the target-control gate $\widetilde{CNOT}_{a i}$ is the sequence
    \begin{equation}
        \widetilde{CNOT}_{a i} = \prod_{\kappa=1}^w \Big( \id_i \otimes \ketbra{+1}{+1}_{a(\kappa)} + X_i \otimes \ketbra{-1}{-1}_{a(\kappa)} \Big),
    \end{equation}
    where we have introduced the $\pm 1$ eigenstates $\ket{\pm 1}_{a(\kappa)}$ of the 1-qubit Pauli operator $\overline{Z}_a(\kappa)$ for $\kappa = 1, \dots, w$.
    This expression coincides with Eq.~(\ref{eq:Hproperty}) for a transversal logical Hadamard gate which satisfies $\overline{H}_a \overline{X}_a \overline{H}_a = \overline{Z}_a$.
    \end{proof}

We stress that although the construction of the $\widetilde{CNOT}$ gate is similar to the transversal construction of the $\overline{CNOT}$ logical gate, the gate $\widetilde{CNOT}$ is not fault-tolerant, as an error on the noisy register can spread to many clean registers.
Sacrificing fault-tolerance is inevitable, given the need for at least one gate connecting clean and noisy registers.

Finally, we note that arbitrary and universal logical operations can be achieved within the framework presented in this section by introducing non-Clifford gates, such as the $T$-gate.
The least resource-demanding method allowed by our framework is the introduction of non-Clifford gates on the noisy registers of the circuit.
This is sufficient to promote Clifford gates to universality~\cite{bravyi2005universal}.
However, one can also implement non-Clifford gates directly on clean registers via one of several methods known in the quantum error correction community. 
A generally applicable method is magic state distillation~\cite{bravyi2005universal}, where a significant amount of noisy magic states are required to prepare certain magic states with low error.
This method can complement the Clifford transversality offered by 2-dimensional color codes to achieve universal quantum computation.
It is also well-suited for our framework, as the entire distillation process can be realized by transversal Clifford operations.
Other methods include code switching~\cite{pogorelov2024experimental,connor2014using} which circumvents the Eastin-Knill restriction~\cite{eastin2009restrictions} by choosing two codes with complementary sets of transversal gates; gauge fixing~\cite{paetznick2013universal} that allows transversal implementation of the Hadamard and the $CCZ$ gates which together are universal; and color codes of higher dimensions~\cite{bombin2015gauge}, which can be viewed as an application of gauge fixing.
These methods may generally admit a transversal gate set different to the Clifford group, hence our framework would need to be adapted to the new transversal gate set, which lies outside the scope of this paper.

\subsection*{Part B: Analytic bound on concentration scaling}
\label{sec:analytic-scaling}

We now present our analytic model used to glean insights on the partial error correction framework. First, we discuss an effective logical noise model, before presenting our lower bound on the convergence of random circuit outputs to uniform under this effective noise model. Finally, we discuss the possibility of this lower bound on convergence surpassing known upper bounds for fully noisy circuits, thus illuminating a possible scaling advantage of the partial error correction model over fully noisy circuits, given sufficiently many clean qubits.

\subsection*{Noise model for clean-noisy circuits}\label{sec:noise-model}

In order to fairly compare the clean-noisy setup with equivalent noisy circuits, we need to understand the effective action of physical noise on a partially error-corrected circuit. 
Specifically, our clean-noisy setup introduces a complex type of error propagation via clean-noisy couplings.
We now provide a heuristic approach towards modeling the noise induced by clean-noisy couplings and explore how they affect Pauli error propagation throughout the circuit.

We consider local Pauli noise channels, defined by their action on a single-qubit state $\rho$ as
\begin{align}\label{eq:paulinoise}
    &\PC_k(\rho; p_X, p_Y, p_Z) = \nonumber\\
    &\hspace{10pt} \left( 1 \hspace{5pt}-\hspace{-10pt} \sum_{Q \in \{X,Y,Z\}} \hspace{-5pt}p^{(k)}_Q \right) \rho \hspace{5pt}+\hspace{-10pt} \sum_{Q \in \{X,Y,Z\}} \hspace{-5pt}p^{(k)}_Q Q \rho Q,
\end{align}
where the subscript $k \in \{c,b,d\}$ will label one of three specific Pauli noise channels we consider in Fig.~\ref{fig:gate_types}. The action of the Pauli noise channel is to map non-identity Pauli operators as $Q \rightarrow (1-\varepsilon_Q^{(k)}) Q$, where $\varepsilon_Q^{(k)}= 2 p^{(k)}_{Q'} + 2 p^{(k)}_{Q''}$ for $\{Q,Q',Q''\} = \{X,Y,Z\}$.
We denote
\begin{equation}
    \varepsilon_k = 1- \min_{Q\in {\{X,Y,Z\}}}|1-\varepsilon^{(k)}_Q|,
\end{equation}
and hereon refer to this quantity as the (maximum scaled) error rate. This characterizes the rate at which non-identity Pauli operators are damped due to the noise, but does not take into account any phase factors they incur.
The depolarizing noise channel on a single-qubit quantum state $\rho$ corresponds to $\varepsilon^{(k)}_X = \varepsilon^{(k)}_Y = \varepsilon^{(k)}_Z = \varepsilon_k$, and we denote its action as
\begin{equation} \label{eq:depol_noise}
    \D_k(\rho) = (1-\varepsilon_k)\rho + \varepsilon_k\frac{\id}{2}.
\end{equation}
In SI Sec.~I (Supplementary Information Section I), we extend the scope of our proofs to arbitrary Pauli noise, which we define precisely. That may include correlations between two logical registers, as is common in the analysis of quantum error correction code thresholds~\cite{stephens2014fault}. However, for clarity of exposition, we focus on single-qubit Pauli noise in the main text.

We assume that physical gate noise is the dominant noise process on both clean and noisy registers, and that the effects of this noise can be characterized by instances of single-qubit Pauli noise that follow application of the ideal gate. It is simple to model noise for all-noisy systems, but considerations are not so obvious when one considers error correction on part of the circuit architecture. When modeling noisy two-qubit gates, we first assume that $CNOT$ gate error dominates the noise process, and, secondly, that the effective error rate of two-qubit gates scales proportionally to the effective error rate of a $CNOT$ gate. 
We note that arbitrary two-qubit gates can be compiled from at most 3 $CNOT$s and single-qubit gates~\cite{shende2004minimal}.

\begin{figure}[t]
    \includegraphics[width=0.9\linewidth]{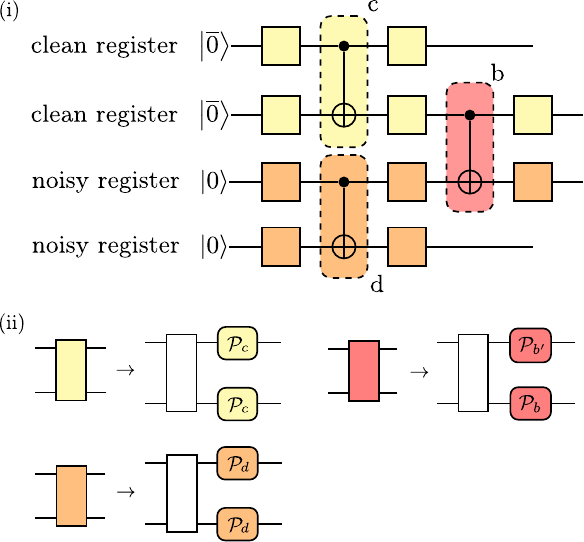}
    \caption{\textbf{Heuristic noise model for the clean-noisy setup.} 
     In the partial error correction model there are two types of logical single-qubit gates with corresponding effective noise models to consider: gates acting on noisy registers, and gates acting on error-corrected (``clean'') registers plotted in~\textbf{(i)} as orange and yellow squares, respectively.  Furthermore, we identify \textit{three} types of logical two-qubit gates which we label as ($d,c,b$) and show in \textbf{(i)} for the example of $CNOT$ gates. 
     These $CNOT$ gates along with single-qubit gates complete a universal gate set (Theorem~\ref{thm:cnotboundary}), and we assume them to be the dominant source of error.
     Type-$d$ $CNOT$s act on noisy registers (shaded orange in the plot) comprising 1 physical gate.
     Type-$c$ $CNOT$s (also written as $\overline{CNOT}$) act upon clean registers (shaded yellow) comprising $w$ physical gates implemented transversally according to the code.
     Finally, type-$b$ $CNOT$s (also written as $\widetilde{CNOT}$) on the boundary act on a noisy and a clean register (shaded red) comprising $w$ physical gates implemented according to Theorem~\ref{thm:cnotboundary}.
     In \textbf{(ii)}, we show schematically our effective gate error model where noisy two-qubit gates are modeled as ideal unitary gates followed by single-qubit Pauli noise of the form in Eq.~(\ref{eq:paulinoise}). The effective noise in our model acts on the logical (noisy or clean) registers, and not on the underlying physical registers comprising the clean qubits. In SI Sec.~I, we extend this noise model by considering arbitrary multi-qubit Pauli channels following the ideal logical gates instead of single-qubit Pauli channels.}
     \label{fig:gate_types}
\end{figure} 

In Fig.~\hyperref[{fig:gate_types}]{2(i)} we identify three different types of logical $CNOT$ gates, each of which has a different physical implementation. These lead to three types of two-qubit gates to consider, each with their own effective noise model, which we outline in Fig.~\hyperref[{fig:gate_types}]{2(ii)}. The first type of $CNOT$ gate we identify acts on two noisy qubits (labeled~$d$ in Fig.~\hyperref[{fig:gate_types}]{2(i)}). We characterize the noisy implementation of this gate with an error rate which reflects the underlying physical error rate. 
The second type of $CNOT$ gate acts only on clean qubits (labeled $c$) and are described by a reduced logical error rate compared to the type~$d$ gates, representing an imperfect QEC cycle following the gate implementation.
The third $CNOT$ gate acts on a boundary (labeled $b$), i.e.~coupling one clean qubit to one noisy qubit.
In this case, we model the effects of noise as local Pauli channels with augmented error rates, reflecting the fact that multiple physical $CNOT$s coupled to one noisy qubit are required for the implementation of the logical $CNOT$ gate acting on the boundary. 
We proceed to elaborate upon and justify these choices. 

The modeling of gates acting on noisy registers (colored orange in Fig.~\hyperref[{fig:gate_types}]{2(ii)}) is straightforward. Such a noisy gate is modeled as a sequence of the ideal gate followed by a local Pauli noise channel with an error rate of $\varepsilon_d$. We equate $\varepsilon_d$ with the effective base physical error rate of the device.

Gates acting only on clean registers (yellow in Fig.~\hyperref[{fig:gate_types}]{2(ii)}) are implemented according to the 
error-correcting properties of the code.
Given an $[[N,1,d]]$  QEC code, each clean register can be thought of as an encoding of $N$ underlying noisy physical registers. 
We describe local physical noise $\PC_d^{\otimes N}$ on the underlying registers by an effective logical noise channel $\PC_c$, which must have a lower error rate $\varepsilon_c < \varepsilon_d$ due to the QEC process. 
In the absence of unitary dynamics, i.e. when no logical gates are applied, the form of the idling channel $\PC_c$ was calculated in Ref.~\cite{rahn2002exact} taking into account the structure of stabilizer codes.
In particular, it is shown that physical Pauli noise is represented by effective logical Pauli noise which maps the code space back into the code space.
For example, following the discussion for the Steane code in Sec.~VI of Ref.~\cite{rahn2002exact}, given physical idling noise in the form of a Pauli channel with error rate $\varepsilon_d$, we find an effective logical error rate $\varepsilon_c = 42\varepsilon_d^2 +O(\varepsilon_d^3)$ achievable with perfect encoding and decoding.
In the presence of unitary dynamics, effective logical noise is decohered by the decoding process for stabilizer codes, tending to Pauli noise~\cite{beale2018quantum}.
Thus, we deem it reasonable to model imperfect two-qubit gates on clean registers as the ideal gate, followed by a Pauli noise channel of some reduced error rate $\varepsilon_c < \varepsilon_d$.

Two-qubit gates acting on the clean-noisy boundary i.e.~on a noisy and a clean register (red in Fig.~\hyperref[{fig:gate_types}]{2(ii)}) are implemented according to Theorem~\ref{thm:cnotboundary}.
The implementation of each $\widetilde{CNOT}$ gate takes $w$ physical gates according to Theorem~\ref{thm:cnotboundary}, each of which involves the noisy qubit. 
Therefore, calculating the effective logical error rate of the $\widetilde{CNOT}$ due to noisy physical $CNOT$ gates results in complications in this picture. This is due to non-trivial commutations between the actions of the physical $CNOT$ gates and the physical noise channels on the noisy register.
For this reason, we adopt a heuristic noise model for our analytic results by modeling a two-qubit gate acting on a boundary by the ideal gate followed by Pauli noise channels with fixed but different error rates on the logical clean and noisy registers.
We note that in Results Part C we directly simulate the effect of Pauli noise on partially error corrected circuits and our findings qualitatively agree with the analytic results that we derive using the heuristic noise model laid out in this section. 

In the heuristic model, effects of the noise on the clean and noisy registers are approximated by Pauli noise channel with new error rates $\varepsilon_b,\varepsilon_{b'}$, as shown schematically in Fig.~\hyperref[{fig:gate_types}]{2(ii)}. 
On the noisy register, since we have $w$ physical $CNOT$ gates acting on it, a first order approximation for the error rate $\varepsilon_b$ on the noisy register is that the effective Pauli noise channel is $w$ times more likely to produce an error than a single physical Pauli noise channel, and therefore we can think of $\varepsilon_b = w\varepsilon_d$. 
However, our following derivations need not make use of any specific expression for $\varepsilon_b$ in terms of $\varepsilon_d$.

We also set an effective error rate $\varepsilon_{b'}$ on the clean register after a clean-noisy coupling.
It is evident that $\varepsilon_{b'} < \varepsilon_b$ should hold, as the same gates are involved in both registers, but the clean register undergoes a QEC cycle.
We also expect that $\varepsilon_{b'} > \varepsilon_c$ because in a clean-clean coupling each underlying physical qubit comprising the clean registers is acted upon by a physical $CNOT$ exactly once during the implementation of a $\overline{CNOT}$. 
This is due to the transversal property of the gate.
On the other hand, in a clean-noisy coupling, the noisy register is used multiple times during the implementation of a $\widetilde{CNOT}$.

We further expect the stronger condition that $\varepsilon_{b'} > \varepsilon_{d}$, as the implementation of a $\widetilde{CNOT}$ requires $w \ge d$ gates, where $d$ denotes the code distance, and the number of correctable errors is $t = \lfloor (d-1)/2 \rfloor < w/2$.
Therefore, an error occurring on the noisy register due to one of the first $w/2$ physical $CNOT$ gates, will propagate to more than $t$ physical registers comprising the clean qubit, hence it is more likely that a logical error occurs compared to the noisy-noisy case.
For example, consider the implementation of $\widetilde{CNOT}$ with the repetition code in Fig.~\ref{fig:bit-flip}.
If an error occurs on the first (or second) physical $CNOT$, then the noisy qubit as well as the first (or second) physical register of the clean qubit will be flipped.
As the noisy qubit has flipped, this will likely also lead to a flip on the subsequent second and third (or third) physical registers of the clean qubit, resulting in a logical error of weight 2 (or 3).
On the contrary, if the error occurred on the third physical $CNOT$, then the logical error on the clean qubit is of weight 1 and thus detectable by the repetition code. 

We employ the above noise model in the analytic derivations which we present in the remainder of Results Part B.
In Results Part C, we then test our noise assumptions in our numerical simulations, corroborating our analytic results for a more realistic setting.
To this end, we investigate the performance of partial error correction with realistic quantum error correction using the Steane code. 
In particular, a logical gate is simulated as a sequence of noisy physical gates with physical Pauli noise. Furthermore, we take idling noise into account and study its effects on the fraction of clean qubits required for an advantage over an all-noisy device.

\subsection*{Random circuit model}
\label{sec:rcmodel}

Here, we study random brick-layered circuits using the heuristic noise model discussed thus far in order to understand how the number of clean qubits affects the distinguishability of circuit outputs with increasing gate depth. 

In this model, we consider a brick-layered circuit $C$ composed of alternating two-qubit blocks (see Fig.~\ref{fig:error_scramble_setup}). We model the explicit constructions of Sec. ``Clifford gates on clean-noisy qubit pairs'' by working at the level of logical operations, and adopting the model of effective noise as laid out in Sec. ``Noise model for clean-noisy circuits''. Specifically, as discussed, we consider an instance of local depolarizing noise to occur after each brick with different error rates assigned to each qubit type. Namely, we assign a low depolarizing probability $\varepsilon_{c}$ to the error-corrected (``clean'') qubits; we assign a high depolarizing probability $\varepsilon_{d}$ to the bulk of the fully noisy qubits; and finally, we assign higher depolarizing probabilities $\varepsilon_{b},\varepsilon_{b'}$ to the qubits at the boundary between clean and noisy registers. 

In this analysis, we focus on the trace distance of the output state of a circuit $\rho(C)$ to the maximally mixed state, which we denote as $T\big(\rho(C), {\id}/2^n\big)$. This is an important quantity to study due to its operational meaning as the measure of maximal distinguishability between the output state and the maximally mixed state, for any observable \cite{fuchs1999cryptographic}. The case of a fully noisy circuit with local depolarizing noise as in Eq.~(\ref{eq:depol_noise}), where $\varepsilon_{c} = \varepsilon_{b} = \varepsilon_{b'} = \varepsilon_{d}$, has been studied extensively in the past \cite{aharonov1996limitations, ben2013quantum, muller2016relative}. For instance, it is known that in this setting  
\begin{equation}\label{eq:upper-bound}
    T\big(\rho({C}), \frac{{\id}}{2^n}\big) \leq \sqrt{\ln 4 \cdot n}\;(1-\varepsilon_{d})^{L} \,,
\end{equation}
for all noisy circuits ${C}$ of $L$ layers (using the result of Ref.~\cite{muller2016relative} along with the quantum Pinsker's inequality~\cite{ohya2004quantum}), where a layer is defined as a set of unitary gates after which an instance of tensor-product local noise occurs. This implies that, in the fully noisy case, the output state is operationally exponentially hard to distinguish from white noise with increasing circuit depth. How much can this improve with a partially error-corrected setup?

We prove a lower bound on the trace distance of the output state $\rho(C)$ of the brick-layered circuit $C$ to the maximally mixed state, denoted $T\big(\rho(C), {\id}/2^n\big)$, averaged over the ensemble of circuits $C\sim\YC$ such that each local two-qubit unitary brick forms a 2-design. 
We recall that in practical terms, such arbitrary  gates can be obtained with at most 3 $CNOT$ gates~\cite{shende2004minimal}.
Since the Clifford group is a 2-design~\cite{divincenzo2002quantum}, the averaging can also be performed by sampling only over two-qubit Clifford gates.
We establish this bound by proving a stronger bound on the total variation distance $\delta\big(\rho(C), {\id}/{2^n}\big)$, which is proportional to the $\ell_1$ (taxicab) distance between the output probability distribution of $\rho(C)$ from the uniform distribution.

\begin{figure}[t]
\begin{center}
\includegraphics[width=0.95\linewidth]{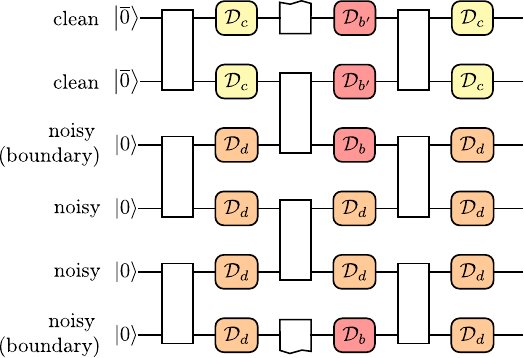}
\caption{\textbf{Random circuit model.} Our analytic result considers the average performance across an ensemble of logical brick-layered circuits composed of alternating layers of two-qubit gates on a one-dimensional  ring of qubits. The circuits consist of a collection of noisy logical registers and a collection of error-corrected (``clean'') logical registers, and we consider a single-qubit effective noise model with the four different error rates introduced in Fig.~\ref{fig:gate_types}. 
We note that in the example circuit depicted here, there is on average one clean-noisy coupling per layer, at each boundary between clean and noisy registers.} 
\label{fig:error_scramble_setup}
\end{center}
\end{figure}

\begin{theorem}\label{thm:lower-bound}
    Consider the brick-layered circuit with periodic boundary conditions in Fig.~\ref{fig:error_scramble_setup} on an even number of qubits $n$ and with an even number of layers $L$, under the local depolarizing noise model of Eq.~(\ref{eq:paulinoise}), characterized by three error rates $\varepsilon_c, \varepsilon_d, \varepsilon_b$, where we set $\varepsilon_{b'}=\varepsilon_{b}$ for simplicity.
    We consider the average behavior of circuits $C$ drawn over an ensemble of circuits $\YC$ such that each local two-qubit unitary forms a 2-design, which we denote as $\mathbb{E}_{C\sim \YC}[\cdot]$.
    The output state $\rho(C)$ satisfies 
    \begin{align}\label{eq:lower-bound}
        \underset{C\sim \YC}{\mathbb{E}} &\Big[T\big(\rho(C), \frac{\id}{2^n}\big)\Big] \geq \underset{C\sim \YC}{\mathbb{E}}\Big[\delta\big(\rho(C), \frac{\id}{2^n}\big)\Big] \geq \nonumber \\
        &\geq \frac{1}{12} \left(\frac{2}{5}\right)^L (1-\varepsilon_{c})^{ 2Lf_c} (1-\varepsilon_{d})^{ 2Lf_d} (1-\varepsilon_{b})^{2Lf_b}\,,
    \end{align}
    where $\YC$ is a distribution over brick-layered circuits composed of local 2-designs, and $f_c, f_d, f_b$ are fractions defined as
    \begin{align}
        f_c=\frac{n_c}{n}- \frac{n_b}{2n},\; f_d=\frac{n_d}{n}-\frac{n_b}{2n},\; f_b=\frac{n_b}{n}\,,
    \end{align}
    such that $f_c + f_d + f_b = 1$, where $n_c, n_d$ are the respective numbers of error-corrected qubits and noisy qubits satisfying $n_c+n_d=n$, and $n_b$ is the number of boundaries between the noisy and clean qubits.
\end{theorem}

Theorem \ref{thm:lower-bound} demonstrates that the concentration of circuit outputs on average is somewhat mitigated by the number of clean (``error-corrected'') qubits. However, this comes at a cost for each boundary between noisy and clean qubits, which each contribute a strong damping factor $(1-\varepsilon_b)^{2L}$ to the bound. We emphasize that this result holds for any configuration of noisy and clean qubits on the considered one-dimensional architecture.
In fact, we are able to prove a stronger result, where this lower bound on the total variation distance is obtained for any space and time distribution of local \textit{Pauli} noise in the circuit. The local Pauli noise model we consider includes the slight generalization with $\varepsilon_{b'}\neq\varepsilon_{b}$, as well as standard noise models such as single-qubit bit-flip noise, and correlated 2-qubit Pauli errors. This comes by extending the techniques of Ref.~\cite{aharonov2022polynomial} to circuit architectures that include local noise channels of different error rates throughout the circuit, and we present the details in SI Sec.~I.

We remark that noisy random circuits have been studied in many other contexts \cite{dalzell2022randomquantum, deshpande2022tight, aharonov2022polynomial} under local noise models with fixed error rates across the circuit. Theorem \ref{thm:lower-bound} demonstrates that the qualitative intuitions from these previous results hold also with variable error rates, which is that the decay of the total variation distance can be thought of as being characterized by the total number of error instances in the circuit. 

\subsection*{Threshold on clean qubits for scaling advantage}
\label{sec:logicalthreshold}

We remark that our lower bound in Eq.~(\ref{eq:lower-bound}) can surpass the all-noisy upper bound in Eq.~(\ref{eq:upper-bound}) for certain problem parameters. This indicates that our model of a partially error-corrected system can display a concrete advantage compared to a fully noisy system in distinguishing observables with increasing depth $L$ when the number of clean qubits $n_c$ lies above some threshold.

Concrete analysis of exactly what parameters the lower bound of Eq.~(\ref{eq:lower-bound}) surpasses the upper bound of Eq.~(\ref{eq:upper-bound}) may not be particularly insightful, as these bounds may not be tight. For instance, we anticipate that the factor of $(2/5)^L$ in our lower bound in Eq.~(\ref{eq:lower-bound}) should not be present in the true scaling, as when the circuit forms a global 2-design (achievable in depth exponential in $n$) the noise-free concentration is at worst exponential in $n$ \cite{mcclean2018barren}.

However, we recall Theorem \ref{thm:lower-bound} generalizes the result of Ref.~\cite{aharonov2022polynomial}, which is, to our knowledge, the tightest known lower bound for fully noisy circuits, and exactly reduces down to this result in the fully noisy regime ($n_c,n_b=0$). Assuming that these bounds contain key features of the scaling behavior with increasing $L$, we can observe when the lower bound Eq.~(\ref{eq:lower-bound}) in the clean-noisy setting surpasses the same lower bound in the fully noisy regime. 

Perhaps surprisingly, this transition does not come trivially with one clean qubit -- this is due to the coupling of the error-corrected qubit to a noisy qubit having a high effective error rate. 
Concretely, we require
\begin{equation}
    \frac{(1-\varepsilon_{c})^{2Lf_c}(1-\varepsilon_{d})^{2Lf_d} (1-\varepsilon_{b})^{2Lf_b}}{(1-\varepsilon_{d})^{2L}} \geq 1 \,,
\end{equation}
which is achieved when
\begin{equation}
    n_c \geq n_b\, \frac{\tfrac{1}{2}\log (1-\varepsilon_{c}) + \tfrac{1}{2}\log (1-\varepsilon_{d}) - \log (1-\varepsilon_{b})}{\log (1-\varepsilon_{c}) - \log (1-\varepsilon_{d})} \,. \label{eq:threshold}
\end{equation}

We have thus obtained a lower bound which shows that the partial error-corrected model strictly outperforms the fully-noisy model when the number of clean qubits lies above a threshold independent of the qubit number $n$ --- instead, it depends on the number of clean-noisy boundaries, and the different error parameters in the circuit.
Our numerical results in Results Part C are consistent with the exponential decay of output state fidelity with circuit depth $L$ as stated in Theorem~\ref{thm:lower-bound}, leading to a threshold number of clean qubits which is also independent of $n$ in the absence of idling noise. 

\subsection*{Part C: Numerical investigation}\label{sec:numerics}
From here on, we numerically study the noise properties of the  $\widetilde{CNOT}$ gate and observe the existence of a clean qubit threshold for brick-layered Clifford circuits. We use a depolarizing noise model which is implemented at the circuit level and motivated by trapped-ion quantum computers.
We first demonstrate that the $\widetilde{CNOT}$ gate coupling the clean and noisy registers has a higher associated error rate than a $CNOT$ gate between two noisy registers.
We then present clear evidence of a clean logical qubit threshold at which advantage is obtained over fully noisy circuits, as alluded to by our analytic results. Furthermore, a detailed study of the impact of idling noise on the value of this threshold and an implementation under realistic noise rates are studied.

\subsection*{Randomized benchmarking} \label{sec:lrb}

Here, we numerically study fidelity decay rates for sequences of random two-qubit Clifford gates acting on an error-corrected and a noisy qubit, comparing them with analogous sequences for two error-corrected and two noisy qubits. Such sequences are employed by randomized benchmarking~\cite{knill2008randomized, magesan2011scalable, combes2017logical} to estimate logical error rates of gates. We encode error-corrected qubits using the Steane code.  

Randomized benchmarking (RB) considers random sequences of Clifford gates that are compilations of identity, acting on an initial state. The protocol estimates the average error rate per element of the Clifford group by calculating the fidelity between the final and the initial states.
Here, we consider the two-qubit Clifford group. 
As we are interested in the regime of pre-fault-tolerant quantum computation, we use post-selection on syndrome outcomes to mitigate the effects of errors. We note that post-selection on syndrome outcomes is utilized by the current hardware implementations of the quantum error correction and detection codes~\cite{bluvstein2023logical,mayer2024benchmarking}. Since the randomized benchmarking does not account for post-selection, we use this setup to investigate fidelity decay in the presence of noisy and error-corrected qubits rather than to generalize a notion of the error rate per Clifford to our case.

We consider sequences of $m$ random elements of the 2-qubit Clifford group followed by a 2-qubit Clifford unitary chosen that the whole sequence compiles to the identity. Each sequence is run 4000 times, and eight sequences are used for each $m$. 
In the case of two noisy qubits, the Clifford unitaries are compiled to a gate set $G \coloneqq \{H, S, S^{\dag}, X, Y, Z, CNOT\}$ and applied to an initial state $\ket{00}$. Analogously, for two error-corrected qubits we use gate set $\overline{G} \coloneqq \{ \overline{g}\, |\, g \in G \}$ and state $\ket{\overline{00}}$, and for a pair of error-corrected and noisy qubits we use gate set $G \cup \overline{G} \cup \{\widetilde{CNOT}\}$ and state $\ket{\overline{0}0}$. 
Therefore, in all three cases, we compile a Clifford group element to the same sequence of logical Clifford gates that differ only by their hardware implementation, depending on whether they act on error-corrected or noisy qubits.

We assume that physical single-qubit gates are affected by single-qubit depolarizing noise with an error rate of $\varepsilon_1=10^{-5}$ and physical $CNOT$s by a tensor-product of single-qubit depolarizing channels, each with an error rate of
$\varepsilon_{CNOT}=10^{-4}$.

We  choose $\varepsilon_{CNOT}$ to obtain a clear advantage of the fully error-corrected setup over the fully noisy one, and assume all-to-all device connectivity, as is the case in trapped-ion quantum computers. The two-qubit gate noise captures the dominant errors occurring in the entangling gates of Quantinuum's H2 quantum computers~\cite{haghshenas2026digital}. Below, in a quantum error detection setup for random mirrored circuits, we additionally show an advantage  for our framework with realistic $\varepsilon_{CNOT}$ error rates and $CNOT$ noise dominated by two-qubit Pauli errors. Furthermore, we model state preparation (measurement) noise by the single-qubit depolarizing noise with the error rate $\varepsilon_{CNOT}$ acting at all the prepared (measured) qubits after the state preparation (before the measurement). The $\varepsilon_{CNOT}/\varepsilon_1$ ratio is chosen to mimic the error rate ratios in current quantum devices~\cite{cincio2021machine}. We neglect idling noise in this section for the sake of clarity of exposition, but we consider it in our large-scale simulations of the subsequent sections.

The logical gates involving error-corrected qubits are compiled to the physical gates as discussed in SI Sec.~II. After an action of a logical gate on error-corrected qubits, we apply QEC rounds to those qubits. We implement those rounds by adapting a Quantinuum implementation of the Steane code~\cite{ryan2021realization} to our setup.  In particular, we post-select based on the syndrome outcomes to reject runs likely affected by irrecoverable errors.  Furthermore, for simplicity, we implement syndrome decoding targeting only single-qubit $CNOT$ errors, which are more frequent than two-qubit gate errors in the case of our noise model and Quantinuums's  H2 devices~\cite{haghshenas2026digital}.   We detail our error correction algorithm in SI Sec.~II.

 \begin{figure}[t]
    \centering 
    \includegraphics[width=0.99\columnwidth]{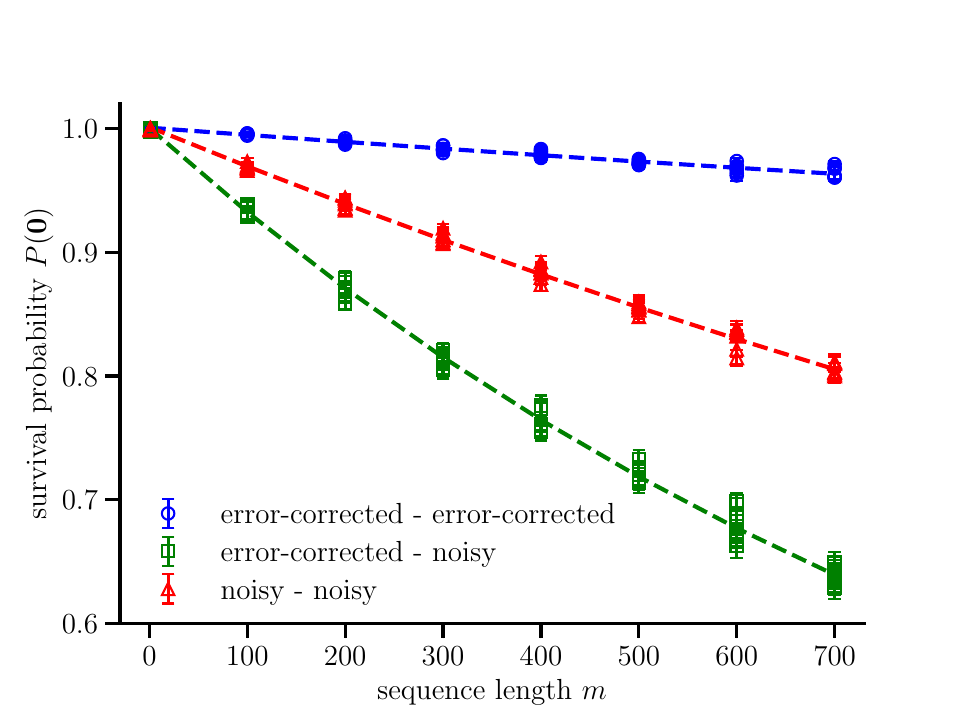}
    \caption{ {\bf Randomized benchmarking of two-qubit gates for the partial error correction framework.} Survival probability $P(\textbf{0})$ of a sequence of random two-qubit logical Cliffords for two error-corrected qubits (circles), two noisy qubits (triangles), and an error-corrected and a noisy qubit (squares). Each sequence is chosen to compile to the identity, so in the noiseless case, we have $P(\textbf{0})=1$. We simulate QEC using the Steane code. Logical Clifford gates are compiled to a native physical noisy gate set affected by depolarizing noise with the error rates $10^{-5}$ and $10^{-4}$ for single- and two-qubit gates, respectively. For each sequence length, we show results for eight different sequences run $4000$ times each. For each sequence, we plot the error bar computed as standard deviation of the mean.  Dashed lines are the best fits for the exponential decay described by Eq.~(\ref{eq:decay}).}
    \label{fig:RB}
\end{figure}

For the fully error-corrected,  partially error-corrected, and fully noisy setups, we use maximal sequence lengths of $m =  700$, 700, and 1400, respectively. This choice of $m$ results in a probability of rejecting a run during post-selection which is smaller than $0.26$, as shown in Fig.~7 of SI Sec.~III. For each of the three setups, we fit the decay of the survival probability  $P(\vec{0})$, i.e.~ the expected value of the $\dya{\psi_i}$ projector (with $\psi_i=\ket{00}$, $\ket{\overline{0}\overline{0}}$, or $\ket{\overline{0}0}$ depending on the qubit types) after applying the sequence of Cliffords, with sequence length $m$ to 
\begin{equation}
P(\vec{0}) = A c^{m} + B\,.
\label{eq:decay}
\end{equation}
Here  $A, B, c$ are the fit parameters, and we compute $P(\vec{0})$ as a ratio of the accepted runs with an all-zero logical measurement outcome to all the accepted runs. In the case of two noisy qubits, for which no post-selection was employed, Eq.~(\ref{eq:decay}) is exact in the limit of sequences of random Cliffords with gate and time-independent noise~\cite{magesan2011scalable}. In this case, $c$ is related to the average error per $n$-qubit Clifford $r$ by $$r = \frac{2^n - 1}{2^n} (1 - c)$$ and $A, B$ account for state-preparation and measurement errors. Nevertheless, the fitted formula approximates the results well in all cases, as shown in Fig.~\ref{fig:RB}. For plot clarity, in Fig.~\ref{fig:RB}, we only show $m\le700$ data. We present our complete RB data in SI Sec.~III. For two error-corrected qubits, we see a slow decay in $P(\textbf{0})$ that indicates the physical noise strength is low enough for our implementation of QEC to correct the vast majority of errors. This behavior contrasts with a much faster decay for an error-corrected qubit coupled to a noisy qubit whose decay rate is even larger than for two noisy qubits. As the Clifford compilation to logical Clifford basis gate set is the same for the partially error-corrected  and fully noisy settings, and $CNOT$ errors are dominant, this strongly indicates that $\widetilde{CNOT}$ has a higher error rate than a $CNOT$ on all-noisy qubits, confirming the heuristic analysis of Sec. ``Noise model for clean-noisy circuits''. 

We compute the parameter $c$ for noisy-noisy and partially error-corrected couplings obtaining $1-c_{d} =  4.3(3)\times10^{-4}$ and $1-c_{b} =  9.8 (5)\times 10^{-4}$, respectively. The ratio of these estimates suggests that the partially error-corrected Clifford gates are $2.3$ times more prone to error than the noisy Cliffords.

\subsection*{Random Clifford circuits: setup}
\label{sec:Cliff_setup}

Thus far, we have focused on the effective error rates of the three different types of two-qubit gates considered in our partial error correction framework, observing that coupling clean and noisy qubits using $\widetilde {CNOT}$ introduces more noise than coupling two noisy qubits. We now increase the number of logical qubits and observe an advantage of partial error correction over all-noisy computation when a clean qubit threshold is passed. 
We perform benchmarking of random Clifford circuits of increasing depth and size, explore whether trading logical qubits for a larger number of physical qubits can be advantageous. Furthermore, we investigate the performance of our framework at realistic two-qubit gate error rates using quantum error detection.

We simulate random mirrored Clifford circuits which compile to identity
\begin{equation}
U = VV^{\dag}\,.
\label{eq:mirr_Cliff}
\end{equation}
We arrange logical qubits into a line and build $V$ from $L$-layers of alternating nearest-neighbor $CNOT$s decorated by randomly chosen single-qubit Clifford gates, as exemplified in Fig.~\ref{fig:random_Cliff}. We note that this corresponds to $4L$ layers of logical gates in $V$. The single-qubit gates are chosen from a gate set $\{H, X, Y, Z, S, S^{\dag} \}$. Furthermore, we choose the first $n_c$ registers to be clean while the remaining $n_d$ are noisy. Thus, a block of clean qubits is connected to the noisy ones by $2L\,\,CNOT$ gates,  as depicted in Fig.~\ref{fig:random_Cliff}. The initial state is $\psi_i = \ket{\overline{0}}^{\otimes n_c} \ket{0}^{\otimes n_d}$. 

\begin{figure}[t]
    \includegraphics[width=0.99\columnwidth]{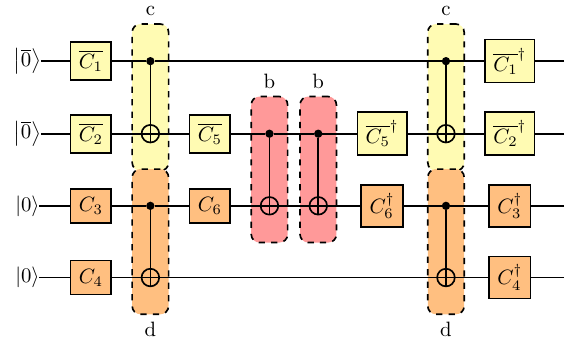}
    \caption{\textbf{Mirrored random Clifford circuits.} 
    A random mirrored Clifford circuit of the form of Eq.~(\ref{eq:mirr_Cliff}). Here $V$ is a layer of brick-like arranged nearest-neighbor $CNOT$, $\widetilde{CNOT}$, and  $\overline{CNOT}$  preceded by single-qubit Clifford gates $\overline{C_1},\overline{C_2},C_3,C_4,\overline{C_5},C_6$. Hence, the layer is built of $4$ layers of logical gates. We assume a block of clean registers (the first two) connected to a block of noisy ones by two $\widetilde{CNOT}$s. Single-qubit unitaries $C$, $\overline{C}$ are chosen randomly from a set $\{H,X,Y,Z,S,S^{\dag}\}$. }
    \label{fig:random_Cliff}
\end{figure}

As earlier, we simulate the circuits using an implementation of the Steane code, detailed in SI Sec.~II. Similar to the randomized benchmarking setup, we estimate the fidelity of the noisy and the noiseless state using the survival probability $P(\vec{0})$, which is defined here as an expected value of the $\dya{\psi_i}$ projector for the runs accepted by the post-selection on syndrome outcomes.  Also similar to  our randomized benchmarking simulations, we assume a physical gate set $\{H, S, S^{\dag}, X, Y, Z,$ $CNOT\}$ and a single-qubit depolarizing noise model with an error rate $\varepsilon_1=1.5\cdot10^{-5}$ for the single-qubit gates. We also use the single-qubit depolarizing noise with error rates $\varepsilon_{2}=\varepsilon_{CNOT}=\varepsilon_{\rm SPAM}=1.5\cdot10^{-4}$ for  $CNOT$ gate, $\ket{0}$ preparation, and a qubit measurement.
As earlier, the ratio $\varepsilon_2/\varepsilon_1$ is chosen to resemble this ratio for the current devices while the value of $\varepsilon_2$ and and the error Pauli weights are taken to ensure good error suppression by our QEC implementation. 
More precisely, we assume that the $CNOT$ is followed by single-qubit depolarizing channels, with an error rate $\varepsilon_{CNOT}$, acting on the same qubits as the gate.
Additionally, we follow (precede) a $\ket{0}$ preparation (qubit measurement) with the same noise channel applied to this qubit. For the sake of notational clarity, we denote the  state preparation and measurement noise error rate by $\varepsilon_{\rm SPAM}$.  Moreover, we also take into account idling of physical qubits during a physical gate execution, the state preparation, and qubit measurements. We model an idling qubit by an identity gate followed by the single-qubit depolarizing noise with an error rate range $\varepsilon_I=0-3\cdot10^{-6}$, in order to investigate the effect of increasing idling noise, as discussed in detail below.  

As mentioned above, we note that real-world $CNOT$ gates involve two-qubit correlated Pauli errors. 
Nevertheless, as argued  above, suppression of single-qubit errors is sufficient for seeking  an advantage with partially-fault tolerant QEC implementations targeting current state-of-the-art trapped-ion quantum devices~\cite{chertkov2025error}. 
Furthermore, to investigate the framework's performance scaling with $n$, we assume $\varepsilon_{CNOT}$ which is an order of magnitude smaller than in current trapped-ion devices~\cite{moses2023race}, \cite{ransford2025Helios}. 
Similarly, the single-qubit gates and idling error ratios are smaller than in the current devices, though $\varepsilon_{I}$ analyzed here have been obtained with a trapped-ion qubit~\cite{sepiol2019probing}. 
We demonstrate an advantage of our approach over the all-noisy setup for higher, realistic $CNOT$ error rates  below, in the context of quantum error detection.
Furthermore, for that use case, we model $CNOT$ noise by two-qubit depolarizing noise which consists primarily of two-qubit Pauli errors.
In this work, for simplicity, we assume throughout that the idling error is the same for each idling qubit and for each circuit layer. 
We compile circuits to layers for all-to-all device connectivity, as shown in SI Sec.~II. 
In practice, the idling error will vary layer to layer depending on the gate types, the presence of measurements, and the qubit transport time required for all-to-all connectivity. 

\subsection*{Random Clifford circuits: results for quantum error correction  }

We estimate the fidelity of the noisy and the noiseless state for $n=n_c+n_d=12-28$, $L=5-25$, and multiple choices of $0\le n_c\le n$. Numerical simulations are performed with a stabilizer simulator~\cite{gottesman1998heisenberg,aaronson2004improved} which samples the physical gate, measurement, and state preparation errors according to the error rates ($\varepsilon_1, \varepsilon_{CNOT}, \varepsilon_I$). In our case of the single-qubit depolarizing noise the errors are single-qubit Pauli gates $X,Y,Z$.  For each $n$, $n_c$, and $L$, we compute the average fidelity as the mean of the survival probabilities of $2.5\cdot10^4-2\cdot10^5$ runs of the random circuits with a single run per circuit.

\begin{figure}[t]
    \centering 
    \includegraphics[width=0.99\columnwidth]{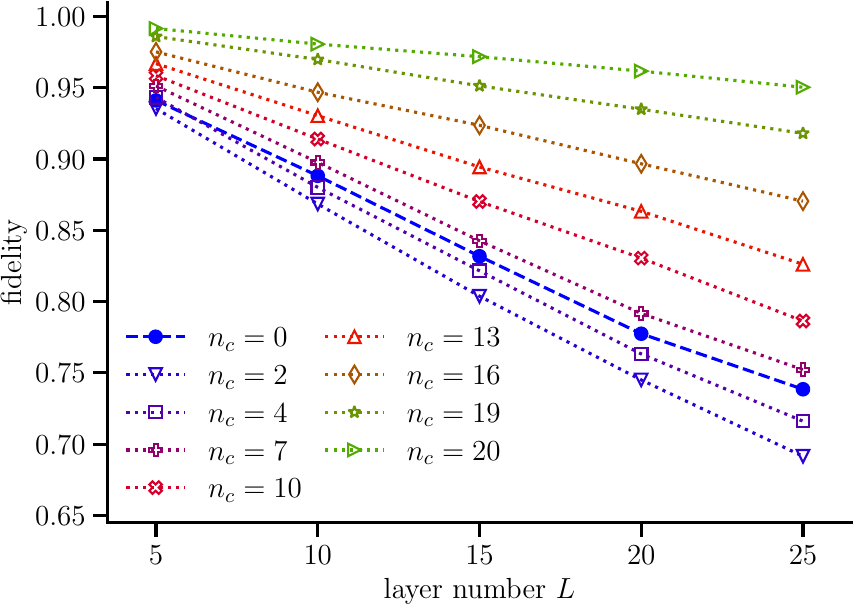}
    \caption{ {\bf Random Clifford circuits with a depolarizing noise.}   Average fidelities of the noisy and noiseless states prepared by random mirrored Clifford circuits with $n_c$ clean qubits, and $n=20$ qubits in total, are plotted versus the number of layers $L$. We use the Steane code to encode the clean qubits.  The circuit structure is shown schematically in Fig.~\ref{fig:random_Cliff}. We simulate the setup with a single-qubit depolarizing noise with the error rates $\varepsilon_1=1.5\cdot10^{-5}$,  $\varepsilon_{CNOT}=\varepsilon_{\rm SPAM}=1.5\cdot10^{-4}$, and $\varepsilon_I=0.75\cdot10^{-6}$,  for the physical single-qubit gates, $CNOT$, state preparation and measurement, and qubit idling, respectively. The fidelity is estimated as a mean of the survival probabilities of post-selected circuit runs with $2.5\cdot10^4-5\cdot10^5$ total runs  per data point, and a single run per random circuit. The statistical uncertainties of the data points given by a standard  deviation of the mean are smaller than the marker sizes. Note that the $y$ scale is logarithmic, indicating exponential decay of the fidelities.}
    \label{fig:mirr_Lplt}
\end{figure} 

 In Fig.~\ref{fig:mirr_Lplt} we show the results for $n=20$ and $\varepsilon_I=0.75\cdot10^{-6}$ versus $L$ for several $n_c$. Additionally, the post-selection probabilities for this setup are shown in Fig.~9 of SI Sec.~III, and are larger than $0.7$ in all cases.  For all ansatz depths $L$, we observe that the fidelity for the smallest non-zero clean qubit number $n_c=2,4$ is lower than that at $n_c=0$. At the same time, it systematically improves  with an increasing $n_c$. For $n_c=7$ it is already better than for the all-the-noisy ($n_c=0$) setup. Therefore, the results clearly demonstrate a threshold behavior  that reflects our analytic result in Theorem \ref{thm:lower-bound}. Hence, we introduce a threshold clean qubit number $n_{\rm threshold}$, such that for $n_c > n_{\rm threshold}$ the fidelity is better than for $n_c=0$. In Fig.~8 of SI Sec.~III, we show that the fidelity behaves qualitatively the same for $n=20$ and no idling noise. In that case we find that $2 < n_{\rm threshold} < 4$, which is smaller than for a setup of Fig.~\ref{fig:mirr_Lplt} with the idling noise. This comparison indicates that the idling noise affects $n_{\rm threshold}$ in a detrimental way. 

 \begin{figure}[t]
    \centering 
    \includegraphics[width=0.99\columnwidth]{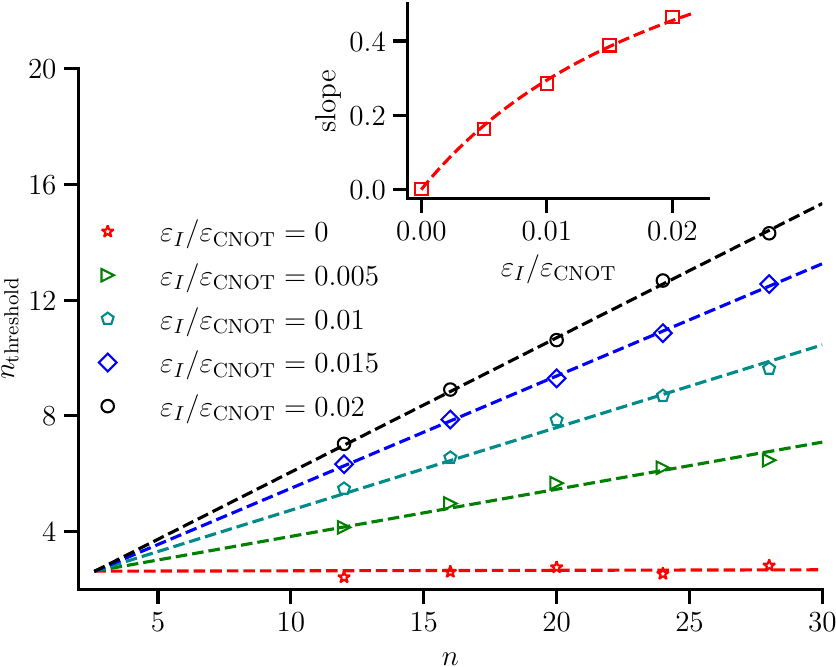}
    \caption{{\bf The clean qubit number required for an advantage over an all-noisy device.} Here we plot $n_{\rm threshold}$ for the random mirrored Clifford circuits with $n=12-28$ and several ratios of the noisy qubit idling and $CNOT$ gate error rates $\varepsilon_I/\varepsilon_{CNOT} =0,0.005,0.01,0.015,0.02$. The results are fitted by a linear ansatz  of Eq.~(\ref{eq:nt_pheno}) (dashed lines) with the slope dependent on the idling noise strength, unlike the intercept.  In the inset, we show the dependence of the slope on $\varepsilon_I/\varepsilon_{CNOT}$ fitted by a heuristic phenomenological ansatz of Eq.~(\ref{eq:a_pheno}) (dashed line).}
    \label{fig:mirr_threshold}
\end{figure}

For all $n$ and $\varepsilon_I$ we find that there exist $n_{\rm threshold}>2$, which does not depend on $L$ for fixed total qubit number and the idling noise error rate. We analyze its scaling with $n$ and $\varepsilon_I$ in Fig.~\ref{fig:mirr_threshold}.  We find that without idling noise $n_{\rm threshold}\approx3$ and does not depend on $n$. Otherwise, it grows approximately linearly with $n$, and the growth slope increases with the idling noise strength. This dependence is expected as the logical gates acting on the clean registers have higher depth when compiled to the native physical gate set than gates acting only on noisy registers. Furthermore, they are followed by QEC rounds that further enlarge the depth of the circuits applied to the clean qubits. Therefore, a random Clifford circuit compiled to the native gates in the presence of the clean qubits has an increased depth resulting in more idling at the noisy registers than in the case of an all-noisy device. 
Still, even in the presence of idling, we find that replacing a fraction of noisy qubits with error-corrected qubits leads to an advantage over the all-noisy device. For example, for the largest simulated $n=28$, we find $n_{\rm threshold}=7$ for $\varepsilon_I/\varepsilon_{CNOT}=0.005$, and $n_{\rm threshold}=15$ for $\varepsilon_I/\varepsilon_{CNOT}=0.02$.

To intuitively explain the  behavior of $n_{\rm threshold}$, we propose a heuristic model. We assume that the fidelity is proportional to the average number of errors occurring during a circuit execution. As in the case of the heuristic noise model of Sec. ``Noise model for clean-noisy circuits'', we assume that the errors occur primarily at the noisy qubits such that we can neglect the contribution from the clean registers. Furthermore, we assume that the dominant sources of errors are type-$d$ and $b$ two-qubit gates (similar to Sec. ``Noise model for clean-noisy circuits'' and the idling noise of the noisy registers. 
Finally, we assume that for fixed $n$, the number of type-$d$ gates including the idling gate is proportional to $n_d$, while the number of type-$b$ gates does not depend on $n_d$. Such scaling is expected for large $n$ with brick-layered circuits and a single block of clean registers, as is simulated here.

It follows that for fixed circuit depth, the average number of errors due to $d$-gates and $b$-gates are $n_d \varepsilon^{\rm eff}_d$ and $\varepsilon^{\rm eff}_b$, respectively, where $\varepsilon^{\rm eff}_d, \varepsilon^{\rm eff}_b$ are the effective type-$d$ and type-$b$ error rates, and do not depend on $n$. In our case, idling occurs at the noisy registers mostly when a layer of logical gates is executed in the presence of clean qubits, as the depth of the error correction rounds compiled to the physical gates (see SI Sec.~II) is an order of magnitude higher than the depth of a layer $d$-type gates.  Since in such a case the number of idling gates is approximately proportional to $n_d$, we assume that for the fixed circuit depth, the number of errors occurring during idling is $\varepsilon^{\rm eff}_I n_d$ with $\varepsilon^{\rm eff}_I$ independent of $n_d$. Setting the number of errors for $n_c=n_{\rm threshold}$ to be the same as for the noisy device, we obtain
\begin{equation}
  n \varepsilon^{\rm eff}_d = \varepsilon^{\rm eff}_b + (n-n_{\rm threshold}) (\varepsilon^{\rm eff}_d+ \varepsilon^{\rm eff}_I)\,, 
  \end{equation}
and 
\begin{align}
n_{\rm threshold} = b+&a(n-b)\,, \nonumber \\ a =\frac{\varepsilon^{\rm eff}_I/\varepsilon^{\rm eff}_d}{\varepsilon^{\rm eff}_I/\varepsilon^{\rm eff}_d+1}& ,\quad b = \frac{\varepsilon^{\rm eff}_b}{\varepsilon^{\rm eff}_d}\,.
\label{eq:nt_pheno}
\end{align}
We expect that $\varepsilon^{\rm eff}_d, \varepsilon^{\rm eff}_I$ are approximately proportional to the noisy $CNOT$ and idling qubit error rates $\varepsilon_{CNOT} , \varepsilon_I$, respectively. Therefore, we expect
\begin{equation}
    a=\frac{c \varepsilon_{I}/\varepsilon_ {CNOT} }{c \varepsilon_{I}/\varepsilon_ {CNOT}  +1}\,, 
\label{eq:a_pheno}
\end{equation}
where  $c$ is some constant. Additionally, for zero idling noise, we expect
\begin{equation}
    n_{\rm threshold} = b\,.
    \label{eq:nt_Idle0}
\end{equation}

We apply these heuristic predictions to our numerical results as shown in   Fig.~\ref{fig:mirr_threshold}. First, for each $\varepsilon_I/\varepsilon_ {CNOT} $ we fit  $n_{\rm threshold}(n)$  with the linear ansatz in Eq.~(\ref{eq:nt_pheno}) taking $a$ as the fit parameter and estimating $b$ by fitting the no idling noise case with the constant ansatz~(\ref{eq:nt_Idle0}). We obtain good agreement with the numerical results. Secondly, we analyze the behavior of $a$ versus $\varepsilon_{I}/\varepsilon_ {CNOT} $ in more detail, by
fitting it with Eq.~(\ref{eq:a_pheno}). Again, we obtain good agreement with the numerical results (see Fig.~\ref{fig:mirr_threshold}). 
The fitted $a$ values indicate that in the limit of a large system, $\varepsilon_I/\varepsilon_ {CNOT} =0.005$ ($\varepsilon_I/\varepsilon_{CNOT} =0.02$)  and $n_c/n =0.17$ ($n_c/n =0.47$) results in an advantage over an all-noisy machine.  These results demonstrate that replacing a fraction of the noisy registers with the error-corrected ones enhances the power of a quantum device.

\begin{figure}[t]
    \centering \includegraphics[width=0.999\columnwidth]{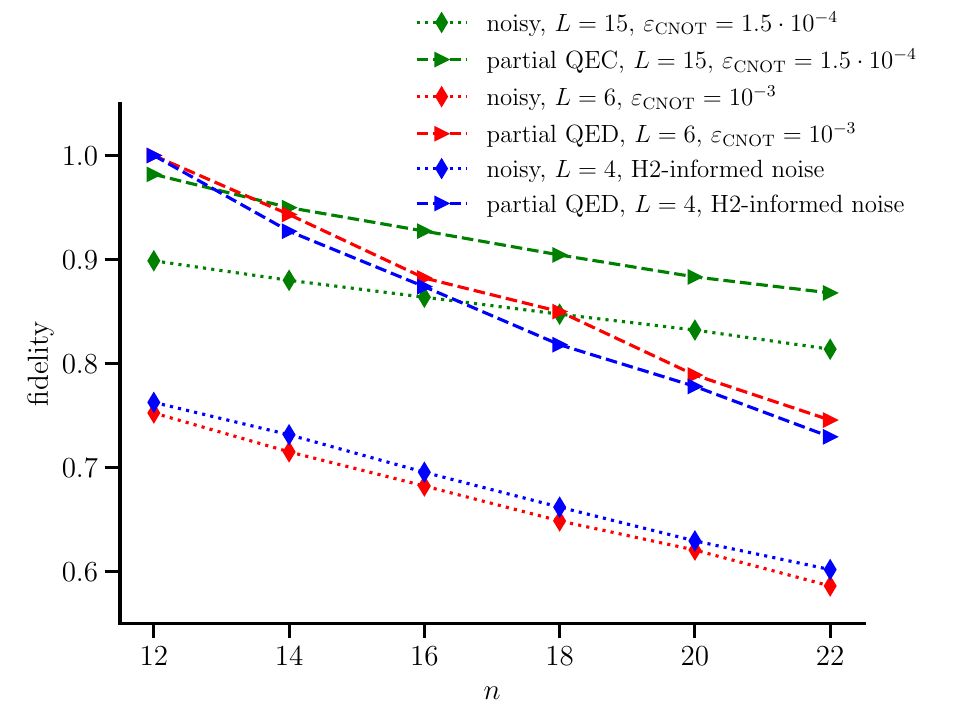}
    \caption{ {\bf Fidelity of partially error-corrected, partially error-detected, and noisy random-mirrored Clifford circuits}. The green dashed line indicates the fidelity for  partially error-corrected  $L=15$ circuits with $n_c=12$ clean qubits and $n\in\{12,14,16,18,20,22\}$. For reference, the the green dotted line indicates the fidelity for the noisy case ($n_c=0$) and the same circuits. The blue (red) dashed lines show the results of random-mirrored $L=6$ ($L=4$) Clifford circuits obtained by partial quantum error detection (QED) with $n_c=12$ clean qubits encoded using the Steane code. In this case, syndromes are extracted during the final measurement and the results are post-selected based on the syndromes (see details in the main text).  For  reference, the blue (red) dotted lines represent noisy results for the $L=6$ ($L=4$) circuits. For the $L=15$ circuits, the noise model and the clean qubit implementation, are the same as in Fig.~\ref{fig:mirr_Lplt}. For $L=6$, we use a single-qubit depolarizing noise with $\varepsilon_{CNOT}=10^{-3}$ and $\varepsilon_{I}/\varepsilon_{CNOT}=0.1$, and for $L=4$ a noise model informed by Quantinuum's H2 device (see SI Sec.~IV). The fidelity is estimated as in Fig.~\ref{fig:mirr_Lplt}, with $3.1\cdot 10^4-2\cdot10^5$ ($2.6\cdot10^4$) circuit executions per data point for $L=15$ ($L=6$). The post-selection probabilities are given in Tabs.~III,~IV of SI Sec.~III, and are $8-11\%$ in the error detection case.
    The statistical uncertainties of the data points measured by a standard  deviation of the mean are smaller than the marker sizes. We note that  partial QEC  (QED) used here requires $130$ ($94$) physical qubits, which are not sufficient for $n>13$ when only encoded qubits are used.  }
    \label{fig:fidelity}
\end{figure}

In practice, the observed advantage can be exploited to improve upon the effective error rates of an all-noisy qubit device, for systems which would be too big to simulate using only error-corrected qubits. 
To exemplify this, we consider as a thought experiment a quantum device with $130$ physical qubits. As we need 7 data qubits and 3 ancilla qubits for a single error-corrected register (see SI Sec.~II), we can have at most $n_c=13$ clean qubits. 
At the same time, we can extend the accessible $n$ by designating a subset of physical qubits to be noisy qubits. 
For example, with $n_d=10$ we can have $n_c=12$ and $n=22$. To make this extension desirable, the partial QEC framework needs to improve upon a noisy device with the same $n$. 
In Fig.~\ref{fig:fidelity}, we illustrate this improvement for  $\varepsilon_I/\varepsilon_ {CNOT} =0.005$,  $L=15$, and $n\in\{12,14,16,18,20,22\}$. 
We  compare the fidelity of the random mirrored circuits with $n_c=12$ qubits to the case of $n_c=0$ for a given $n$. 
We find that, while in both cases the fidelity decreases with increasing $n$, the fidelities of the partially error-corrected system improves substantially over the noisy case. 
In particular, for $n=14$ the infidelity (defined as $1-{\rm fidelity}$) is reduced by a factor of $2.4$, and for $n=22$ by a factor of $1.4$. 
We note that a reduction of circuit infidelity typically  enables better quality quantum error mitigation~\cite{cai2022quantum}, which has been proposed as a necessary ingredient for early QEC implementations~\cite{suzuki2022quantum}. 
Therefore, while the quality of the results can be further improved in both cases using error mitigation methods, the improved accuracy of the partially-corrected results remains advantageous, even when allowing for error mitigation.

Finally, in SI Sec.~V we provide additional results on the $n_{\rm threshold}$ scaling with $n$ for a more general non-Pauli noise model which is handled at scale by simulation with tensor network methods with idealized QEC round treatment, see details in SI Sec.~V.  Those results show a systematic improvement of the partial QEC performance with an increasing $n_c$ that results in an advantage over the noisy device for $n_c>n_{\rm threshold}$. Additionally, the threshold clean qubit number is found to scale  in agreement with the heuristic model introduced here.

\subsection*{Random Clifford circuits: results for quantum error detection (QED)} \label{sec:QED}

In the simulations described above, we assume $\varepsilon_{CNOT} \in \{10^{-4},1.5\cdot10^{-4}\}$, which is an order of magnitude smaller than the current state of the art~\cite{ransford2025Helios}. Such a small error rate is required to ensure an  error reduction by  Steane code error correction.  For higher error rates, partially fault tolerant QEC implementations and quantum error detection codes have been proposed as feasible paths to quantum advantage.   In particular, QED has been demonstrated to suit well capabilities of current  hardware~\cite{bluvstein2023logical, chertkov2025error}. To illustrate the usefulness of our scheme in the near term, i.e. for $\varepsilon_{CNOT} \approx 10^{-3} $  we consider  a Steane code implementation as a QED code.

We simulate the mirrored random Clifford circuits, described in Eq.~(\ref{eq:mirr_Cliff}), and encode the clean qubits and the clean-clean and clean-noisy gates with the Steane code, the same as for the QEC simulations. We do not follow the gates with rounds of the syndrome extractions, as for the QEC use case. Instead, we perform the syndrome extraction only during the final measurement, as described in SI Sec.~II  for the QEC scenario. Unlike in the QEC case, we reject all outcomes with errors. First, we use the single-qubit depolarizing noise model used in the QEC simulations, with $\varepsilon_{CNOT}=10^{-3}$. This error rate is similar to the two-qubit gate error rates of Quantinuum's quantum computers~\cite{moses2023race, ransford2025Helios}. Furthermore, we set $\varepsilon_{I}/\varepsilon_{CNOT}=0.1$, to explore a regime with higher idling errors.   Here, we expect good performance of the clean-noisy setup despite a much higher $\varepsilon_{I}/\varepsilon_{CNOT}$, since the idling time is reduced by the lack of QEC rounds after each gate. As before, we have $\varepsilon_{CNOT}=\varepsilon_{\rm SPAM}$ and  $\varepsilon_{1}/\varepsilon_{CNOT}=0.1$. Additionally, we perform simulations using a noise model constructed with error rates from publicly available experimental data from Quantinuum's H2 trapped-ion device. Furthermore, to study the robustness of our approach to two-qubit gate errors, in this noise model we represent the $CNOT$ noise by two-qubit depolarizing noise, which is dominated by  weight-2 Pauli errors. We describe this noise model in more detail in SI Sec.~IV. 

We investigate $n\in\{12,14,16,18,20,22\}$ with $n_c=12$, and the noisy $n_c=0$ setup with the same $n$. Since errors accumulate during the circuit execution when not corrected, we choose shallower circuits than in the QEC case. For the  single-qubit depolarizing noise we use $L=6$, and for the H2-informed noise we choose smaller value of $L=4$ to account for higher error rates.The results are shown in Fig.~\ref{fig:fidelity}. Similarly to the QEC case,  we observe that for a fixed $n$ the combined scheme results in better fidelity than the noisy device. More precisely,  we find that the infidelity is reduced for the single-qubit depolarizing   
noise  by a factor of $5.0$  ($3.7$) for $n=14$, and $1.6$ ($1.5$) for $n=22$. As we reject erroneous outcomes, the post-selection rates, which are $0.095-0.11$ for the depolarizing noise and $0.086-0.11$ for the H2-informed noise, are lower than in the QEC case, as detailed in Tab.~IV from SI Sec. III. Hence, the price for the error reduction is an increase of the shot burden by a factor of approximately $10$, or alternatively an increase of the  shot uncertainty.     Looking beyond QED, good performance of our setup here indicates that in the case of high idling noise replacing resource-intensive QEC with less frequent or partial syndrome extraction is a way to exploit the potential of our framework. We note that such approaches to QEC have been already proposed in the case of QEC performed on all registers~\cite{akahoshi2024partially,dangwal2025variational,akahoshi2024compilation,toshio2025practical}.

\section*{\MakeUppercase{Discussion}}
\label{sec:discussion}

Recent experimental demonstrations of QEC suggest that quantum devices will soon be capable of implementing a limited number of error-corrected qubits which, in isolation, will have a significantly lower error rate than their noisy counterparts. Motivated by the fact that physical qubit counts will likely severely restrict the ability to perform classically intractable computations solely using error-corrected  qubits for the foreseeable future, we have proposed implementations of logical gates connecting error-corrected (``clean'') and noisy qubits. We have demonstrated the advantage of combining clean and noisy qubits in terms of mitigating the effects of Pauli noise and real-device-inspired noise models. 

The core of our framework comprises a simple recipe for constructing general partially error-corrected circuits. The first step involves identifying a QEC code that an experimentalist can construct in the lab. Then, one needs to construct an appropriate clean-noisy interaction, such as the $\widetilde{CNOT}$ gate we give for codes that admit transversal $X$ and $Z$ logical gates, that demonstrates an advantage over noisy set-ups. We hope that our framework motivates the construction of other partial QEC set-ups.

Flexibility with the choice of code may also be necessary in order to ensure desired code properties or a certain level of logical error rate.
To this end, one future direction is to explore codes where it is desirable to implement $CNOT$ gates non-transversally.
As an example, lattice surgery~\cite{horsman2012surface} admits fault-tolerant implementations for topological codes~\cite{poulsen2017fault} on devices with nearest-neighbor constraints.  
Sacrificing fault-tolerance, one could implement a clean-noisy 2-qubit gate via lattice surgery by treating the noisy qubit as a distance-1 memory with trivial stabilizers and logical operators of weight 1.
Even though our present approach appears more naturally suited to devices with all-to-all connectivity, lattice surgery may be potentially preferable for certain noise models or for future refinements of our framework with more restricted qubit connectivity.
Crucially, one needs not restrict to transversal implementations of the appropriate clean-noisy interaction used in our recipe.

An important open question stemming from our framework is how to rigorously analyze the error propagation and effective error rate for a clean-noisy gate.
A step towards this direction could be generalizing existing quantum benchmarking protocols to provide information about the noise characteristics of different partial error correction implementations. In particular, a generalization akin to interleaved randomized benchmarking \cite{magesan2012efficient} could allow one to estimate the error rate of the gate(s) used at the clean-noisy boundary, potentially allowing for focused optimization of this crucial circuit component.

We also expect that more intricate correction strategies on the clean registers may allow detection of specific types of errors on noisy registers.
For example, assuming that the noisy qubit is susceptible to bitflip errors only, one could devise a strategy allowing a clean qubit to detect this error under certain conditions, leading to its correction on the noisy register.

To utilize the partial error correction framework in practice, it is crucial to account for the noise in real-world devices and the limitations of early implementations of quantum error correction. As performance is negatively impacted by idling, as found in Sec. ``Random Clifford circuits: setup'', it is important to target practical strategies for minimizing idling time. These strategies should be leveraged to identify the most promising applications, where the gains due to the partial error correction are maximized, while  the additional errors introduced by the noisy-clean gates and the longer idling times are minimized. Rates of those excess errors will depend strongly on  device-specific  noise properties, like ratio of gate times to qubit coherence times, device connectivity, and dynamical decoupling schemes~\cite{viola1999dynamical,rahman2024learning} used to suppress idling errors. Furthermore, they will depend on the circuit structure, the choice of the code, and its implementation.  For example, devices with restricted connectivity may require many layers of  SWAP gates to realize circuits with high connectivity. During those SWAP layers most of the qubits may idle. In turn, the circuit connectivity and depth depend on the choice of the code and its implementation, as usually  code syndromes can be extracted in multiple ways.  Therefore, state-of-the-art device-specific optimization methods~\cite{ransford2025Helios,benito2025comparative} for quantum algorithms need to be applied to specific algorithms to fully utilize  the framework's potential.

Another closely related issue is how to apply quantum error mitigation to implementations of the framework. Error mitigation techniques have already been recognized as crucial for large-scale implementations of partially-fault-tolerant  QEC~\cite{suzuki2022quantum}. Hence, it is natural to envision their  application to the scheme proposed here. However, the presence of both error-corrected and noisy qubits in the scheme poses novel technical challenges to existing mitigation methods.
Therefore, modifications of existing error mitigation methods may be required for their successful application within the partial error correction framework.

One of the main limitations of early QEC  is large overhead of fault-tolerant implementations of non-transversal gates, which are required to execute non-Clifford circuits for many popular codes, like the Steane code and the surface code.  Our partial error correction framework can be used to circumvent this limitation, as  non-Clifford noisy gates are  usually as cheap to execute as Clifford gates in the absence of error correction. If non-Clifford gates occur only on a subset of qubits, in our framework they can be performed as noisy gates by implementing these qubits using noisy qubits. Furthermore,  even if the non-Clifford gates occur on all circuit qubits, it may be useful to execute all non-Clifford gates on noisy qubits as a compilation strategy by swapping qubits as appropriate. 
Another route to reduce the resources required for error-corrected non-Clifford gates is to use the noisy qubits as ancilla qubits that provide  resource states consumed when the non-Clifford gates are implemented through injection~\cite{akahoshi2024partially,toshio2025practical}. These strategies need to be investigated for specific applications, while taking into account real-device noise, to determine their  practical usefulness.        

Finally, it is natural to inquire which applications are particularly suited for our framework. As the noisy qubits introduce errors that are not corrected, an obvious target is to extend the power of early quantum error correction and quantum error detection implementations. Naturally, this suits low-depth algorithms such as 
variational quantum algorithms~\cite{cerezo2020variationalreview,puig2024variational}, Trotterized dynamics simulations~\cite{kim2023evidence}, and subspace-based quantum algorithms~\cite{robledo2025chemistry,kanno2023quantum}. Accounting for the simplicity of implementing non-Clifford gates on noisy qubits, an interesting question is to ask: to what extent can circuit structure be optimized for a given problem to minimize the number of such gates on clean registers? Similarly, we leave it as an open question as to what applications are the most resilient to the non-uniform error profile of our framework.
Beyond the early-QEC era, an important question is to explore if and when the presence of errors on noisy registers may even be advantageous for a quantum algorithm.
An avenue to explore in this regard is noise-assisted algorithms for simulating open quantum systems~\cite{dambal2025harnessing}.

\section*{\MakeUppercase{Methods}}
\label{sec:methods}

\subsection*{Analytic methods}
For our analytic results, we extended common techniques from random circuit analysis as described in Supplementary Information I.

\subsection*{Numerical methods}
For our numerical results, we developed techniques based on the combination of techniques used in simulations of noisy, error-corrected and error-detected quantum circuits as laid out in Supplementary Information II--VI.

\section*{\MakeUppercase{Acknowledgements}}
We thank Max Hunter Gordon, Marco Cerezo, Zo\"e Holmes and Patrick Coles for discussions. The research for this publication has been supported by a grant from the Priority Research Area DigiWorld under the Strategic Programme Excellence Initiative at Jagiellonian University. 
This research has been partially supported by U.S. DOE, Office of Science, National Quantum Information Science Research Centers, Quantum Science Center.

NK and TO have been supported by the U.S. Department of Energy (DOE) through a quantum computing program sponsored by the Los Alamos National Laboratory (LANL)
Information Science and Technology Institute. 
NK acknowledges support by the EPSRC Centre for Doctoral Training in Controlled Quantum Dynamics. 
SW acknowledges support by the Samsung GRP grant. 
TO acknowledges support by the EPSRC through an EPSRC iCASE studentship award in collaboration with IBM Research. DB acknowledges support from the European Union’s Horizon 2020 research and innovation programme under the Marie Sk\l{}odowska-Curie grant agreement No. 955479.
LC acknowledges support from LANL LDRD program under project number 20230049DR.
PC acknowledges support  by  the National Science Centre (NCN), Poland under project 2022/47/D/ST2/03393.

\section*{\MakeUppercase{Code Availability}}

Descriptions of the code used to produce the data in this manuscript are available by the authors upon reasonable request.

\section*{\MakeUppercase{Data Availability}}

All data used to generate the Figures are contained as a SI file.

\section*{\MakeUppercase{Competing Interests}}

The authors declare no competing financial or non-financial interests.

\section*{\MakeUppercase{Author Contributions}}

NK \& SW led the analytic investigation. TO, DB, LC \& PC led the numerical investigation. All authors contributed to discussions and writing of the manuscript.

\bibliography{quantum.bib}

@article{rahn2002exact,
  title = {Exact performance of concatenated quantum codes},
  author = {Rahn, Benjamin and Doherty, Andrew C. and Mabuchi, Hideo},
  journal = {Phys. Rev. A},
  volume = {66},
  issue = {3},
  pages = {032304},
  numpages = {13},
  year = {2002},
  month = {Sep},
  publisher = {American Physical Society},
  doi = {10.1103/PhysRevA.66.032304},
  url = {https://link.aps.org/doi/10.1103/PhysRevA.66.032304}
}

@article{acharya2022suppressing,
  title={Suppressing quantum errors by scaling a surface code logical qubit},
  author={Acharya, Rajeev and Aleiner, Igor and Allen, Richard and Andersen, Trond I and Ansmann, Markus and Arute, Frank and Arya, Kunal and Asfaw, Abraham and Atalaya, Juan and Babbush, Ryan and others},
  journal={Nature},
  volume={614},
  number={7949},
  pages={676--681},
  year={2023},
  publisher={Nature Publishing Group UK London},
url={https://www.nature.com/articles/s41586-022-05434-1},
doi={10.1038/s41586-022-05434-1}
}

@article{steane1996multiple,
  title={Multiple-particle interference and quantum error correction},
  author={Steane, Andrew},
  journal={Proceedings of the Royal Society of London. Series A: Mathematical, Physical and Engineering Sciences},
  volume={452},
  number={1954},
  pages={2551--2577},
  year={1996},
  publisher={The Royal Society London},
  doi={https://doi.org/10.1098/rspa.1996.0136},
  url={https://royalsocietypublishing.org/doi/10.1098/rspa.1996.0136}
}

@article{dalzell2022randomquantum,
  title = {Random Quantum Circuits Anticoncentrate in Log Depth},
  author = {Dalzell, Alexander M. and Hunter-Jones, Nicholas and Brand\~ao, Fernando G. S. L.},
  journal = {PRX Quantum},
  volume = {3},
  issue = {1},
  pages = {010333},
  numpages = {43},
  year = {2022},
  month = {Mar},
  publisher = {American Physical Society},
  doi = {10.1103/PRXQuantum.3.010333},
  url = {https://link.aps.org/doi/10.1103/PRXQuantum.3.010333}
}

@article{robledo2025chemistry,
  title={Chemistry beyond the scale of exact diagonalization on a quantum-centric supercomputer},
  author={Robledo-Moreno, Javier and Motta, Mario and Haas, Holger and Javadi-Abhari, Ali and Jurcevic, Petar and Kirby, William and Martiel, Simon and Sharma, Kunal and Sharma, Sandeep and Shirakawa, Tomonori and others},
  journal={Science Advances},
  volume={11},
  number={25},
  pages={eadu9991},
  year={2025},
  publisher={American Association for the Advancement of Science},
url={https://www.science.org/doi/10.1126/sciadv.adu9991},
doi={10.1126/sciadv.adu9991}
}

@article{cincio2018learning,
	doi = {10.1088/1367-2630/aae94a},
	url = {https://doi.org/10.1088%2F1367-2630%2Faae94a},
	year = 2018,
	month = {nov},
	publisher = {{IOP} Publishing},
	volume = {20},
	number = {11},
	pages = {113022},
	author = {Lukasz Cincio and Yi{\u{g}}it Suba{\c{s}}{\i} and Andrew T Sornborger and Patrick J Coles},
	title = {Learning the quantum algorithm for state overlap},
	journal = {New Journal of Physics},
}

@inproceedings{murali2019noise,
author = {Murali, Prakash and Baker, Jonathan M. and Javadi-Abhari, Ali and Chong, Frederic T. and Martonosi, Margaret},
title = {Noise-Adaptive Compiler Mappings for Noisy Intermediate-Scale Quantum Computers},
year = {2019},
isbn = {9781450362405},
publisher = {Association for Computing Machinery},
address = {New York, NY, USA},
url = {https://doi.org/10.1145/3297858.3304075},
doi = {10.1145/3297858.3304075},
booktitle = {Proceedings of the Twenty-Fourth International Conference on Architectural Support for Programming Languages and Operating Systems},
pages = {1015–1029},
numpages = {15},
location = {Providence, RI, USA},
series = {ASPLOS '19}
}

@article{kitaev2003fault,
  title={Fault-tolerant quantum computation by anyons},
  author={Kitaev, A Yu},
  journal={Annals of physics},
  volume={303},
  number={1},
  pages={2--30},
  year={2003},
  publisher={Elsevier},
  url={https://www.sciencedirect.com/science/article/pii/S0003491602000180},
  doi={10.1016/S0003-4916(02)00018-0}
}

@article{temme2017error,
  title = {Error Mitigation for Short-Depth Quantum Circuits},
  author = {Temme, Kristan and Bravyi, Sergey and Gambetta, Jay M.},
  journal={Physical review letters},
  volume = {119},
  issue = {18},
  pages = {180509},
  numpages = {5},
  year = {2017},
  month = {Nov},
  publisher = {American Physical Society},
  doi = {10.1103/PhysRevLett.119.180509},
  url = {https://link.aps.org/doi/10.1103/PhysRevLett.119.180509}
}

@article{mcclean2018barren,
  title={Barren plateaus in quantum neural network training landscapes},
  author={McClean, Jarrod R and Boixo, Sergio and Smelyanskiy, Vadim N and Babbush, Ryan and Neven, Hartmut},
  journal={Nature {C}ommunications},
  volume={9},
  number={1},
  pages={1--6},
  year={2018},
  publisher={Nature Publishing Group},
  url={https://doi.org/10.1038/s41467-018-07090-4},
  doi={10.1038/s41467-018-07090-4}
}

@article{kubica2015universal,
  title = {Universal transversal gates with color codes: A simplified approach},
  author = {Kubica, Aleksander and Beverland, Michael E.},
  journal = {Phys. Rev. A},
  volume = {91},
  issue = {3},
  pages = {032330},
  numpages = {12},
  year = {2015},
  month = {Mar},
  publisher = {American Physical Society},
  doi = {10.1103/PhysRevA.91.032330},
  url = {https://link.aps.org/doi/10.1103/PhysRevA.91.032330}
}

@article{kubica2015unfolding,
  title={Unfolding the color code},
  author={Kubica, Aleksander and Yoshida, Beni and Pastawski, Fernando},
  journal={New Journal of Physics},
  volume={17},
  number={8},
  pages={083026},
  year={2015},
  publisher={IOP Publishing},
  doi={10.1088/1367-2630/17/8/083026},
  url={https://iopscience.iop.org/article/10.1088/1367-2630/17/8/083026/meta#back-to-top-target}
}

@article{bombin2006topological,
  title={Topological quantum distillation},
  author={Bombin, Hector and Martin-Delgado, Miguel Angel},
  journal={Physical Review Letters},
  volume={97},
  number={18},
  pages={180501},
  year={2006},
  publisher={APS},
  doi={https://doi.org/10.1103/PhysRevLett.97.180501}
}

@article{bombin2015gauge,
  title={Gauge color codes: optimal transversal gates and gauge fixing in topological stabilizer codes},
  author={Bomb{\'\i}n, H{\'e}ctor},
  journal={New Journal of Physics},
  volume={17},
  number={8},
  pages={083002},
  year={2015},
  publisher={IOP Publishing},
  doi={10.1088/1367-2630/17/8/083002}
}

@article{preskill2018quantum,
  title={Quantum Computing in the {NISQ} era and beyond},
  author={Preskill, John},
  journal={Quantum},
  volume={2},
  pages={79},
  year={2018},
  publisher={Verein zur F{\"o}rderung des Open Access Publizierens in den Quantenwissenschaften},
  doi={10.22331/q-2018-08-06-79},
  url={https://quantum-journal.org/papers/q-2018-08-06-79/}
}

@article{nielsen2002simple,
  title={A simple formula for the average gate fidelity of a quantum dynamical operation},
  author={Nielsen, Michael A},
  journal={Physics Letters A},
  volume={303},
  number={4},
  pages={249--252},
  year={2002},
  publisher={Elsevier},
  doi={10.1016/S0375-9601(02)01272-0}
}

@article{stephens2014fault,
  title={Fault-tolerant thresholds for quantum error correction with the surface code},
  author={Stephens, Ashley M},
  journal={Physical Review A},
  volume={89},
  number={2},
  pages={022321},
  year={2014},
  publisher={APS},
  doi={doi.org/10.1103/PhysRevA.89.022321},
  url={https://link.aps.org/doi/10.1103/PhysRevA.89.022321}
}

@article{dangwal2025variational,
      title={Variational Quantum Algorithms in the era of Early Fault Tolerance}, 
      author={Siddharth Dangwal and Suhas Vittal and Lennart Maximillian Seifert and Frederic T. Chong and Gokul Subramanian Ravi},
      year={2025},
      journal={arXiv preprint arXiv:2503.20963},
      url={https://arxiv.org/abs/2503.20963},
      doi={doi.org/10.48550/arXiv.2503.20963}
}

@article{akahoshi2024compilation,
      title={Compilation of Trotter-Based Time Evolution for Partially Fault-Tolerant Quantum Computing Architecture}, 
      author={Yutaro Akahoshi and Riki Toshio and Jun Fujisaki and Hirotaka Oshima and Shintaro Sato and Keisuke Fujii},
      year={2024},
      journal={arXiv preprint arXiv:2408.14929},
      url={https://arxiv.org/abs/2408.14929},
      doi={doi.org/10.48550/arXiv.2408.14929}
}

@article{toshio2025practical,
  title={Practical quantum advantage on partially fault-tolerant quantum computer},
  author={Toshio, Riki and Akahoshi, Yutaro and Fujisaki, Jun and Oshima, Hirotaka and Sato, Shintaro and Fujii, Keisuke},
  journal={Physical Review X},
  volume={15},
  number={2},
  pages={021057},
  year={2025},
  publisher={APS},
  url={https://link.aps.org/doi/10.1103/PhysRevX.15.021057},
  doi={doi.org/10.1103/PhysRevX.15.021057}
}

@article{magesan2011scalable,
  title={Scalable and robust randomized benchmarking of quantum processes},
  author={Magesan, Easwar and Gambetta, Jay M and Emerson, Joseph},
  journal={Physical review letters},
  volume={106},
  number={18},
  pages={180504},
  year={2011},
  publisher={APS},
  url={https://journals.aps.org/prl/abstract/10.1103/PhysRevLett.106.180504},
  doi={10.1103/PhysRevLett.106.180504}
}

@article{knill2008randomized,
  title={Randomized benchmarking of quantum gates},
  author={Knill, Emanuel and Leibfried, Dietrich and Reichle, Rolf and Britton, Joe and Blakestad, R Brad and Jost, John D and Langer, Chris and Ozeri, Roee and Seidelin, Signe and Wineland, David J},
  journal={Physical Review A},
  volume={77},
  number={1},
  pages={012307},
  year={2008},
  publisher={APS},
  url={https://journals.aps.org/pra/abstract/10.1103/PhysRevA.77.012307},
  doi={10.1103/PhysRevA.77.012307}
}

@book{nielsen2000quantum,
 author = {Michael A. Nielsen and Isaac L. Chuang},
 year = {2000},
 title = {Quantum Computation and Quantum Information},
 publisher = {Cambridge University Press},
 address = {Cambridge}
}

@article{bharti2021noisy,
  title={Noisy intermediate-scale quantum algorithms},
  author={Bharti, Kishor and Cervera-Lierta, Alba and Kyaw, Thi Ha and Haug, Tobias and Alperin-Lea, Sumner and Anand, Abhinav and Degroote, Matthias and Heimonen, Hermanni and Kottmann, Jakob S and Menke, Tim and others},
  journal={Reviews of Modern Physics},
  volume={94},
  number={1},
  pages={015004},
  year={2022},
  publisher={APS},
  url={https://journals.aps.org/rmp/abstract/10.1103/RevModPhys.94.015004},
  doi={10.1103/RevModPhys.94.015004}
}

@misc{yin2025flexion,
      title={Flexion: Adaptive In-Situ Encoding for On-Demand QEC in Ion Trap Systems}, 
      author={Keyi Yin and Xiang Fang and Zhuo Chen and Ang Li and David Hayes and Eneet Kaur and Reza Nejabati and Hartmut Haeffner and Wes Campbell and Eric Hudson and Jens Palsberg and Travis Humble and Yufei Ding},
      year={2025},
      eprint={2504.16303},
      archivePrefix={arXiv},
      primaryClass={quant-ph},
      url={https://arxiv.org/abs/2504.16303}, 
}

@article{akahoshi2024partially,
  title = {Partially Fault-Tolerant Quantum Computing Architecture with Error-Corrected Clifford Gates and Space-Time Efficient Analog Rotations},
  author = {Akahoshi, Yutaro and Maruyama, Kazunori and Oshima, Hirotaka and Sato, Shintaro and Fujii, Keisuke},
  journal = {PRX Quantum},
  volume = {5},
  issue = {1},
  pages = {010337},
  numpages = {21},
  year = {2024},
  month = {Mar},
  publisher = {American Physical Society},
  doi = {10.1103/PRXQuantum.5.010337},
  url = {https://link.aps.org/doi/10.1103/PRXQuantum.5.010337}
}

@article{cerezo2020variationalreview,
   title={Variational quantum algorithms},
   author={Cerezo, M. and Arrasmith, Andrew and Babbush, Ryan and Benjamin, Simon C and Endo, Suguru and Fujii, Keisuke and McClean, Jarrod R and Mitarai, Kosuke and Yuan, Xiao and Cincio, Lukasz and Coles, Patrick J. },
   journal={Nature Reviews Physics},
   volume={3},
   number={1},
   pages={625–644},
   publisher={Nature Publishing Group},
   year={2021},
   url={https://www.nature.com/articles/s42254-021-00348-9},
   doi={10.1038/s42254-021-00348-9}
 }

@article{deshpande2022tight,
  title={Tight bounds on the convergence of noisy random circuits to the uniform distribution},
  author={Deshpande, Abhinav and Niroula, Pradeep and Shtanko, Oles and Gorshkov, Alexey V and Fefferman, Bill and Gullans, Michael J},
  journal={PRX Quantum},
  volume={3},
  number={4},
  pages={040329},
  year={2022},
  publisher={APS},
  url={https://doi.org/10.1103/PRXQuantum.3.040329}
}

@article{cincio2020machine,
  title = {Machine Learning of Noise-Resilient Quantum Circuits},
  author = {Cincio, Lukasz and Rudinger, Kenneth and Sarovar, Mohan and Coles, Patrick J.},
  journal = {PRX Quantum},
  volume = {2},
  issue = {1},
  pages = {010324},
  numpages = {19},
  year = {2021},
  month = {Feb},
  publisher = {American Physical Society},
  doi = {10.1103/PRXQuantum.2.010324},
  url = {https://link.aps.org/doi/10.1103/PRXQuantum.2.010324}
}

@article{trout2018simulating,
	doi = {10.1088/1367-2630/aab341},
	url = {https://doi.org/10.1088/1367-2630/aab341},
	year = 2018,
	publisher = {{IOP} Publishing},
	volume = {20},
	number = {4},
	pages = {043038},
	author = {Colin J Trout and Muyuan Li and Mauricio Guti{\'{e}}rrez and Yukai Wu and Sheng-Tao Wang and Luming Duan and Kenneth R Brown},
	title = {Simulating the performance of a distance-3 surface code in a linear ion trap},
	journal = {New Journal of Physics}
}

@article{connor2014using,
  title = {Using Concatenated Quantum Codes for Universal Fault-Tolerant Quantum Gates},
  author = {Jochym-O'Connor, Tomas and Laflamme, Raymond},
  journal = {Phys. Rev. Lett.},
  volume = {112},
  issue = {1},
  pages = {010505},
  numpages = {5},
  year = {2014},
  month = {Jan},
  publisher = {American Physical Society},
  doi = {10.1103/PhysRevLett.112.010505},
  url = {https://link.aps.org/doi/10.1103/PhysRevLett.112.010505}
}

@article{horsman2012surface,
  title={Surface code quantum computing by lattice surgery},
  author={Horsman, Clare and Fowler, Austin G and Devitt, Simon and Van Meter, Rodney},
  journal={New Journal of Physics},
  volume={14},
  number={12},
  pages={123011},
  year={2012},
  publisher={IOP Publishing},
  doi={10.1088/1367-2630/14/12/123011}
}

@article{cincio2021machine,
  title = {Machine Learning of Noise-Resilient Quantum Circuits},
  author = {Cincio, Lukasz and Rudinger, Kenneth and Sarovar, Mohan and Coles, Patrick J.},
  journal = {PRX Quantum},
  volume = {2},
  issue = {1},
  pages = {010324},
  numpages = {19},
  year = {2021},
  month = {Feb},
  publisher = {American Physical Society},
  doi = {10.1103/PRXQuantum.2.010324},
  url = {https://link.aps.org/doi/10.1103/PRXQuantum.2.010324}
}

@article{knill1998power,
  title={Power of one bit of quantum information},
  author={Knill, Emanuel and Laflamme, Raymond},
  journal={Physical Review Letters},
  volume={81},
  number={25},
  pages={5672},
  year={1998},
  publisher={APS},
  url={https://journals.aps.org/prl/abstract/10.1103/PhysRevLett.81.5672},
  doi={10.1103/PhysRevLett.81.5672}
}

@article{beale2018quantum,
  title={Quantum error correction decoheres noise},
  author={Beale, Stefanie J and Wallman, Joel J and Guti{\'e}rrez, Mauricio and Brown, Kenneth R and Laflamme, Raymond},
  journal={Physical review letters},
  volume={121},
  number={19},
  pages={190501},
  year={2018},
  publisher={APS},
  doi = {10.1103/PhysRevLett.121.190501},
  url = {https://link.aps.org/doi/10.1103/PhysRevLett.121.190501}
}

@article{dankert2009exact,
  title={Exact and approximate unitary 2-designs and their application to fidelity estimation},
  author={Dankert, Christoph and Cleve, Richard and Emerson, Joseph and Livine, Etera},
  journal={Physical Review A},
  volume={80},
  number={1},
  pages={012304},
  year={2009},
  publisher={APS},
  url={https://journals.aps.org/pra/abstract/10.1103/PhysRevA.80.012304}
}

@article{orus2014practical,
  title={A practical introduction to tensor networks: Matrix product states and projected entangled pair states},
  author={Or{\'u}s, Rom{\'a}n},
  journal={Annals of Physics},
  volume={349},
  pages={117--158},
  year={2014},
  publisher={Elsevier},
  doi = {10.1016/j.aop.2014.06.013}
}

@article{harrow2009random,
  title={Random quantum circuits are approximate 2-designs},
  author={Harrow, Aram W and Low, Richard A},
  journal={Communications in Mathematical Physics},
  volume={291},
  number={1},
  pages={257--302},
  year={2009},
  publisher={Springer},
  url={https://link.springer.com/article/10.1007%2Fs00220-009-0873-6},
  doi={10.1007/s00220-009-0873-6}
}

@article{divincenzo2002quantum,
  title={Quantum data hiding},
  author={DiVincenzo, David P and Leung, Debbie W and Terhal, Barbara M},
  journal={IEEE Transactions on Information Theory},
  volume={48},
  number={3},
  pages={580--598},
  year={2002},
  publisher={IEEE},
  doi={10.1109/18.985948},
  url={https://doi.org/10.1109/18.985948}
}

@article{fuchs1999cryptographic,
  title={Cryptographic distinguishability measures for quantum-mechanical states},
  author={Fuchs, Christopher A and Van De Graaf, Jeroen},
  journal={IEEE Transactions on Information Theory},
  volume={45},
  number={4},
  pages={1216--1227},
  year={1999},
  publisher={IEEE},
  doi={10.1109/18.761271},
  url = {https://ieeexplore.ieee.org/document/761271}
}

@inproceedings{gottesman1998heisenberg,
  title={The heisenberg representation of quantum computers, talk at},
  author={Gottesman, Daniel},
  booktitle={International Conference on Group Theoretic Methods in Physics},
  year={1998},
  url={http://citeseerx.ist.psu.edu/viewdoc/summary?doi=10.1.1.252.9446},
  organization={Citeseer}
}

@book{ohya2004quantum,
  title={Quantum entropy and its use},
  author={Ohya, Masanori and Petz, D{\'e}nes},
  year={2004},
  publisher={Springer Science \& Business Media}
}

@article{muller2016relative,
  title={Relative entropy convergence for depolarizing channels},
  author={M{\"u}ller-Hermes, Alexander and Stilck Fran{\c{c}}a, Daniel and Wolf, Michael M},
  journal={Journal of Mathematical Physics},
  volume={57},
  number={2},
  pages={022202},
  year={2016},
  publisher={AIP Publishing LLC},
  url={https://aip.scitation.org/doi/10.1063/1.4939560}
}

@article{bultrini2022battle,
  title={The battle of clean and dirty qubits in the era of partial error correction},
  author={Bultrini, Daniel and Wang, Samson and Czarnik, Piotr and Gordon, Max Hunter and Cerezo, M. and Coles, Patrick J. and Cincio, Lukasz},
  journal={Quantum},
  volume={7},
  pages={1060},
  year={2023},
  publisher={Verein zur F{\"o}rderung des Open Access Publizierens in den Quantenwissenschaften},
  url={https://doi.org/10.22331/q-2023-07-13-1060},
  doi={10.22331/q-2023-07-13-1060}
}

@article{fowler2012surface,
  title = {Surface codes: Towards practical large-scale quantum computation},
  author = {Fowler, Austin G. and Mariantoni, Matteo and Martinis, John M. and Cleland, Andrew N.},
  journal = {Physical Review A},
  volume = {86},
  issue = {3},
  pages = {032324},
  numpages = {48},
  year = {2012},
  month = {Sep},
  publisher = {American Physical Society},
  doi = {10.1103/PhysRevA.86.032324},
  url = {https://link.aps.org/doi/10.1103/PhysRevA.86.032324}
}

@article{shende2004minimal,
  title = {Minimal universal two-qubit controlled-NOT-based circuits},
  author = {Shende, Vivek V. and Markov, Igor L. and Bullock, Stephen S.},
  journal = {Physical Review A},
  volume = {69},
  issue = {6},
  pages = {062321},
  numpages = {8},
  year = {2004},
  publisher = {American Physical Society},
  doi = {10.1103/PhysRevA.69.062321},
  url = {https://link.aps.org/doi/10.1103/PhysRevA.69.062321}
}

@article{egan2021fault,
  title={Fault-tolerant control of an error-corrected qubit},
  author={Egan, Laird and Debroy, Dripto M and Noel, Crystal and Risinger, Andrew and Zhu, Daiwei and Biswas, Debopriyo and Newman, Michael and Li, Muyuan and Brown, Kenneth R and Cetina, Marko and others},
  journal={Nature},
  volume={598},
  number={7880},
  pages={281--286},
  year={2021},
  publisher={Nature Publishing Group}, 
  doi={10.1038/s41586-021-03928-y},
  url={https://www.nature.com/articles/s41586-021-03928-y}
}

@article{gottesman2009introduction,
  title = {An Introduction to Quantum Error Correction and Fault-Tolerant Quantum Computation},
  author = {Gottesman, Daniel},
  journal= {Quantum information science and its contributions to mathematics, Proceedings of Symposia in Applied Mathematics},
  year = {2010},
  volume={63},
  pages={13-58},
  url={https://arxiv.org/abs/0904.2557},
  doi={10.1090/psapm/068/2762145}
}

@article{fujii2016power,
  author =	{Keisuke Fujii and Hirotada Kobayashi and Tomoyuki Morimae and Harumichi Nishimura and Shuhei Tamate and Seiichiro Tani},
  title =	{{Power of Quantum Computation with Few Clean Qubits}},
  journal =	{43rd International Colloquium on Automata, Languages, and Programming (ICALP 2016)},
  pages =	{13:1--13:14},
  series =	{Leibniz International Proceedings in Informatics (LIPIcs)},
  ISBN =	{978-3-95977-013-2},
  ISSN =	{1868-8969},
  year =	{2016},
  volume =	{55},
  editor =	{Ioannis Chatzigiannakis and Michael Mitzenmacher and Yuval Rabani and Davide Sangiorgi},
  publisher =	{Schloss Dagstuhl--Leibniz-Zentrum fuer Informatik},
  address =	{Dagstuhl, Germany},
  URL =		{http://drops.dagstuhl.de/opus/volltexte/2016/6296},
  URN =		{urn:nbn:de:0030-drops-62960},
  doi =		{10.4230/LIPIcs.ICALP.2016.13}
}

@article{morimae2017power,
  title = {Power of one nonclean qubit},
  author = {Morimae, Tomoyuki and Fujii, Keisuke and Nishimura, Harumichi},
  journal = {Physical Review A},
  volume = {95},
  issue = {4},
  pages = {042336},
  numpages = {6},
  year = {2017},
  month = {Apr},
  publisher = {American Physical Society},
  doi = {10.1103/PhysRevA.95.042336},
  url = {https://link.aps.org/doi/10.1103/PhysRevA.95.042336}
}

@article{aharonov1996limitations,
  title={Limitations of noisy reversible computation},
  author={Aharonov, Dorit and Ben-Or, Michael and Impagliazzo, Russell and Nisan, Noam},
  journal={arXiv preprint quant-ph/9611028},
  year={1996},
  url={https://arxiv.org/abs/quant-ph/9611028}
}

@article{ben2013quantum,
  title={Quantum refrigerator},
  author={Ben-Or, Michael and Gottesman, Daniel and Hassidim, Avinatan},
  journal={arXiv preprint arXiv:1301.1995},
  year={2013},
  url={https://arxiv.org/abs/1301.1995}
}

@article{cao2021nisq,
  title={{NISQ}: Error Correction, Mitigation, and Noise Simulation},
  author={Cao, Ningping and Lin, Junan and Kribs, David and Poon, Yiu-Tung and Zeng, Bei and Laflamme, Raymond},
  journal={arXiv preprint arXiv:2111.02345},
  year={2021},
  url = {https://arxiv.org/abs/2111.02345}
}

@article{suzuki2022quantum,
  title={Quantum Error Mitigation as a Universal Error Reduction Technique: Applications from the {NISQ} to the Fault-Tolerant Quantum Computing Eras},
  author={Suzuki, Yasunari and Endo, Suguru and Fujii, Keisuke and Tokunaga, Yuuki},
  journal={PRX Quantum},
  volume={3},
  number={1},
  pages={010345},
  year={2022},
  publisher={APS},
  url={https://journals.aps.org/prxquantum/abstract/10.1103/PRXQuantum.3.010345}
}

@inproceedings{holmes2020nisq+,
  title={{NISQ}+: Boosting quantum computing power by approximating quantum error correction},
  author={Holmes, Adam and Jokar, Mohammad Reza and Pasandi, Ghasem and Ding, Yongshan and Pedram, Massoud and Chong, Frederic T},
  booktitle={2020 ACM/IEEE 47th Annual International Symposium on Computer Architecture (ISCA)},
  pages={556--569},
  year={2020},
  organization={IEEE},
  url={https://dl.acm.org/doi/10.1109/ISCA45697.2020.00053}
}

@article{cai2022quantum,
  title={Quantum error mitigation},
  author={Cai, Zhenyu and Babbush, Ryan and Benjamin, Simon C and Endo, Suguru and Huggins, William J and Li, Ying and McClean, Jarrod R and O’Brien, Thomas E},
  journal={Reviews of Modern Physics},
  volume={95},
  number={4},
  pages={045005},
  year={2023},
  publisher={APS},
url={https://journals.aps.org/rmp/abstract/10.1103/RevModPhys.95.045005},
doi={10.1103/RevModPhys.95.045005}
}

@article{krinner2022realizing,
  title={Realizing repeated quantum error correction in a distance-three surface code},
  author={Krinner, Sebastian and Lacroix, Nathan and Remm, Ants and Di Paolo, Agustin and Genois, Elie and Leroux, Catherine and Hellings, Christoph and Lazar, Stefania and Swiadek, Francois and Herrmann, Johannes and others},
  journal={Nature},
  volume={605},
  number={7911},
  pages={669--674},
  year={2022},
  publisher={Nature Publishing Group UK London},
  doi={https://doi.org/10.1038/s41586-022-04566-8}
}

@article{zhao2022realization,
  title={Realization of an error-correcting surface code with superconducting qubits},
  author={Zhao, Youwei and Ye, Yangsen and Huang, He-Liang and Zhang, Yiming and Wu, Dachao and Guan, Huijie and Zhu, Qingling and Wei, Zuolin and He, Tan and Cao, Sirui and others},
  journal={Physical Review Letters},
  volume={129},
  number={3},
  pages={030501},
  year={2022},
  publisher={APS},
  doi = {10.1103/PhysRevLett.129.030501},
  url = {https://link.aps.org/doi/10.1103/PhysRevLett.129.030501}
}

@article{sundaresan2023demonstrating,
  title={Demonstrating multi-round subsystem quantum error correction using matching and maximum likelihood decoders},
  author={Sundaresan, Neereja and Yoder, Theodore J and Kim, Youngseok and Li, Muyuan and Chen, Edward H and Harper, Grace and Thorbeck, Ted and Cross, Andrew W and C{\'o}rcoles, Antonio D and Takita, Maika},
  journal={Nature Communications},
  volume={14},
  number={1},
  pages={2852},
  year={2023},
  publisher={Nature Publishing Group UK London},
url={https://www.nature.com/articles/s41467-023-38247-5},
doi={10.1038/s41467-023-38247-5}
}

@article{fannes1992finitely,
	author = {Fannes, M. and Nachtergaele, B. and Werner, R.},
	journal = {Comm. in Math. Phys.},
	pages = {443-490},
	title = {Finitely correlated states on quantum spin chains},
	volume = {144},
	year = {1992},
	url = {http://dx.doi.org/10.1007/BF02099178}}

@article{aharonov2022polynomial,
  title={A polynomial-time classical algorithm for noisy random circuit sampling},
  author={Aharonov, Dorit and Gao, Xun and Landau, Zeph and Liu, Yunchao and Vazirani, Umesh},
  journal={Proceedings of the 55th Annual ACM Symposium on Theory of Computing},
  pages={945},
  year={2023},
  doi={10.1145/3564246.3585234},
  url={https://dl.acm.org/doi/abs/10.1145/3564246.3585234},
  publisher={Association for Computing Machinery},
  address={New York, NY, USA},
}

@article{bombin2013self,
  title={Self-correcting quantum computers},
  author={Bombin, Hector and Chhajlany, Ravindra W and Horodecki, Micha{\l} and Martin-Delgado, Miguel-Angel},
  journal={New Journal of Physics},
  volume={15},
  number={5},
  pages={055023},
  year={2013},
  publisher={IOP Publishing},
  doi={10.1088/1367-2630/15/5/055023}
}

@article{brown2016fault,
  title={Fault-tolerant error correction with the gauge color code},
  author={Brown, Benjamin J and Nickerson, Naomi H and Browne, Dan E},
  journal={Nature communications},
  volume={7},
  number={1},
  pages={12302},
  year={2016},
  publisher={Nature Publishing Group UK London}
}

@article{combes2017logical,
  title={Logical randomized benchmarking},
  author={Combes, Joshua and Granade, Christopher and Ferrie, Christopher and Flammia, Steven T},
  journal={arXiv preprint arXiv:1702.03688},
  year={2017},
  url={https://arxiv.org/abs/1702.03688}
}

@article{viola1999dynamical,
  title = {Dynamical Decoupling of Open Quantum Systems},
  author = {Viola, Lorenza and Knill, Emanuel and Lloyd, Seth},
  journal = {Phys. Rev. Lett.},
  volume = {82},
  issue = {12},
  pages = {2417--2421},
  numpages = {0},
  year = {1999},
  month = {Mar},
  publisher = {American Physical Society},
  doi = {10.1103/PhysRevLett.82.2417},
  url = {https://link.aps.org/doi/10.1103/PhysRevLett.82.2417}
}

@article{self2022protecting,
  title={Protecting expressive circuits with a quantum error detection code},
  author={Self, Chris N and Benedetti, Marcello and Amaro, David},
  journal={Nature Physics},
  volume={20},
  number={2},
  pages={219--224},
  year={2024},
  publisher={Nature Publishing Group UK London},
  url={https://www.nature.com/articles/s41567-023-02282-2},
  doi={10.1038/s41567-023-02282-2}
}

@article{bravyi2005universal,
  title = {Universal quantum computation with ideal Clifford gates and noisy ancillas},
  author = {Bravyi, Sergey and Kitaev, Alexei},
  journal = {Phys. Rev. A},
  volume = {71},
  issue = {2},
  pages = {022316},
  numpages = {14},
  year = {2005},
  month = {Feb},
  publisher = {American Physical Society},
  doi = {10.1103/PhysRevA.71.022316},
  url = {https://link.aps.org/doi/10.1103/PhysRevA.71.022316}
}

@article{nigg2014quantum,
  title={Quantum computations on a topologically encoded qubit},
  author={Nigg, Daniel and Mueller, Markus and Martinez, Esteban A and Schindler, Philipp and Hennrich, Markus and Monz, Thomas and Martin-Delgado, Miguel A and Blatt, Rainer},
  journal={Science},
  volume={345},
  number={6194},
  pages={302--305},
  year={2014},
  publisher={American Association for the Advancement of Science},
  doi = {10.1126/science.1253742}
}

@article{postler2022demonstration,
  title={Demonstration of fault-tolerant universal quantum gate operations},
  author={Postler, Lukas and Heu$\beta$en, Sascha and Pogorelov, Ivan and Rispler, Manuel and Feldker, Thomas and Meth, Michael and Marciniak, Christian D and Stricker, Roman and Ringbauer, Martin and Blatt, Rainer and others},
  journal={Nature},
  volume={605},
  number={7911},
  pages={675--680},
  year={2022},
  doi={10.5281/zenodo.6244536},
  publisher={Nature Publishing Group UK London}
}

@article{paetznick2013universal,
  title = {Universal Fault-Tolerant Quantum Computation with Only Transversal Gates and Error Correction},
  author = {Paetznick, Adam and Reichardt, Ben W.},
  journal = {Phys. Rev. Lett.},
  volume = {111},
  issue = {9},
  pages = {090505},
  numpages = {5},
  year = {2013},
  month = {Aug},
  publisher = {American Physical Society},
  doi = {10.1103/PhysRevLett.111.090505},
  url = {https://link.aps.org/doi/10.1103/PhysRevLett.111.090505}
}

@article{poulsen2017fault,
  title={Fault-tolerant interface between quantum memories and quantum processors},
  author={Poulsen Nautrup, Hendrik and Friis, Nicolai and Briegel, Hans J},
  journal={Nature communications},
  volume={8},
  number={1},
  pages={1321},
  year={2017},
  publisher={Nature Publishing Group UK London},
  doi={10.1038/s41467-017-01418-2}
}

@article{vasmer2019three,
  title = {Three-dimensional surface codes: Transversal gates and fault-tolerant architectures},
  author = {Vasmer, Michael and Browne, Dan E.},
  journal = {Phys. Rev. A},
  volume = {100},
  issue = {1},
  pages = {012312},
  numpages = {20},
  year = {2019},
  month = {Jul},
  publisher = {American Physical Society},
  doi = {10.1103/PhysRevA.100.012312},
  url = {https://link.aps.org/doi/10.1103/PhysRevA.100.012312}
}

@article{aaronson2004improved,
  title={Improved simulation of stabilizer circuits},
  author={Aaronson, Scott and Gottesman, Daniel},
  journal={Physical Review A},
  volume={70},
  number={5},
  pages={052328},
  year={2004},
  publisher={APS},
  doi = {10.1103/PhysRevA.70.052328},
  url = {https://link.aps.org/doi/10.1103/PhysRevA.70.052328}
}

@article{magesan2012efficient,
  title={Efficient measurement of quantum gate error by interleaved randomized benchmarking},
  author={Magesan, Easwar and Gambetta, Jay M and Johnson, Blake R and Ryan, Colm A and Chow, Jerry M and Merkel, Seth T and Da Silva, Marcus P and Keefe, George A and Rothwell, Mary B and Ohki, Thomas A and others},
  journal={Physical review letters},
  volume={109},
  number={8},
  pages={080505},
  year={2012},
  publisher={APS},
  doi={10.1103/PhysRevLett.109.080505},
  url={https://doi.org/10.1103/PhysRevLett.109.080505}
}

@misc{mayer2024benchmarking,
  title={Benchmarking logical three-qubit quantum Fourier transform encoded in the Steane code on a trapped-ion quantum computer}, 
  author={Karl Mayer and Ciarán Ryan-Anderson and Natalie Brown and Elijah Durso-Sabina and Charles H. Baldwin and David Hayes and Joan M. Dreiling and Cameron Foltz and John P. Gaebler and Thomas M. Gatterman and Justin A. Gerber and Kevin Gilmore and Dan Gresh and Nathan Hewitt and Chandler V. Horst and Jacob Johansen and Tanner Mengle and Michael Mills and Steven A. Moses and Peter E. Siegfried and Brian Neyenhuis and Juan Pino and Russell Stutz},
  journal={arXiv preprint arXiv:2404.08616},
  year={2024},
  url={https://arxiv.org/abs/2404.08616}
}

@article{pogorelov2024experimental,
  title={Experimental fault-tolerant code switching},
  author={Pogorelov, Ivan and Butt, Friederike and Postler, Lukas and Marciniak, Christian D and Schindler, Philipp and M{\"u}ller, Markus and Monz, Thomas},
  journal={arXiv preprint arXiv:2403.13732},
  year={2024},
  url={https://arxiv.org/abs/2403.13732}
}

@article{campbell2017roads,
  title={Roads towards fault-tolerant universal quantum computation},
  author={Campbell, Earl T and Terhal, Barbara M and Vuillot, Christophe},
  journal={Nature},
  volume={549},
  number={7671},
  pages={172--179},
  year={2017},
  publisher={Nature Publishing Group UK London},
  doi={https://doi.org/10.1038/nature23460},
  url={https://www.nature.com/articles/nature23460}
}

@article{livingston2022experimental,
  title={Experimental demonstration of continuous quantum error correction},
  author={Livingston, William P and Blok, Machiel S and Flurin, Emmanuel and Dressel, Justin and Jordan, Andrew N and Siddiqi, Irfan},
  journal={Nature communications},
  volume={13},
  number={1},
  pages={1--7},
  year={2022},
  publisher={Nature Publishing Group},
  url={https://www.nature.com/articles/s41467-022-29906-0}
}

@article{eastin2009restrictions,
  title={Restrictions on transversal encoded quantum gate sets},
  author={Eastin, Bryan and Knill, Emanuel},
  journal={Physical review letters},
  volume={102},
  number={11},
  pages={110502},
  year={2009},
  publisher={APS},
  doi = {10.1103/PhysRevLett.102.110502},
  url = {https://link.aps.org/doi/10.1103/PhysRevLett.102.110502}
}

@article{piveteau2021error,
  title = {Error Mitigation for Universal Gates on Encoded Qubits},
  author = {Piveteau, Christophe and Sutter, David and Bravyi, Sergey and Gambetta, Jay M. and Temme, Kristan},
  journal = {Phys. Rev. Lett.},
  volume = {127},
  issue = {20},
  pages = {200505},
  numpages = {6},
  year = {2021},
  month = {Nov},
  publisher = {American Physical Society},
  doi = {10.1103/PhysRevLett.127.200505},
  url = {https://link.aps.org/doi/10.1103/PhysRevLett.127.200505}
}

@inproceedings{wahl2023zero,
  title={Zero noise extrapolation on logical qubits by scaling the error correction code distance},
  author={Wahl, Misty A and Mari, Andrea and Shammah, Nathan and Zeng, William J and Ravi, Gokul Subramanian},
  booktitle={2023 IEEE International Conference on Quantum Computing and Engineering (QCE)},
  volume={1},
  pages={888--897},
  year={2023},
  organization={IEEE},
  url={https://ieeexplore.ieee.org/document/10313873}
}

@article{peng2020simulating,
  title = {Simulating Large Quantum Circuits on a Small Quantum Computer},
  author = {Peng, Tianyi and Harrow, Aram W. and Ozols, Maris and Wu, Xiaodi},
  journal = {Phys. Rev. Lett.},
  volume = {125},
  issue = {15},
  pages = {150504},
  numpages = {6},
  year = {2020},
  month = {Oct},
  publisher = {American Physical Society},
  doi = {10.1103/PhysRevLett.125.150504},
  url = {https://link.aps.org/doi/10.1103/PhysRevLett.125.150504}
}

@article{piveteau2023circuit,
  title={Circuit knitting with classical communication},
  author={Piveteau, Christophe and Sutter, David},
  journal={IEEE Transactions on Information Theory},
  volume={70},
  number={4},
  pages={2734--2745},
  year={2023},
  publisher={IEEE},
  url={https://ieeexplore.ieee.org/document/10236453}
}

@article{dennis2002topological,
  title={Topological quantum memory},
  author={Dennis, Eric and Kitaev, Alexei and Landahl, Andrew and Preskill, John},
  journal={Journal of Mathematical Physics},
  volume={43},
  number={9},
  pages={4452--4505},
  year={2002},
  doi={https://doi.org/10.1063/1.1499754},
  publisher={American Institute of Physics}
}

@article{raussendorf2007fault,
  title = {Fault-Tolerant Quantum Computation with High Threshold in Two Dimensions},
  author = {Raussendorf, Robert and Harrington, Jim},
  journal = {Phys. Rev. Lett.},
  volume = {98},
  issue = {19},
  pages = {190504},
  numpages = {4},
  year = {2007},
  month = {May},
  publisher = {American Physical Society},
  doi = {10.1103/PhysRevLett.98.190504},
  url = {https://link.aps.org/doi/10.1103/PhysRevLett.98.190504}
}

@article{puig2024variational,
  title = {Variational Quantum Simulation: A Case Study for Understanding Warm Starts},
  author = {Puig, Ricard and Drudis, Marc and Thanasilp, Supanut and Holmes, Zo\"e},
  journal = {PRX Quantum},
  volume = {6},
  issue = {1},
  pages = {010317},
  numpages = {39},
  year = {2025},
  month = {Jan},
  publisher = {American Physical Society},
  doi = {10.1103/PRXQuantum.6.010317},
  url = {https://link.aps.org/doi/10.1103/PRXQuantum.6.010317}
}

@article{kim2023evidence,
  title={Evidence for the utility of quantum computing before fault tolerance},
  author={Kim, Youngseok and Eddins, Andrew and Anand, Sajant and Wei, Ken Xuan and Van Den Berg, Ewout and Rosenblatt, Sami and Nayfeh, Hasan and Wu, Yantao and Zaletel, Michael and Temme, Kristan and others},
  journal={Nature},
  volume={618},
  number={7965},
  pages={500--505},
  year={2023},
  publisher={Nature Publishing Group UK London},
  doi={10.1038/s41586-023-06096-3},
  url={https://www.nature.com/articles/s41586-023-06096-3}
}

@article{haghshenas2026digital,
  title={Digital quantum magnetism on a trapped-ion quantum computer},
  author={Haghshenas, Reza and Chertkov, Eli and Mills, Michael and Kadow, Wilhelm and Lin, S-H and Chen, Yi-Hsiang and Cade, Chris and Niesen, Ido and Begu{\v{s}}i{\'c}, Tomislav and Rudolph, Manuel S and others},
  journal={Nature},
  pages={1--7},
  year={2026},
  publisher={Nature Publishing Group UK London},
  url={https://www.nature.com/articles/s41586-026-10445-3}
}

@article{ryan2021realization,
  title = {Realization of Real-Time Fault-Tolerant Quantum Error Correction},
  author = {Ryan-Anderson, C. and Bohnet, J. G. and Lee, K. and Gresh, D. and Hankin, A. and Gaebler, J. P. and Francois, D. and Chernoguzov, A. and Lucchetti, D. and Brown, N. C. and Gatterman, T. M. and Halit, S. K. and Gilmore, K. and Gerber, J. A. and Neyenhuis, B. and Hayes, D. and Stutz, R. P.},
  journal = {Phys. Rev. X},
  volume = {11},
  issue = {4},
  pages = {041058},
  numpages = {29},
  year = {2021},
  month = {Dec},
  publisher = {American Physical Society},
  doi = {10.1103/PhysRevX.11.041058},
  url = {https://link.aps.org/doi/10.1103/PhysRevX.11.041058}
}

@article{sivak2023real,
  title={Real-time quantum error correction beyond break-even},
  author={Sivak, VV and Eickbusch, Alec and Royer, Baptiste and Singh, Shraddha and Tsioutsios, Ioannis and Ganjam, Suhas and Miano, Alessandro and Brock, BL and Ding, AZ and Frunzio, Luigi and others},
  journal={Nature},
  volume={616},
  number={7955},
  pages={50--55},
  year={2023},
  publisher={Nature Publishing Group UK London},
url={https://www.nature.com/articles/s41586-023-05782-6}
}

@article{bluvstein2023logical,
  title={Logical quantum processor based on reconfigurable atom arrays},
  author={Bluvstein, Dolev and Evered, Simon J and Geim, Alexandra A and Li, Sophie H and Zhou, Hengyun and Manovitz, Tom and Ebadi, Sepehr and Cain, Madelyn and Kalinowski, Marcin and Hangleiter, Dominik and others},
  journal={Nature},
  pages={1--3},
  year={2023},
  publisher={Nature Publishing Group UK London},
url={https://www.nature.com/articles/s41586-023-06927-3#article-info}
}

@article{moses2023race,
  title={A race-track trapped-ion quantum processor},
  author={Moses, Steven A and Baldwin, Charles H and Allman, Michael S and Ancona, R and Ascarrunz, L and Barnes, C and Bartolotta, J and Bjork, B and Blanchard, P and Bohn, M and others},
  journal={Physical Review X},
  volume={13},
  number={4},
  pages={041052},
  year={2023},
  publisher={APS},
  doi={10.1103/PhysRevX.13.041052},
  url={https://doi.org/10.1103/PhysRevX.13.041052}
}

@article{google2025quantum,
  title={Quantum error correction below the surface code threshold},
  author={Google Quantum AI and Collaborators},
  journal={Nature},
  year={2025},
  volume={638},
  number={8013},
  doi = {10.1038/s41586-024-08449-y},
  url = {https://doi.org/10.1038/s41586-024-08449-y}
}

@article{benito2025comparative,
  title={Comparative study of quantum error correction strategies for the heavy-hexagonal lattice},
  author={Benito, C{\'e}sar and L{\'o}pez, Esperanza and Peropadre, Borja and Bermudez, Alejandro},
  journal={Quantum},
  volume={9},
  pages={1623},
  year={2025},
  publisher={Verein zur F{\"o}rderung des Open Access Publizierens in den Quantenwissenschaften},
  doi = {10.22331/q-2025-02-06-1623},
  url = {https://doi.org/10.22331/q-2025-02-06-1623}
}

@article{ransford2025helios,
  title={Helios: A 98-qubit trapped-ion quantum computer},
  author={Ransford, Anthony and Allman, MS and Arkinstall, Jake and Campora III, JP and Cooper, Samuel F and Delaney, Robert D and Dreiling, Joan M and Estey, Brian and Figgatt, Caroline and Hall, Alex and others},
  journal={arXiv preprint arXiv:2511.05465},
  url={https://arxiv.org/abs/2511.05465},
  doi={10.48550/arXiv.2511.05465},
  year={2025}
}

@article{sepiol2019probing,
  title = {Probing Qubit Memory Errors at the Part-per-Million Level},
  author = {Sepiol, M. A. and Hughes, A. C. and Tarlton, J. E. and Nadlinger, D. P. and Ballance, T. G. and Ballance, C. J. and Harty, T. P. and Steane, A. M. and Goodwin, J. F. and Lucas, D. M.},
  journal = {Phys. Rev. Lett.},
  volume = {123},
  issue = {11},
  pages = {110503},
  numpages = {5},
  year = {2019},
  month = {Sep},
  publisher = {American Physical Society},
  doi = {10.1103/PhysRevLett.123.110503},
  url = {https://link.aps.org/doi/10.1103/PhysRevLett.123.110503}
}

@article{chertkov2025error,
  title = {Error detection without postselection in adaptive quantum circuits},
  author = {Chertkov, Eli and Potter, Andrew C. and Hayes, David and Foss-Feig, Michael},
  journal = {Phys. Rev. Res.},
  volume = {8},
  issue = {2},
  pages = {023057},
  numpages = {8},
  year = {2026},
  month = {Apr},
  publisher = {American Physical Society},
  doi = {10.1103/6b14-7wkt},
  url = {https://link.aps.org/doi/10.1103/6b14-7wkt}
}

@article{kanno2023quantum,
  title={Quantum-selected configuration interaction: Classical diagonalization of Hamiltonians in subspaces selected by quantum computers},
  author={Kanno, Keita and Kohda, Masaya and Imai, Ryosuke and Koh, Sho and Mitarai, Kosuke and Mizukami, Wataru and Nakagawa, Yuya O},
  journal={arXiv preprint arXiv:2302.11320},
  year={2023},
  url={https://arxiv.org/abs/2302.11320},
  doi={10.48550/arXiv.2302.11320}
}

@article{dambal2025harnessing,
  title={Harnessing intrinsic noise for quantum simulation of open quantum systems},
  author={Dambal, Sameer and Sone, Akira and Zhang, Yu},
  journal={arXiv preprint arXiv:2510.21075},
  year={2025},
  url={https://arxiv.org/abs/2510.21075},
  doi={10.48550/arXiv.2510.21075}
}

@article{rahman2024learning,
  title = {Learning how to dynamically decouple by optimizing rotational gates},
  author = {Rahman, Arefur and Egger, Daniel J. and Arenz, Christian},
  journal = {Phys. Rev. Appl.},
  volume = {22},
  issue = {5},
  pages = {054074},
  numpages = {12},
  year = {2024},
  month = {Nov},
  publisher = {American Physical Society},
  doi = {10.1103/PhysRevApplied.22.054074},
  url = {https://link.aps.org/doi/10.1103/PhysRevApplied.22.054074}
}

@Article{varbanov2020leakage,
author={Varbanov, Boris Mihailov
and Battistel, Francesco
and Tarasinski, Brian Michael
and Ostroukh, Viacheslav Petrovych
and O'Brien, Thomas Eugene
and DiCarlo, Leonardo
and Terhal, Barbara Maria},
title={Leakage detection for a transmon-based surface code},
journal={npj Quantum Information},
year={2020},
month={Dec},
day={14},
volume={6},
number={1},
pages={102},
issn={2056-6387},
doi={10.1038/s41534-020-00330-w},
url={https://doi.org/10.1038/s41534-020-00330-w}
}
\bibliographystyle{naturemag}

\onecolumngrid

\newpage

\begin{center}
{\Large \textbf{Supplemetary Information -- A framework of partial error correction for intermediate-scale quantum computers}}
\end{center}

\section{Analytic scaling results}\label{appdx:scaling-results}

\subsection{Definitions}

Throughout this section we will use boldface to denote vectors. 
The notation $\|\cdot\|_1$ will denote the Schatten 1-norm (trace norm) when applied to a matrix $A$, i.e.~ $\|A\|_1 \coloneqq \Tr|A|$, and the $\ell_1$ vector norm when applied to a vector $\bmv = (v_1, \dots, v_N) \in \mathbb{R}^N$, i.e.~$\|\bmv\|_1 \coloneqq \sum_i^N |v_i|$.

We also denote the set on $n$-qubit Pauli operators as follows.
\begin{definition}[Pauli operators]\label{def:Pauli}
We denote the set of Pauli operators on $n$ qubits as 
\begin{align}
    \mathrm{P}_n \coloneqq \{ \id, X, Y, Z \}^{\otimes n}\,.
\end{align}
For a given Pauli operator $s\in \mathrm{P}_n$, we denote $|s|$ as the Hamming weight of non-identity single-qubit Pauli operators.
\end{definition}

We now introduce the noise channels that appear in our noise models, starting with the depolarizing channel. 
\begin{definition}[Local uniform depolarizing noise]\label{def:DC}
    The single-qubit depolarizing noise channel $\D$ maps any operator $A \in \mathbb{C}^{2\times 2}$ to  
    \begin{align}
        \D(A) = (1-\varepsilon)A + \varepsilon \Tr[A] \frac{\id}{2}\,,
    \end{align}
    where we will refer to $\varepsilon$ as the error rate. 
    In particular, $\D$ maps single-qubit quantum states to
    \begin{align}
        \D(\rho) = (1-\varepsilon)\rho + \varepsilon \frac{\id}{2}.
    \end{align}
\end{definition}

Our results in fact apply when we consider general Pauli channels.

\begin{definition}[Local Pauli noise]\label{def:PauliC}
    The single-qubit Pauli noise channel $\PC$ maps any operator $A \in \mathbb{C}^{2\times 2}$ to  
    \begin{align}
        \PC(A) = \left( 1 \hspace{5pt}-\hspace{-10pt} \sum_{Q \in \{X,Y,Z\}} \hspace{-5pt}p_Q \right) A \hspace{5pt}+\hspace{-10pt} \sum_{Q \in \{X,Y,Z\}} \hspace{-5pt}p_Q Q A Q\,.
    \end{align}
    We denote the \emph{error rate} $\varepsilon \coloneqq 1- \min_{Q\in {\{X,Y,Z\}}}|1-\varepsilon_Q|$, where $\varepsilon_Q = 2 p_{Q'} + 2 p_{Q''}$ for $\{Q,Q',Q''\} = \{X,Y,Z\}$.
\end{definition}
The channel $\PC$ is the unital map that suppresses non-identity Pauli operators as $Q \rightarrow (1 - \varepsilon_Q) Q$ for any $Q \in \{X,Y,Z\}$, justifying naming $\varepsilon$ as the error rate. Indeed, we can equally write $\varepsilon = 1- \frac{1}{2}\min_{Q\in {\{X,Y,Z\}}} |\Tr[Q \PC(Q)]|$. We will use this picture in our results.
When $p_X = p_Y = p_Z = \varepsilon/4$, the action of channel $\PC$ on a Pauli is equivalent to the action of the depolarizing channel $\D$.

We can generalize Definition~\ref{def:PauliC} to an arbitrary (correlated) Pauli noise channel.
\begin{definition}[$m$-qubit Pauli noise]\label{def:PauliCorr}
    The $m$-qubit Pauli noise channel $\PC$ maps any operator $A \in \mathbb{C}^{2^m\times 2^m}$ to  
    \begin{align}
        \PC(A) = \left( 1 \hspace{5pt}-\hspace{-10pt} \sum_{Q \in \mathrm{P}_m\setminus\{\id^{\otimes m}\}} \hspace{-5pt}p_Q \right) A \hspace{5pt}+\hspace{-10pt} \sum_{Q \in \mathrm{P}_m\setminus\{\id^{\otimes m}\}} \hspace{-5pt}p_Q Q A Q\,.
    \end{align}
    For this channel, we denote the \emph{error rate} $\varepsilon \coloneqq 1- \frac{1}{2^m}\min_{Q \in \mathrm{P}_m\setminus\{\id^{\otimes m}\}} |\Tr[Q \PC(Q)]|$.
\end{definition}

\begin{figure}[t]
\begin{center}
\begin{quantikz}
\lstick{}{\ket{0}} & \gate[2]{ } & \gate[style={rounded corners}]{\PC_{11}} & \gate[]{} & \gate[2, style={rounded corners}]{\PC_{12}} & \gate[2]{ } & \gate[style={rounded corners}]{\PC_{13}} & \gate[]{} & \gate[style={rounded corners}]{\PC_{14}} & \qw  \\
\lstick{}{\ket{0}} & \qw & \gate[3,style={rounded corners}]{\PC_{21}}  & \gate[2]{ } & \qw & \qw & \gate[style={rounded corners}]{\PC_{23}} & \gate[2]{ } & \gate[2, style={rounded corners}]{\PC_{24}} & \qw  \\
\lstick{}{\ket{0}}  & \gate[2]{ } & \qw & \qw  & \gate[style={rounded corners}]{\PC_{22}} & \gate[2]{ } & \gate[style={rounded corners}]{\PC_{33}} & \qw & \qw & \qw \\
\lstick{}{\ket{0}} & \qw & \qw & \gate[2]{ } & \gate[2,style={rounded corners}]{\PC_{32}} & \qw & \gate[style={rounded corners}]{\PC_{43}} & \gate[2]{ } & \gate[3, style={rounded corners}]{\PC_{34}} & \qw  \\
\lstick{}{\ket{0}} & \gate[2]{} & \gate[2, style={rounded corners}]{\PC_{31}} & \qw & \qw & \gate[2]{} & \gate[style={rounded corners}]{\PC_{53}} & \qw & \qw & \qw  \\
\lstick{}{\ket{0}} & \qw & \qw & \gate[]{} & \gate[style={rounded corners}]{\PC_{22}} & \qw & \gate[style={rounded corners}]{\PC_{63}} & \gate[]{} &  \qw & \qw 
\end{quantikz}
\caption{\textbf{Setting for lower bound analysis.} We consider an alternating (brick-)layered ansatz on a 1D ring architecture. After each layer of two-qubit gates, there is a layer of Pauli noise channels with different error rates depending on their location in the circuit. The Pauli noise may be correlated (for instance, as in the first layer); correlated across non-adjacent qubits (as in the second layer); or completely uncorrelated (as in the third layer). Thus, this setting includes as a special case the uncorrelated depolarizing noise model of Fig.~3 and Theorem~2 in the main text, as well as more generalized models of Pauli noise and correlated Pauli noise.
We consider the average behavior of this circuit over an ensemble of gates where each local two-qubit gate is drawn independently from a 2-design.}
\label{fig:error_scramble_setup_app}
\end{center}
\end{figure}

We then explicitly define the circuit architectures on which our results apply.
\begin{definition}[Circuit architectures]\label{def:CA}
    We draw random quantum circuits from the following architecture.
    Each circuit $C \coloneqq \{\N_1 \circ \U_1, \N_2 \circ \U_2, \dots, \N_L \circ \U_L\}$ is defined on $n$ registers as a collection of $L$ layers, where $n$ and $L$ are even. 
    Each layer consists of a layer of unitary gates followed by an instance of single-qubit Pauli noise on all registers.
    
    Each unitary layer, $\U_{\ell}$ for $\ell=1,\dots,L$, consists of a tensor product of two-qubit unitaries as in Fig.~\ref{fig:error_scramble_setup_app}. The two-qubit unitaries are drawn from a distribution $\YC$ over 2-designs.

    Each noise layer, $\N_{\ell}$ for $\ell=1,\dots,L$, consists of a tensor product of local Pauli noise channels $\{\PC_{i \ell}\}_i$ with different error rates $\varepsilon_{i\ell}$ on each register $i$.
\end{definition}

\begin{definition}[Total variation distance]\label{def:TVD}
    Let ${\vec{p}}(C)\in \mathbb{R}^{2^n}$ be the output probability distribution of circuit $C$ in the computational basis. 
    We define the total variation distance of the output probability distribution from the uniform distribution as
    \begin{align}
        \delta(C) \coloneqq \frac{1}{2} \left\| {\vec{p}}(C) - \vec{u} \right\|_1\,,
    \end{align}
    where $\|\cdot\|_1$ denotes the $\ell_1$ vector norm, and $\vec{u}:=\left(\frac{1}{2^n},..., \frac{1}{2^n}\right)$ is the uniform distribution on $n$ bits.
\end{definition}

We note that the total variation distance lower bounds the trace distance $T(\rho(C), \id/2^n)$ from the maximally mixed state, that is 
\begin{equation}
        T(\rho(C), \id/2^n) \coloneqq \frac{1}{2} \left\| \rho(C) - \frac{1}{2^n}\id \right\|_1 \geq \frac{1}{2} \left\| {\vec{p}}(C) - \vec{u} \right\|_1 \eqqcolon \delta(C),
\end{equation}
where $\|\cdot\|_1$ denotes the trace norm in the first instance, as proven in Lemma \ref{lem:td-tvd}. 
Therefore, lower bounds on the total variation distance also serve as lower bounds on the trace distance.

\subsection{Preliminary lemmas}

Firstly, we present a simple result connecting the total variation distance to the trace distance.

\begin{lemma}\label{lem:td-tvd}
The probability distribution $\vec{p}\in \mathbb{R}^{2^n}$ in the computational basis corresponding to any quantum state $\rho$ satisfies 
\begin{align}
    \Big\| \rho - \frac{\id}{2^n} \Big\|_1 \geq \big\| {\vec{p}} - \vec{u} \big\|_1
\end{align}
where $\vec{u}$ is the uniform distribution on $n$ bits.
\end{lemma}

\begin{proof}
Denote $\vec{p}=(p_1,...,p_{2^n})$. We can write 
\begin{align}
    \big\| {\vec{p}} - \vec{u} \big\|_1 &= \sum_k \Big| p_k - \frac{1}{2^n} \Big| \\
    &= \sum_{i \in \{0,1\}^n} \left| \Tr\Big[\ketbra{i}{i}\Big(\rho-\frac{\id}{2^n}\Big)\Big] \right|\\
    &\leq \sum_{i \in \{0,1\}^n}\Tr\Big[\ketbra{i}{i}\Big|\rho-\frac{\id}{2^n}\Big|\Big] \\
    &= \Tr\Big[\Big|\rho-\frac{\id}{2^n}\Big|\Big] = \Big\| \rho - \frac{\id}{2^n} \Big\|_1\,.
\end{align}
where in the inequality we have considered the spectral decomposition of the Hermitian operator $(\rho - \id/2^n)$.
\end{proof}

We also present a basic lemma with relevance to output probabilities of quantum circuits.

\begin{lemma}[Concentration of marginal bounds]\label{lem:marginal}
    Denote $p^{(i)}_0(C)$ as the probability of obtaining the computational basis measurement outcome $``0"$ on qubit $i$. 
    Then, for any $i=1,\dots,n$, we have
    \begin{align}
        \delta(C) \geq \bigg(p^{(i)}_0(C) -\frac{1}{2} \bigg)^2\,.
    \end{align}
\end{lemma}
\begin{proof}
    We denote 
    $p^{(i)}_{a}(C) = \sum_{\vec{x}:\, x_i=a} {p_{\vec{x}}}(C)$
    as the marginal probability of obtaining measurement outcome $a\in\{0,1\}$ on qubit $i$, where the sum over ${\vec x} \coloneqq (x_1, \dots, x_{n}) \in \{0,1\}^n$ runs over all $2^n/2$ possible bit strings which have $``a"$ in qubit $i$.
    The statement then follows from general properties of marginal probabilities.
    \begin{align}
        \delta(C) = \frac{1}{2} \sum_{j=1}^n \sum_{x_j=0,1} \left| {p_{\vec{x}}}(C) - \frac{1}{2^n} \right| &\geq \frac{1}{2} \sum_{x_i=0,1} \left| \sum_{j \neq i} \sum_{x_j=0,1} \left( {p_{\vec{x}}}(C) - \frac{1}{2^n} \right) \right| \\
        &= \frac{1}{2} \sum_{x_i=0,1} \left| p^{(i)}_{x_i}(C) - \frac{1}{2} \right| \\ 
        &= \left| p^{(i)}_{0}(C) - \frac{1}{2} \right|\\ &\geq \left( p^{(i)}_{0}(C) - \frac{1}{2} \right)^2\,,
    \end{align}
    where the first inequality is an application of the triangle inequality, the first equality from the definition of the marginal probability $p^{(i)}_{x_i}(C)$, and the remaining equalities from basic properties of marginal probabilities.
\end{proof}

It will also be useful to state a standard result for averaging over unitary 2-designs.

\begin{lemma}\label{lem:two-qubit-unitary}
For $t,q,r,s \in \{\id, X, Y, Z \}^{\otimes 2}$ we have
\begin{align}
    \underset{U\sim\mathbb{U}(4)}{\mathbb{E}} \Tr[tUqU^{\dag}] \Tr[rUsU^{\dag}] &= 0 \quad \text{if}\;\; t\neq r\;\; \text{or}\;\; q\neq s\,, \\
    \underset{U\sim\mathbb{U}(4)}{\mathbb{E}} \Tr[tUqU^{\dag}]^2 &= \begin{cases}
        1 \quad \text{if}\;\; t=q=\id^{\otimes 2}/4\,, \\
        0 \quad \text{if}\;\; q=\id^{\otimes 2}/4\,,\; t\neq q\,,\\
        0 \quad \text{if}\;\; t=\id^{\otimes 2}/4\,,\; q\neq t\,,\\
        1/15 \quad \text{if}\;\; t,q \neq \id^{\otimes 2}/4\,.
    \end{cases}
\end{align}
\end{lemma}
Proving Lemma~\ref{lem:two-qubit-unitary} involves standard identities of 2-designs~\cite{dankert2009exact} and we refer to~\cite{harrow2009random} for an explicit proof.

\subsection{Proof of Theorem~2}

For circuit architectures falling under Definition~\ref{def:CA}, a lower bound on the average total variation distance was obtained by Deshpande et al.~in Ref.~\cite{deshpande2022tight} in the case of uniform local Pauli noise throughout all registers, i.e.~$\PC_c = \PC_{b} = \PC_d$. 
The bound was subsequently improved by Aharonov et al.~in Ref.~\cite{aharonov2022polynomial} explicitly for local depolarizing noise (though we show here that it generalizes simply for Pauli noise) and we quote the result of Ref.~\cite{aharonov2022polynomial} here.

\begin{theorem}[Convergence of noisy random circuits to uniform (Ref.~\cite{aharonov2022polynomial}, Theorem 6)]
    For circuits drawn from the architecture of Definition~\ref{def:CA}, where the local depolarizing error rate  is uniform at $\varepsilon_d$, i.e.~$\DC_c = \DC_d = \DC_{b}$, we have
    \begin{align}
        \underset{C\sim \YC}{\mathbb{E}}[{\delta(C)}] \geq \frac{1}{12} \cdot\left(\frac{2}{5}\right)^L \cdot(1-\varepsilon_d)^{2L}.
    \end{align}
\end{theorem}

We generalize this result for non-uniform local Pauli noise in the following statement.
\begin{proposition}\label{prop:nonuniform}
    For circuits $C$ drawn from the architecture of Definition~\ref{def:CA}, with non-uniform Pauli noise instances $\PC_{i \ell}$ on qubit $i$ at layer $\ell$, we have 
    \begin{align}\label{eq:prop-eq}
        \underset{C\sim \YC}{\mathbb{E}} [{\delta(C)}] \geq \frac{1}{12} \cdot \left(\frac{2}{5}\right)^L \cdot \prod_{i=1}^n \prod_{\ell=1}^L (1-\varepsilon_{i \ell})^{\frac{2}{n}}\,,
    \end{align}
    where $\YC$ is the ensemble of two-qubit brick-layered circuits such that each brick forms a local 2-design, and $\varepsilon_{i\ell}$ is the error rate of $\PC_{i\ell}$.
\end{proposition}
\begin{proof}
    For a random circuit $C$ drawn from the distribution described in Definition~\ref{def:CA}, denote the unitary layers by $\{\UC_\ell\}$ and the local noise layers as $\{\N_\ell=\bigotimes_i \PC_{i \ell}\}$ composed of  Pauli noise channels $\PC_{i \ell}$ acting on qubit $i$ for layer $\ell = 1, \dots, L$. 
    Further, let $\vec{s} \coloneqq (s_0, s_1, \dots, s_L)$ denote a list of Pauli tensor-product strings corresponding to the trajectory of Pauli strings through the circuit, where $s_0$ denotes the input Pauli string and $s_{\ell} \in \mathrm{P}_n$ for $\ell \geq 0$ denotes the Pauli string output of layer $\ell$.
    We follow the proof of Ref.~\cite{aharonov2022polynomial} and write the probability of measuring the bit string $\vec{x}\in \{0,1\}^n$ with the circuit $C$ as a sum over all possible $4^{n(L+1)}$ trajectories of Pauli strings,
    \begin{align}
        p_{\vec{x}}(C) &= \Tr\big[\ketbra{\vec{x}}{\vec{x}} \NC_L\circ\UC_L \circ \cdots   \circ \NC_1\circ\UC_1(  \ketbra{0^n}{0^n})\big] \\
        &=\sum_{\vec{s}\in \mathrm{P}^{L+1}_n} g(C,\vec{s}, \vec{x})\,,\label{eq:fourier}
    \end{align}
    where $g(C,\vec{s}, \vec{x})$ is the probability of a trajectory $\vec{s}$ occurring and resulting in the output $\vec{x}$. Explicitly, we can write
    \begin{align}
        g(C,\vec{s}, \vec{x}) 
        &=  \frac{1}{2^{n(L+1)}} \Tr\big[\ketbra{\vec{x}}{\vec{x}} s_L\big]\Tr\big[s_L \NC_L\circ\UC_L(s_{L-1})\big] \cdots  \Tr\big[s_1 \NC_1\circ\UC_1(s_0)\big] \Tr\big[s_0 \ketbra{0^n}{0^n}\big]\\
        &=  \frac{1}{2^{n(L+1)}} \Tr\big[\ketbra{\vec{x}}{\vec{x}} s_L\big]\Tr\big[\NC_L(s_L)\UC_L(s_{L-1})\big] \cdots  \Tr\big[\NC_1(s_1) \UC_1(s_0)\big] \Tr\big[s_0 \ketbra{0^n}{0^n}\big]\\
        &=  \frac{1}{2^{n(L+1)}} \chi(\vec{s})\, \Tr\big[\ketbra{\vec{x}}{\vec{x}} s_L\big]\Tr\big[s_L \UC_d(s_{L-1})\big] \cdots  \Tr\big[s_1 \UC_1(s_0)\big] \Tr\big[s_0 \ketbra{0^n}{0^n}\big]\,, \label{eq:trajectory}
    \end{align}
    where in the second line we have used the fact that the noise channel is self-adjoint and in the third line that it maps Pauli strings to themselves, i.e. $\NC_{\ell}(s_{\ell}) \propto s_{\ell}$ for $\ell = 1, \dots, L$, with a signed damping factor according to Definition \ref{def:PauliC}, which we have factored out to obtain the total damping factor $\chi(\vec{s})$ of the trajectory $\vec{s}$ due to noise.
    In particular,
    \begin{align}\label{eq:chi}
        |\chi(\vec{s})| &= \prod_{\ell=1}^L \big|\Tr[s_{\ell} \NC_{\ell}(s_{\ell})] \big|  \\
        &\geq  \prod_{\ell=1}^L \prod_{i\in \mathrm{H}(s_{\ell})} (1-\varepsilon_{i \ell}),
    \end{align}
    where $\mathrm{H}(s_{\ell})$ is the qubit index set of all qubits with non-identity Pauli elements in the Pauli string $s_{\ell}$, and we recall that $\varepsilon_{i \ell}$ denotes largest error rate for Pauli channel $\PC_{i \ell}$. Thus, $\chi(\vec{s})$ contains at worst a damping factor of $(1-\varepsilon_{i \ell})$ for each non-identity element of each Pauli string $s_{\ell}$ along the trajectory.
    
    We now consider the marginal distribution on each qubit. Recalling that we denote the probability of obtaining computational basis measurement result $``0"$ on qubit $i$ as $p_0^{(i)}$, Lemma \ref{lem:marginal} implies that
    \begin{align}\label{eq:TVD-bound}
        \delta(C) \geq \frac{1}{n} \sum_{i=1}^n \bigg( p_0^{(i)}(C) - \frac{1}{2} \bigg)^2\,.
    \end{align}
    Each $p_0^{(i)}(C)$ is made up of trajectories of the form $g(C,\vec{s}, 0_i)$ given by Eq.~(\ref{eq:trajectory}) with measurement operator $\ketbra{0}{0}_i$ (zero state projector on qubit $i$). 
    Taking expectation values over the circuit ensemble $\YC$ of Definition~\ref{def:CA}, we get
    \begin{align}\label{eq:squared-marginal-bound}
        \underset{C\sim \YC}{\mathbb{E}}[\delta(C)] 
        &\geq \frac{1}{n} \sum_{i=1}^n \underset{C\sim \YC}{\mathbb{E}} \bigg( p_0^{(i)}(C) - \frac{1}{2} \bigg)^2 \\
        &= \frac{1}{n} \sum_{i=1}^n \underset{C\sim \YC}{\mathbb{E}} \bigg( \sum_{\substack{\vec{s}\in \mathrm{P}^{L+1}_n, \\ (s_L)_i = Z}} g(C,\vec{s}, 0_i) - \frac{1}{2} \bigg)^2\\
        &= \frac{1}{n} \sum_{i=1}^n \underset{C\sim \YC}{\mathbb{E}} \bigg( \sum_{\substack{\vec{s}:|\vec{s}|>0, \\ (s_L)_i = Z}} g(C,\vec{s}, 0_i) \bigg)^2\\
        &\geq \frac{1}{n} \sum_{i=1}^n \sum_{\substack{\vec{s}:|\vec{s}|=L+1, \\ (s_L)_i = Z}}  \underset{C\sim \YC}{\mathbb{E}}\, g(C,\vec{s}, 0_i)^2 \,,\label{eq:average_expectation}
    \end{align}
    where in the second line we have used Eq.~(\ref{eq:fourier}) with $|\vec{s}|\coloneqq |s_0| + |s_1| + ... + |s_L|$ denoting the sum of weights of trajectories ending in a Pauli $Z$ on qubit $i$. 
    Then, in the third line, we have evaluated the trajectory corresponding to the identity Pauli string $\id^{\otimes n}$ corresponding to $|\vec{s}|=0$, and in the fourth line we have discarded all terms except those that correspond to the squares of $g(C,\vec{s}, 0_i)$ with the smallest weight $|\vec{s}| = L + 1$, i.e.~$|s_0| = |s_1| = ... = |s_L|=1$, where the cross-terms can be dropped because $g(C,\vec{s}, 0_i)$ are probabilities and thus positive.
    
    Lemma \ref{lem:two-qubit-unitary} and Eq.~(\ref{eq:trajectory}) further imply that 
    \begin{align}
        \underset{C\sim \YC}{\mathbb{E}} \, g(C,\vec{s}, 0_i)^2 &= \frac{1}{2^{2n(L+1)}} \chi(\vec{s})^2 \Tr\big[\ketbra{0}{0}_i s_L\big]^2 \cdot \left(\frac{2^{2n}}{15}\right)^L \cdot \Tr\big[s_0 \ketbra{0^n}{0^n}\big]^2\\
        &= \frac{1}{2^{2n(L+1)}}  \chi(\vec{s})^2 \, 2^{2n-2}  \cdot\left(\frac{2^{2n}}{15}\right)^L \cdot 1\\ 
        &= \frac{\chi(\vec{s})^2}{4\cdot15^L}\,,
    \end{align}
    for all $|\vec{s}|=L+1, (s_L)_i = Z$. 
    With careful inspection of Lemma \ref{lem:two-qubit-unitary}, $|\vec{s}|=L+1$ implies that each $|s_i|=1$, i.e.~such trajectories can only contain Pauli strings of weight 1 mapped to Pauli strings of weight 1. For a fixed qubit register $i$, there are $2^L\cdot3^{L-1}$ such trajectories so that
    \begin{equation}\label{eq:sum-of-terms}
        \sum_{\substack{\vec{s}:|\vec{s}|=L+1, \\ (s_L)_i = Z}} 1 = 2^L\cdot3^{L-1},
    \end{equation}
    
    Recall each trajectory contributes a total damping factor of the form $\chi(\vec{s}) \geq \prod_{\ell=1}^L (1-\varepsilon_{i \ell})$ where now $i$ indicates the unique register of the one non-identity Pauli of trajectory $\vec{s}$ at layer $\ell$ according to Eq.~(\ref{eq:chi}). 
    Combining the above enables us to write
    \begin{align}
        \underset{C\sim \YC}{\mathbb{E}}[\delta(C)] &\geq \frac{1}{n} \sum_{i=1}^n \underset{C\sim \YC}{\mathbb{E}} \sum_{\substack{\vec{s}:|\vec{s}|=L+1, \\ (s_L)_i = Z}} g(C,\vec{s}, 0_i)^2 \\
        &= \frac{1}{n} \frac{1}{4\cdot15^L}\sum_{i=1}^n  \sum_{\substack{\vec{s}:|\vec{s}|=L+1, \\ (s_L)_i = Z}} \chi(\vec{s})^2 \\
        &\geq  \frac{2^L\cdot3^{L-1}}{4\cdot15^L}  \Bigg(\prod_{i=1}^n\prod_{\substack{\vec{s}:|\vec{s}|=L+1, \\ (s_L)_i = Z}} \chi(\vec{s})^2 \Bigg)^{1/\left( 2^L 3^{L-1} n \right)} \label{eq:am-gm}\\
        &\geq  \frac{2^L\cdot3^{L-1}}{4\cdot15^L}  \Bigg(\prod_{i=1}^n\prod_{\ell=1}^L (1-\varepsilon_{i \ell})^{2\cdot2^L 3^{L-1}} \Bigg)^{1/\left( 2^L 3^{L-1} n \right)} \label{eq:depVpauli}\\
        &= \frac{1}{12}\left( \frac{2}{5} \right)^L \prod_{i=1}^n \prod_{\ell=1}^L (1-\varepsilon_{i \ell})^{\frac{2}{n}}.
    \end{align}    
    The first inequality is a reproduction of Eq.~(\ref{eq:average_expectation}). 
    The second inequality is an application of the arithmetic-geometric means inequality, where we have $2^L3^{L-1}n$ terms due to Eq.~(\ref{eq:sum-of-terms}). We note that equality holds in the special case when the error rate $\varepsilon$ is uniform throughout all registers and layers. The third inequality comes from our bound on the total damping factor, and from observing that all Pauli trajectories of weight $L+1$ appear an equal number of times in the product due to the symmetry of the circuit (periodic boundary conditions). 
\end{proof}

In the main text we are mostly interested in the particular case of circuits with registers that are clean (low error rate), if they correspond to a logical qubit that undergoes a QEC cycle; or noisy (high error rate) if they correspond to a physical qubit without error correction. In addition, we also consider an even higher error rate each time a noisy qubit is coupled to a clean qubit. 
We let $n_c$ and $n_d$ be the number of clean and noisy registers respectively. We presume all registers experience local depolarizing noise, where each layer of depolarizing noise consists of channels with different error rates depending on the register: $\varepsilon_c$ on the clean registers, $\varepsilon_{b}$ on the noisy register following a clean-noisy coupling, $\varepsilon_{b'}$ on the clean register following a clean-noisy coupling and $\varepsilon_d$ on the remaining noisy registers. 
Proposition~\ref{prop:nonuniform} immediately implies our main result (Theorem~2) in the context of the clean-noisy setup.

\begin{corollary}[Generalized form of Theorem 2 in the main text]
    Consider the $L$-layered brick-layered circuit with periodic boundary conditions in Fig.~3 (main text) on an even number of qubits $n$, under the local Pauli noise model of Eq.~(\ref{def:PauliC}), characterized by four error rates $\varepsilon_c, \varepsilon_d, \varepsilon_b , \varepsilon_{b'}$.  
    For any even $L$ the output state $\rho(C)$ satisfies 
    \begin{align}\label{eq:lower-bound-appdx}
        \underset{C\sim \YC}{\mathbb{E}} &\Big[T\big(\rho(C), \frac{\id}{2^n}\big)\Big] \geq \underset{C\sim \YC}{\mathbb{E}}\Big[\delta\big(\rho(C), \frac{\id}{2^n}\big)\Big] \geq \nonumber \\
        &\geq \frac{1}{12} \left(\frac{2}{5}\right)^L (1-\varepsilon_{c})^{ 2Lf_c} (1-\varepsilon_{d})^{ 2Lf_d} (1-\varepsilon_{b})^{2Lf_b} (1-\varepsilon_{b'})^{2Lf_{b'}}\,,
    \end{align}
    where $\YC$ is the ensemble of brick-layered circuits composed of local 2-designs, and fractions $f_c, f_d, f_b, f_{b'}$ are given by
    \begin{align}
        f_c=\frac{n_c}{n}-\frac{n_b}{2n},\; f_d=\frac{n_d}{n}-\frac{n_b}{2n},\; f_b= f_{b'}=\frac{n_b}{2n} \,,
    \end{align}
    such that $f_c + f_d + f_b + f_{b'}= 1$, where $n_c, n_d$ are the respective numbers of error-corrected qubits and noisy qubits satisfying $n_c+n_d=n$, and $n_b$ is the number of boundaries between the noisy and error-corrected qubits.
\end{corollary}
\begin{proof}
    Eq.~(\ref{eq:prop-eq}) of Proposition \ref{prop:nonuniform}, states that there is factor of $(1-\varepsilon_{i\ell})^{\frac{2}{n}}$ for each instance of noise with error rate $\varepsilon_{i\ell}$. Inspecting Fig.~3 in the main text, we see that in every $2$ layers, there are $2n_d-n_b$ instances of noise with error rate $\varepsilon_d$, $2n_c-n_b$ instances of noise with rate $\varepsilon_c$, $n_b$ instances of noise with rate $\varepsilon_b$, and $n_b$ instances of noise with rate $\varepsilon_{b'}$. The result then comes directly by substituting in these damping factors. 
\end{proof}

\subsection{Generalization to correlated Pauli noise}

We can further generalize Proposition~\ref{prop:nonuniform} for arbitrary (correlated) Pauli noise as defined in Definition~\ref{def:PauliCorr}.
\begin{proposition}\label{prop:nonuniform2}
    Consider circuits $C$ drawn from the architecture of Definition~\ref{def:CA}, with layers of correlated Pauli noise instances $\bigotimes_{i=1}^{k(\ell)} \PC_{i \ell}$ at each layer $\ell$ (see Definition \ref{def:PauliCorr}). Denote the number of qubits that $\PC_{i\ell}$ acts on as $m(i, \ell)$. We have 
    \begin{align}\label{eq:prop-eq-corr}
        \underset{C\sim \YC}{\mathbb{E}} [{\delta(C)}] \geq \frac{1}{12} \cdot \left(\frac{2}{5}\right)^L \cdot \prod_{i=1}^{k(\ell)} \prod_{\ell=1}^L (1-\varepsilon_{i \ell})^{\frac{2}{n}m(i,\ell)}\,,
    \end{align}
    where $\YC$ is the ensemble of two-qubit brick-layered circuits such that each brick forms a local 2-design, and $\varepsilon_{i\ell}$ is the error rate of noise channel $\PC_{m(i,\ell),\ell}$.
\end{proposition}
\begin{proof}
    Closely following the proof of Proposition \ref{prop:nonuniform}, Pauli trajectories have damping factor that satisfy
    \begin{align}\label{eq:chi-m}
        |\chi(\vec{s})| &= \prod_{\ell=1}^L \big|\Tr[s_{\ell} \NC_{\ell}(s_{\ell})] \big|   \geq  \prod_{\ell=1}^L \prod_{i\in \mathrm{H}(s_{\ell})} (1-\varepsilon_{j(i,\ell) \ell}) ,
    \end{align}
    where, differently to Eq.~(\ref{eq:chi}), here $j(i,\ell)$ is the function that returns the index the noise channel that affects qubit $i$ in layer $\ell$. As in the proof of Proposition \ref{prop:nonuniform}, we only need to consider the set of Pauli trajectories with one non-identity Pauli -- for those trajectories, the sum over $i$ only contains one element. Considering only those trajectories, the product of damping factors over $i$ amounts to
    \begin{align}
        \prod_{i=1}^n |\chi(\vec{s})| & \geq  \prod_{\ell=1}^L \prod_{i=1}^n (1-\varepsilon_{j(i,\ell) \ell})\,.
    \end{align}
    We can equivalently write $\prod_{i=1}^n |\chi(\vec{s})| \geq  \prod_{i=1}^{k(\ell)}\prod_{\ell=1}^L (1-\varepsilon_{i \ell})^{\frac{2}{n}m(i,\ell)}$ by relabeling $i$ as the index that counts over noise channels. With this new expression, the proof follows as the proof of Proposition \ref{prop:nonuniform}.
\end{proof}

The noise channel considered in Proposition \ref{prop:nonuniform2} and accompanying result in Eq.(\ref{eq:prop-eq-corr}) is more general than the one considered in  Proposition \ref{prop:nonuniform} -- that is, it includes single-qubit Pauli noise as the special case when $m(i,\ell) = 1$ for all $i$ and $\ell$. Importantly, it also includes models of correlated noise such as 2-qubit Pauli noise interactions.

\section{Quantum error correction implementation for the Steane code}
\label{app:QEC_impl}

The Steane code encodes logical $\ket{\overline{0}}, \ket{\overline{1}}$ states using seven physical qubits
\begin{align}
\ket{\overline{0}} &=  \frac{1}{\sqrt{8}}\big( \ket{0000000} + \ket{1010101} + \ket{0110011} + \ket{1100110} + \ket{0001111} + \ket{1011010} +\ket{0111100} + \ket{1101001}  \big)\,,
\end{align}
\begin{align}
\ket{\overline{1}} &=  \frac{1}{\sqrt{8}}\big( \ket{1111111} + \ket{0101010} + \ket{1001100}+\ket{0011001} + \ket{1110000} + \ket{0100101} +\ket{1000011} + \ket{0010110}  \big)\,.
\end{align}
In this section, we denote these physical qubits as $q_0$, $q_1$, $q_2$, $q_3$, $q_4$, $q_5$, $q_6$.

We prepare $\ket{\overline{0}}$ using a non-fault-tolerant circuit from Fig.~\ref{fig:circ_0log}. This circuit involves an ancilla, whose measurement is used to detect errors that happen during the circuit execution --- when the measurement outcome is $1$, the prepared state is discarded.  This $\ket{\overline{0}}$ state preparation method was demonstrated with real-world hardware~\cite{ryan2021realization}. 

\begin{figure}[t]
    \centering 
    \includegraphics[width=0.4\columnwidth]{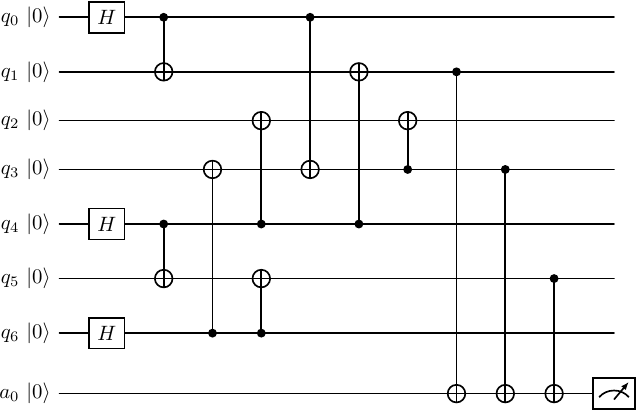}
    \caption{ {\bf A circuit for $\ket{\overline{0}}$  preparation. } The circuit prepares $\ket{\overline{0}}$ state using 7 physical qubits ($q_0$, $q_1$, $q_2$, $q_3$, $q_4$, $q_5$, $q_6$) and an ancilla $a_0$. If the ancilla measurement gives a $0$ outcome, the state is accepted. Otherwise, it is discarded.    }
    \label{fig:circ_0log}
\end{figure}

Logical single-qubit Clifford gates that are used here are implemented 
with physical Clifford gates as
 \begin{equation}
\overline{X} =  \prod_{i=0}^6 X_{i} \,, \quad
\overline{Y} = \prod_{i=0}^6 Y_{i} \,, \quad
\overline{Z} = \prod_{i=0}^6 Z_i,   , \quad
\overline{S} = \prod_{i=0}^6 S_i^{\dag} \,, \quad 
\overline{S}^{\dag} = \prod_{i=0}^6 S_i  \,, \quad 
\overline{H} = \prod_{i=0}^6 H_{i} \,. \label{eq:Z_impl0}
\end{equation}
Here, we assume that a physical gate with an index  $i \in \{0,1,\dots,6 \}$ acts at a qubit $q_i$. We show implementations of logical $CNOT$ gates involving encoded qubits in Fig.~\ref{fig:circ_CNOTcc} and Fig.~\ref{fig:circ_CNOTcd}.

\begin{figure}[t]
    \centering 
    \includegraphics[width=0.33\columnwidth]{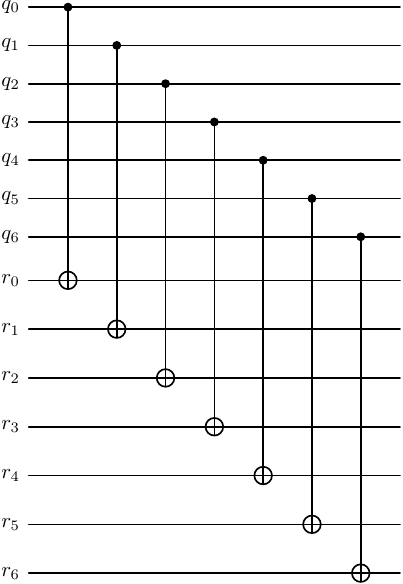}
     \caption{ {\bf A $\overline{CNOT}$ circuit. }  A circuit implementing a logical $CNOT$ gate between two encoded qubits that are built of physical qubits $q_0,q_1.q_2,q_3,q_4,q_5,q_6$ and  $r_0,r_1.r_2,r_3,r_4,r_5,r_6$. The $CNOT$ is controlled at a logical qubit encoded with  $q_i$ physical qubits. }
    \label{fig:circ_CNOTcc}
\end{figure}

\begin{figure}[t]
    \centering 
    \includegraphics[width=0.45\columnwidth]{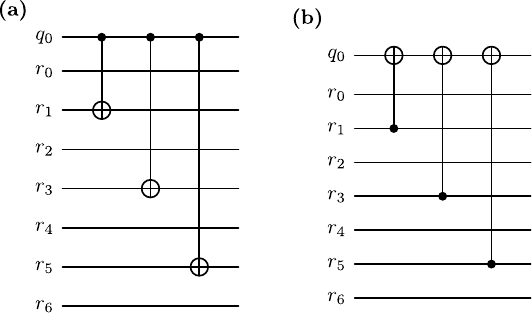}
    \caption{ {\bf $\widetilde{CNOT}$ circuits. }  Circuits realizing logical $CNOT$ gates between a noisy qubit $q_0$ and an encoded one with physical registers $r_0, r_1, r_2, r_3, r_4, r_5, r_6$. In {\bf (a)}, the logical gate is controlled at the noisy qubit,  while in {\bf (b)}, the control qubit is the encoded one. }
    \label{fig:circ_CNOTcd}
\end{figure}

We follow the logical gates acting at the error-corrected qubits by single QEC rounds applied to these qubits. 
We perform a QEC round by applying syndrome circuits to an encoded qubit. First, we execute a flagged syndrome circuit $S_1^f$ from Fig.~\ref{fig:circ_syndf}. The circuit measures three syndromes at the same time by measuring three ancilla qubits. If any of these measurements gives an outcome other than $0$, we proceed with six unflagged syndrome circuits $S_1, S_2, S_3, S_4, S_5, S_6$ from Fig.~\ref{fig:circ_synduf}. Each of them measures a single syndrome. We decode an error using the results of the flagged and the unflagged syndrome circuits, as detailed in Table~\ref{tab:dec_sf1}, and apply a recovery circuit specified in Table~\ref{tab:dec_sf1}. When the obtained syndromes are not consistent with a single-qubit error happening before or during the syndrome circuit execution, we discard the run since the Steane code is not capable of correction for two-qubit and higher-weight errors. Therefore, we post-select based on syndrome outcomes to reject runs that are likely irrecoverably corrupted by errors.

Next, we apply a flagged syndrome circuit $S_2^f$ from Fig.~\ref{fig:circ_syndf}. Again, if a measurement of the ancilla qubits gives an outcome other than all-zeros, we follow that with the unflagged circuits $S_1, S_2, S_3, S_4, S_5, S_6$. We decode an error and a recovery operation according to Table~\ref{tab:dec_sf2}. If the syndrome outcome is inconsistent with a single-qubit error that occurred before or during $S_2^f, S_1^u, S_2^u, S_3^u, S_4^u, S_5^u, S_6^u$ execution, we discard a run. 
This implementation of a QEC round is partially fault-tolerant, and adapts a QEC real-world implementation~\cite{ryan2021realization} to our partial error correction setup.

 \begin{figure}[t]
    \centering 
    \includegraphics[width=\columnwidth]{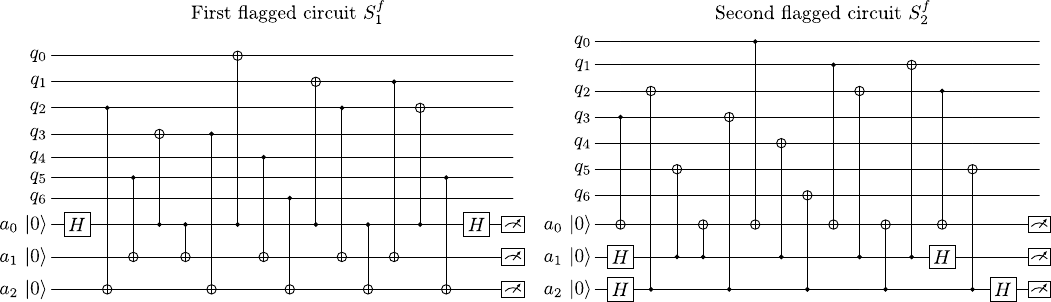}
    \caption{  {\bf Flagged syndrome circuits.} Circuits used for syndrome detection in a flagged manner during a QEC round. See details in the text.  }
    \label{fig:circ_syndf}
\end{figure}

 \begin{figure}[t]
    \centering 
    \includegraphics[width=\columnwidth]{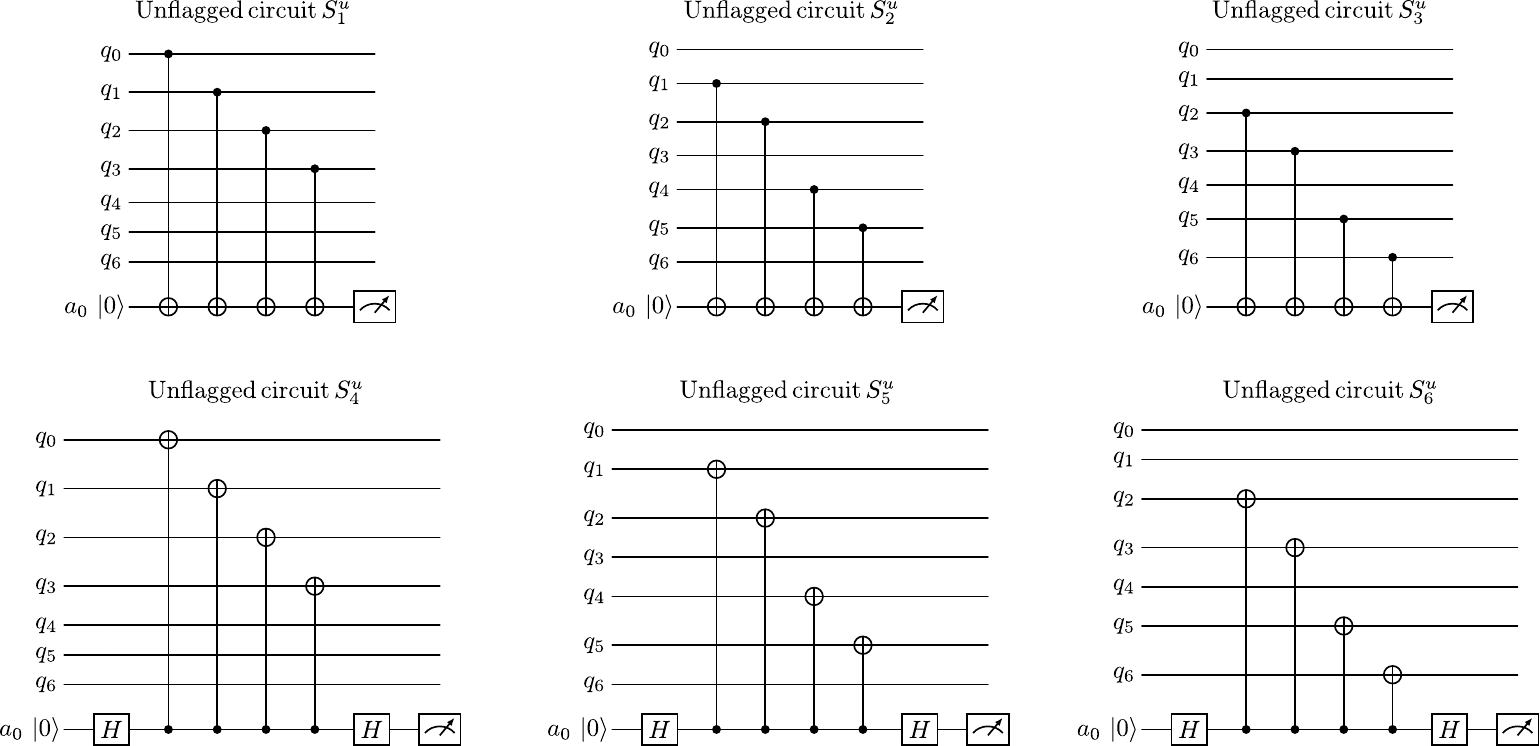}
    \caption{ {\bf Unflagged syndrome circuits.} Circuits used for syndrome detection without flag qubits during a QEC round. See details in the text. 
}
    \label{fig:circ_synduf}
\end{figure}

\begin{table}[p]
\begin{tabular}{|c|c|c|}
\hline
$(a_0,a_1,a_2)$ & $(s_1,s_2,s_3,s_4,s_5,s_6)$   & a recovery operation  \\
\hline
$(1, 0, 0)$ & $(1, 0, 0, 1, 0, 0)$ & $Y_0$ \\
$(1, 0, 0)$ & $(0, 0, 0, 1, 0, 0)$ & $Z_0$ \\
$(0, 1, 0)$ & $(1, 1, 0, 0, 0, 0)$ & $X_1$ \\
$(1, 1, 0)$ & $(1, 1, 0, 1, 1, 0)$ & $Y_1$ \\
$(1, 0, 0)$ & $(0, 0, 0, 1, 1, 0)$ & $Z_1$ \\
$(0, 1, 1)$ & $(1, 1, 1, 0, 0, 0)$ & $X_2$ \\
$(1, 1, 1)$ & $(1, 1, 1, 1, 1, 1)$ & $Y_2$ \\
$(1, 0, 0)$ & $(0, 0, 0, 1, 1, 1)$ & $Z_2$ \\
$(0, 0, 1)$ & $(1, 0, 1, 0, 0, 0)$ & $X_3$ \\
$(1, 0, 1)$ & $(1, 0, 1, 1, 0, 1)$ & $Y_3$ \\
$(1, 0, 0)$ & $(0, 0, 0, 1, 0, 1)$ & $Z_3$ \\
$(0, 1, 0)$ & $(0, 1, 0, 0, 0, 0)$ & $X_4$ \\
$(0, 1, 0)$ & $(0, 1, 0, 0, 1, 0)$ & $Y_4$ \\
$(0, 1, 1)$ & $(0, 1, 1, 0, 0, 0)$ & $X_5$ \\
$(0, 1, 1)$ & $(0, 1, 1, 0, 1, 1)$ & $Y_5$ \\
$(0, 0, 1)$ & $(0, 0, 1, 0, 0, 0)$ & $X_6$ \\
$(0, 0, 1)$ & $(0, 0, 1, 0, 0, 1)$ & $Y_6$ \\
$(1, 0, 0)$ & $(0, 0, 0, 0, 0, 0)$ & error not detected \\
$(0, 1, 0)$ & $(0, 0, 0, 0, 0, 0)$ & error not detected \\
$(0, 0, 1)$ & $(0, 0, 0, 0, 0, 0)$ & error not detected \\
$(1, 0, 1)$ & $(1, 0, 1, 0, 0, 0)$ & $X_3$ \\
$(0, 0, 1)$ & $(1, 0, 1, 1, 0, 1)$ & $Y_3$ \\
$(0, 1, 0)$ & $(1, 1, 1, 0, 0, 0)$ & $X_2$ \\
$(1, 1, 0)$ & $(1, 1, 1, 1, 1, 1)$ & $Y_2$ \\
$(1, 0, 1)$ & $(0, 0, 0, 1, 1, 1)$ & $Z_2$ \\
$(0, 0, 1)$ & $(0, 1, 1, 0, 0, 0)$ & $X_5$ \\
$(0, 0, 1)$ & $(0, 1, 1, 0, 1, 1)$ & $Y_5$ \\
$(0, 1, 0)$ & $(0, 0, 0, 0, 1, 1)$ & $Z_5$ \\
$(0, 1, 1)$ & $(1, 0, 1, 0, 0, 0)$ & $X_3$ \\
$(1, 1, 1)$ & $(1, 0, 1, 0, 0, 0)$ & $X_3$ \\
$(1, 1, 0)$ & $(0, 0, 0, 0, 1, 1)$ & $Z_5$ \\
$(1, 0, 0)$ & $(0, 0, 0, 0, 1, 1)$ & $Z_5$ \\
$(0, 1, 1)$ & $(0, 0, 1, 0, 0, 0)$ & $X_0 X_3$ \\
$(1, 1, 1)$ & $(0, 0, 1, 0, 0, 0)$ & $X_0 X_3$ \\
$(1, 0, 1)$ & $(0, 0, 0, 0, 1, 0)$ & $Z_0 Z_1$ \\
$(1, 0, 0)$ & $(0, 0, 0, 0, 1, 0)$ & $Z_0 Z_1$ \\
$(1, 1, 0)$ & $(0, 0, 0, 0, 0, 1)$ & $Z_0 Z_3$ \\
$(1, 0, 0)$ & $(0, 0, 0, 0, 0, 1)$ & $Z_0 Z_3$ \\
$(0, 0, 1)$ & $(1, 1, 1, 0, 0, 0)$ & $X_2$ \\
$(1, 0, 1)$ & $(1, 1, 1, 0, 0, 0)$ & $X_2$ \\
$(0, 1, 0)$ & $(1, 1, 0, 1, 1, 0)$ & $Y_1$ \\
$(1, 0, 1)$ & $(0, 0, 0, 0, 1, 1)$ & $Z_5$ \\
$(1, 0, 0)$ & $(1, 1, 1, 1, 1, 1)$ & $Y_2$ \\
$(0, 1, 0)$ & $(0, 0, 0, 1, 1, 0)$ & $Z_1$ \\
$(1, 0, 0)$ & $(1, 1, 1, 0, 0, 0)$ & $X_2$ \\
$(0, 0, 1)$ & $(0, 0, 0, 0, 1, 1)$ & $Z_5$ \\
\hline
\end{tabular}
\caption{{\bf A decoding for erroneous $S_1^f$.}  A recovery operation to be applied after an erroneous outcome of $S_1^f$ measurement. Here, $(a_0,a_1,a_2)$ are outcomes of the ancilla qubits $a_0,a_1,a_2$ measurements from a circuit $S_1^f$ shown in Fig.~\ref{fig:circ_syndf}. 
$(s_1,s_2,s_3,s_4,s_5,s_6)$  are outcomes of the following circuits $S_1^u,S_2^u,S_3^u,S_4^u,S_5^u,S_6^u$ measurements from Fig.~\ref{fig:circ_synduf}. In the case of any combination of $(a_0,a_1,a_2) \neq (0,0,0)$ and $(_1,s_2,s_3,s_4,s_5,s_6)$  not listed in the table, a run is discarded as likely irrecoverably corrupted. } 
\label{tab:dec_sf1}
\end{table}

\begin{table}[p]
\begin{tabular}{|c|c|c|}
\hline
$(a_0,a_1,a_2)$ & $(_1,s_2,s_3,s_4,s_5,s_6)$   & a recovery operation  \\
\hline
$(1, 0, 0)$ & $(1, 0, 0, 0, 0, 0)$ & $X_0$ \\
$(1, 0, 0)$ & $(1, 0, 0, 1, 0, 0)$ & $Y_0$ \\
$(1, 0, 0)$ & $(1, 1, 0, 0, 0, 0)$ & $X_1$ \\
$(1, 1, 0)$ & $(1, 1, 0, 1, 1, 0)$ & $Y_1$ \\
$(0, 1, 0)$ & $(0, 0, 0, 1, 1, 0)$ & $Z_1$ \\
$(1, 0, 0)$ & $(1, 1, 1, 0, 0, 0)$ & $X_2$ \\
$(1, 1, 1)$ & $(1, 1, 1, 1, 1, 1)$ & $Y_2$ \\
$(0, 1, 1)$ & $(0, 0, 0, 1, 1, 1)$ & $Z_2$ \\
$(1, 0, 0)$ & $(1, 0, 1, 0, 0, 0)$ & $X_3$ \\
$(1, 0, 1)$ & $(1, 0, 1, 1, 0, 1)$ & $Y_3$ \\
$(0, 0, 1)$ & $(0, 0, 0, 1, 0, 1)$ & $Z_3$ \\
$(0, 1, 0)$ & $(0, 1, 0, 0, 1, 0)$ & $Y_4$ \\
$(0, 1, 0)$ & $(0, 0, 0, 0, 1, 0)$ & $Z_4$ \\
$(0, 1, 1)$ & $(0, 1, 1, 0, 1, 1)$ & $Y_5$ \\
$(0, 1, 1)$ & $(0, 0, 0, 0, 1, 1)$ & $Z_5$ \\
$(0, 0, 1)$ & $(0, 0, 1, 0, 0, 1)$ & $Y_6$ \\
$(0, 0, 1)$ & $(0, 0, 0, 0, 0, 1)$ & $Z_6$ \\
$(1, 0, 0)$ & $(0, 0, 0, 0, 0, 0)$ & error not detected \\
$(0, 1, 0)$ & $(0, 0, 0, 0, 0, 0)$ & error not detected \\
$(0, 0, 1)$ & $(0, 0, 0, 0, 0, 0)$ & error not detected \\
$(0, 0, 1)$ & $(1, 0, 1, 1, 0, 1)$ & $Y_3$ \\
$(1, 0, 1)$ & $(0, 0, 0, 1, 0, 1)$ & $Z_3$ \\
$(1, 0, 1)$ & $(1, 1, 1, 0, 0, 0)$ & $X_2$ \\
$(1, 1, 0)$ & $(1, 1, 1, 1, 1, 1)$ & $Y_2$ \\
$(0, 1, 0)$ & $(0, 0, 0, 1, 1, 1)$ & $Z_2$ \\
$(0, 1, 0)$ & $(0, 1, 1, 0, 0, 0)$ & $X_5$ \\
$(0, 0, 1)$ & $(0, 1, 1, 0, 1, 1)$ & $Y_5$ \\
$(0, 0, 1)$ & $(0, 0, 0, 0, 1, 1)$ & $Z_5$ \\
$(1, 0, 0)$ & $(0, 1, 1, 0, 0, 0)$ & $X_5$ \\
$(1, 1, 0)$ & $(0, 1, 1, 0, 0, 0)$ & $X_5$ \\
$(1, 1, 1)$ & $(0, 0, 0, 1, 0, 1)$ & $Z_3$ \\
$(0, 1, 1)$ & $(0, 0, 0, 1, 0, 1)$ & $Z_3$ \\
$(1, 1, 1)$ & $(0, 0, 0, 0, 0, 1)$ & $Z_0 Z_3$ \\
$(0, 1, 1)$ & $(0, 0, 0, 0, 0, 1)$ & $Z_0 Z_3$ \\
$(1, 0, 0)$ & $(0, 1, 0, 0, 0, 0)$ & $X_0 X_1$ \\
$(1, 0, 1)$ & $(0, 1, 0, 0, 0, 0)$ & $X_0 X_1$ \\
$(1, 0, 0)$ & $(0, 0, 1, 0, 0, 0)$ & $X_0 X_3$ \\
$(1, 1, 0)$ & $(0, 0, 1, 0, 0, 0)$ & $X_0 X_3$ \\
$(0, 1, 0)$ & $(1, 1, 0, 1, 1, 0)$ & $Y_1$ \\
$(1, 0, 1)$ & $(0, 0, 0, 1, 1, 1)$ & $Z_2$ \\
$(0, 0, 1)$ & $(0, 0, 0, 1, 1, 1)$ & $Z_2$ \\
$(1, 0, 1)$ & $(0, 1, 1, 0, 0, 0)$ & $X_5$ \\
$(0, 1, 0)$ & $(1, 1, 0, 0, 0, 0)$ & $X_1$ \\
$(1, 0, 0)$ & $(1, 1, 1, 1, 1, 1)$ & $Y_2$ \\
$(0, 0, 1)$ & $(0, 1, 1, 0, 0, 0)$ & $X_5$ \\
$(1, 0, 0)$ & $(0, 0, 0, 1, 1, 1)$ & $Z_2$ \\
\hline
\end{tabular}
\caption{{\bf A decoding for erroneous $S_2^f$.}  A recovery operation to be applied after an erroneous outcome of $S_2^f$ measurement. Here, $(a_0,a_1,a_2)$ are outcomes of the ancilla qubits $a_0,a_1,a_2$ measurements from a circuit $S_2^f$ shown in Fig.~\ref{fig:circ_syndf}. 
$(_1,s_2,s_3,s_4,s_5,s_6)$  are outcomes of the following circuits $S_1^u,S_2^u,S_3^u,S_4^u,S_5^u,S_6^u$ measurements from Fig.~\ref{fig:circ_synduf}. In the case of any combination of $(a_0,a_1,a_2) \neq (0,0,0)$ and $(_1,s_2,s_3,s_4,s_5,s_6)$  not listed in the table, a run is discarded as likely irrecoverably corrupted. } 

\label{tab:dec_sf2}
\end{table}

Finally, we measure the physical registers of all logical qubits in a computational basis. To decode a logical outcome of the measurement and correct errors that were not corrected during the QEC rounds, we compare the outcomes with the results of $\ket{\overline{0}}$ and $\ket{\overline{1}}$ measurements. For each encoded qubit, we determine an outcome of $\ket{\overline{0}}$ or $\ket{\overline{1}}$ measurement in the computational basis that is closest according to the Hamming distance. If the closest outcome comes from $\ket{\overline{0}}$ ($\ket{\overline{1}}$), we decode the logical outcome of the qubit measurement to be $0$ ($1$), respectively. This decoding was implemented with a real-world quantum computer in~\cite{ryan2021realization}.

\section{Additional numerical results for depolarizing noise}
\label{app:RB}

In this section, first we present our full randomized benchmarking results in Fig.~\ref{fig:RB_full}{\bf (a)} for all sequence lengths $m=1-1400$, including $m\ge700$ values omitted in Fig.~4 (main text) for the sake of plot transparency. The complete data is fitted well by the exponential decay of Eq.~(22) in the main text, reinforcing the conclusion that a set of partially error-corrected two-qubit Clifford gates have higher errors than the corresponding noisy-noisy gates. 
Furthermore, in Fig.~\ref{fig:RB_full}{\bf (b)}, we show the probability that a run was post-selected based on the syndromes for the setups involving error-corrected qubits. We estimate this probability as a ratio of the post-selected runs to all runs.  We find that for $m$ values used here, the post-selection probability is larger than $0.74$ and can be well-fitted by an exponential decay of form $c^{-m}$ with $c=0.999591(2)$ for two error-corrected qubits and $c=0.999805(2)$ for the partial error-correction setup.

 \begin{figure}[t]
    \centering 
    \includegraphics[width=\columnwidth]{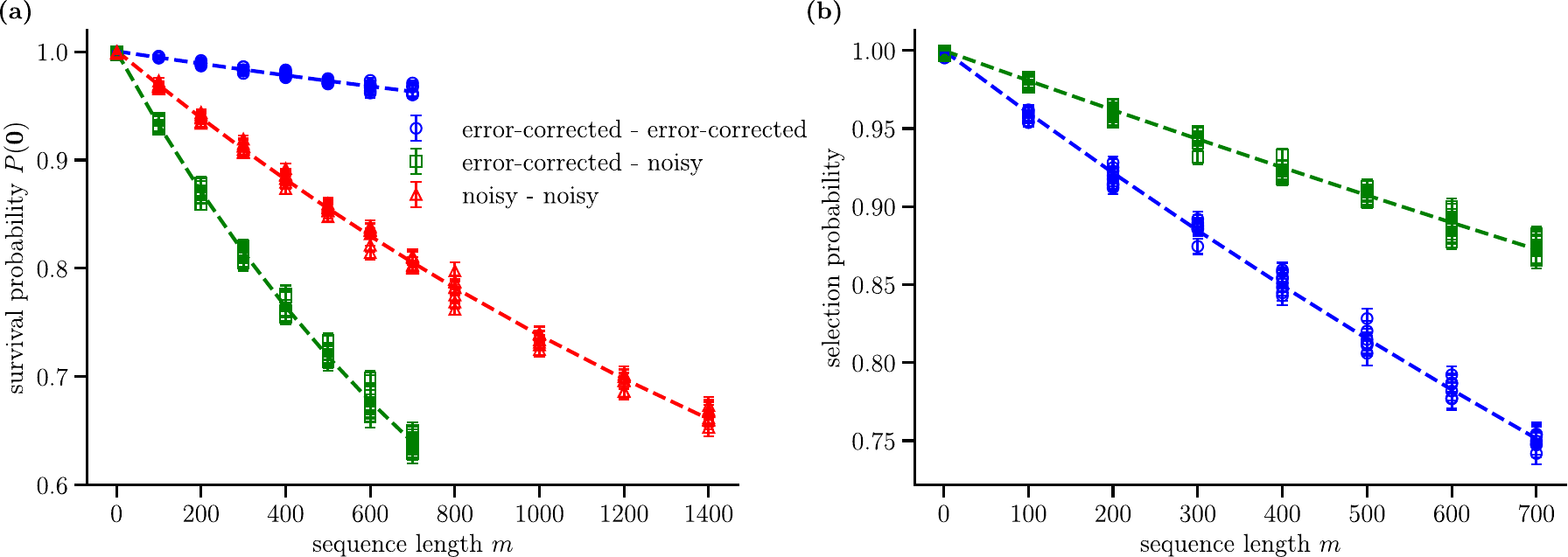}
    \caption{ {\bf Complete results for randomized benchmarking of two-qubit gates.} In {\bf (a)},  we plot the survival probability $P(\textbf{0})$ of a compilation of identity which is a sequence of $m$ random two-qubit logical Clifford gates for two error-corrected qubits (circles), two noisy qubits (triangles), and an error-corrected and a noisy qubit (squares). We include $m>700$ results that are not shown in Fig.~4 (main text) to improve the plot transparency.  
    The dashed lines show the best fits of the data by an exponential decay~(22). 
    In {\bf (b)}, we show the probability that a run is accepted by the post-selection based on syndromes. The dashed lines in {\bf (b)} are the best fits of an exponential decay  $c^{-m}$, with $c$ being a fitted constant.  In both plots, for each sequence, we plot the error bar computed as standard deviation of the mean.}
    \label{fig:RB_full}
\end{figure}

Second, in Fig.~\ref{fig:mirr_Lplt_id0} we show the fidelity versus the layer number $L$ for the random mirrored Clifford circuits described in main text Sec. ``Random Clifford circuits: setup'' in the main text and simulated without the idling noise ($\varepsilon_I=0$). We use a setup ($n=20$, $\varepsilon_1=1.5\cdot10^{-5}, \varepsilon_{2}=1.5\cdot10^{-4}$) that differs from the one used to obtain  Fig.~6 (main text) only by the lack of the idling noise. We find that the fidelity behaves qualitatively the same as in Fig.~6 (main text), which further demonstrates the threshold clean qubit number for the partial error correction. We obtain smaller $n_{\rm threshold}$ value than for $\varepsilon_I=0.75\cdot10^{-6}$ used in  Fig.~6 (main text). This result shows sensitivity of $n_{\rm threshold}$ to the idling strength. 
Moreover, for completeness, in Fig.~\ref{fig:mirr_Lplt_acc} we show the post-selection probability for the mirrored random Clifford circuits simulations from Fig.~6 (main text) and Fig.~\ref{fig:mirr_Lplt_id0}. Finally, the post-selection acceptance rates for simulations from Fig.~8 (main text), are in Tabs.~\ref{tab:psel_fidelity_QEC},~\ref{tab:psel_fidelity_QED}.

\begin{figure}[t]
    \centering 
    \includegraphics[width=0.5\columnwidth]{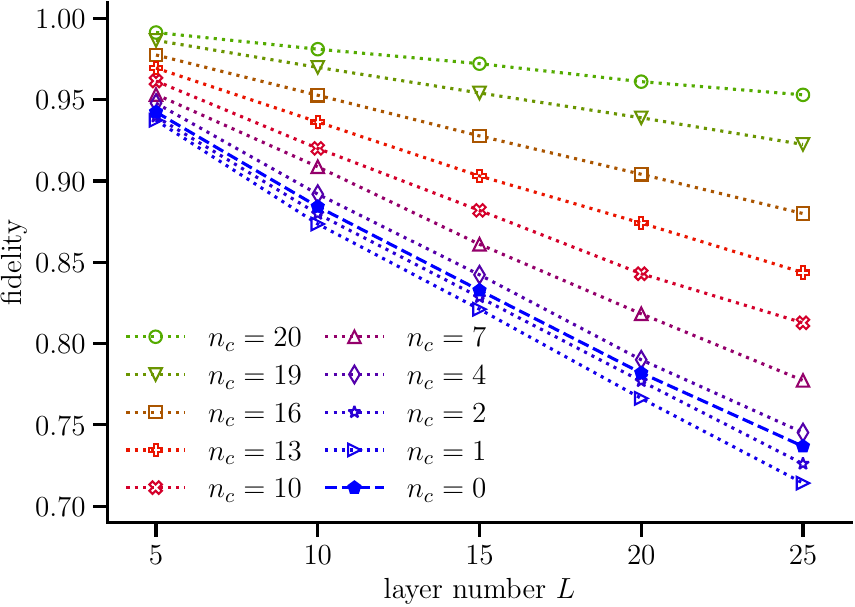}
    \caption{ {\bf Random Clifford circuits without idling noise.} Average fidelities of the noisy and noiseless states prepared by random mirrored Clifford circuits with $n_c$ clean qubits, $n=20$ qubits in total plotted versus the number of layers $L$. 
    The setup differs from Fig.~6 (main text) only by lack of the idling noise ($\varepsilon_I=0$). The same as in Fig.~6 (main text),  the fidelity is estimated as a mean of the survival probabilities of post-selected circuit runs with $2.5\cdot10^4-5\cdot10^5$ total runs  per a data point, and a single run per a random circuit. The statistical uncertainties of the data points given by a standard  deviation of the mean are smaller than the marker sizes. }
    \label{fig:mirr_Lplt_id0}
\end{figure} 

 \begin{figure}[t]
    \centering 
    \includegraphics[width=0.99\columnwidth]{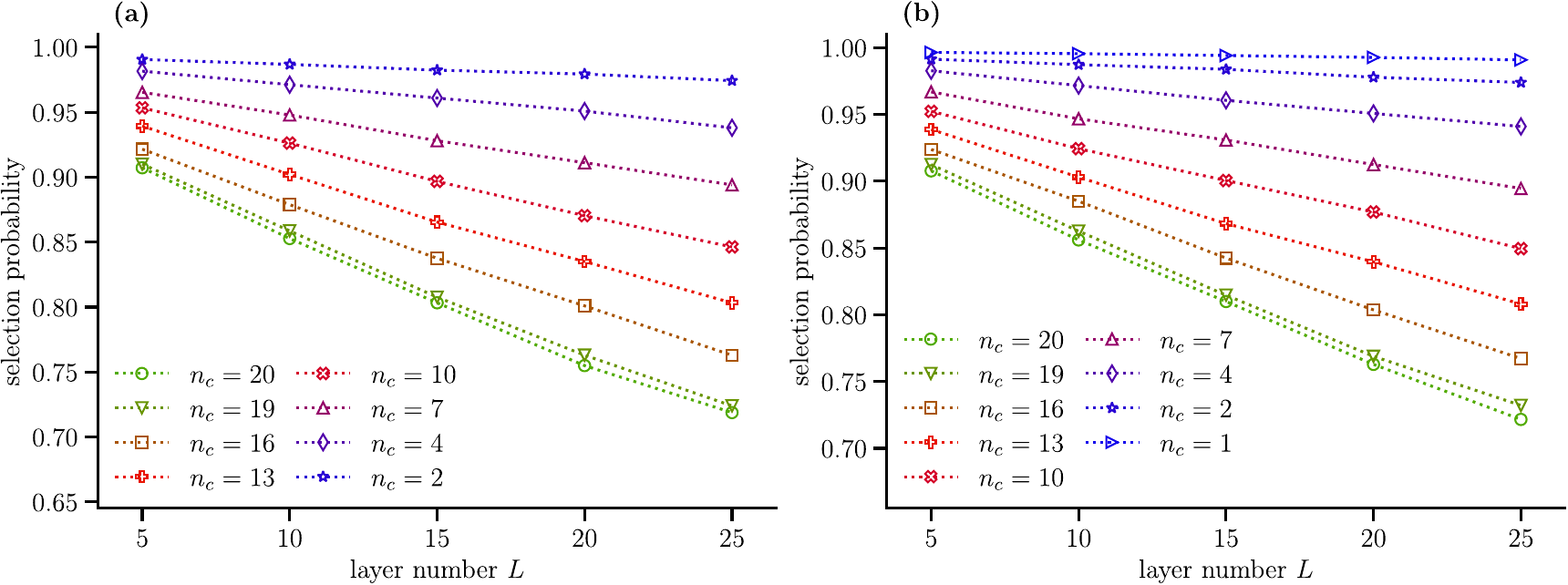}
    \caption{ {\bf Post-selection probabilities for the random Clifford circuits simulations.} In {\bf (a)} the probabilities for the simulations from Fig.~6 (main text -- $\varepsilon_I=0.75\cdot10^{-6}$), and in {\bf (b)} for the ones from Fig.~\ref{fig:mirr_Lplt_id0} ($\varepsilon_I=0$).  The probabilities are computed as the ratios of the accepted runs to all runs. The statistical uncertainties of the data points given by  standard  deviation of the mean are smaller than the marker sizes.}  
    \label{fig:mirr_Lplt_acc}
\end{figure}

\begin{table}
\begin{tabular}{|c|c|c|}
\hline
$n_d$ & $n_c$ & selection  probability \\
\hline
0 & 12 & 0.882(2) \\
2 & 12 & 0.879(1) \\
4 & 12 & 0.876(2) \\
6 & 12 & 0.875(2) \\
8 & 12 & 0.880(1) \\
10 & 12 & 0.878(1) \\
\hline
\end{tabular}
\caption{{\bf Post-selection probabilities for  random Clifford circuits simulations with $n=22$ and partial QEC from Fig.~8 (main text).}  The probabilities are computed as the ratios of the accepted runs to all runs. The statistical uncertainties are  given by standard  deviation of the mean. } 
\label{tab:psel_fidelity_QEC}
\end{table}

\begin{table}
\begin{tabular}{|c|c|c|c|}
\hline
$n_d$ & $n_c$ & selection  probability  (depolarizing noise) &  selection  probability (H2-informed noise)\\
\hline
0 & 12 & 0.112(2) & 0.113(2) \\
2 & 12 & 0.095(2) & 0.086(2) \\
4 & 12 & 0.101(2) & 0.083(2) \\
6 & 12 & 0.097(2) & 0.087(2) \\
8 & 12 & 0.095(2) & 0.086(1) \\
10 & 12 & 0.097(2) & 0.086(2) \\
\hline
\end{tabular}
\caption{{\bf Post-selection probabilities for the random Clifford circuits simulations with partial error detection from Fig.~8 (main text).}  The probabilities are computed as the ratios of the accepted runs to all runs. The statistical uncertainties are  given by standard  deviation of the mean. } 
\label{tab:psel_fidelity_QED}
\end{table}

\section{Quantinuum's H2-informed noise model. }

For the quantum error detection simulations in main text Sec. ``Random Clifford circuits: results for quantum error detection (QED)'' we use a depolarizing noise model with error rates  based on data from Ref.~\cite{moses2023race}.

We model idling errors, SPAM errors and single-qubit gate errors as single-qubit depolarizing noise with error rates
\begin{equation}
\varepsilon_{I}= 3.3 \cdot 10^{-4}, 
\end{equation}
\begin{equation}
\varepsilon_{SPAM}= 1.2 \cdot 10^{-3} + \varepsilon_I,
\end{equation}
\begin{equation}
\varepsilon_{1}=  3.75 \cdot10^{-5} + \varepsilon_I.
\end{equation}
We define single-qubit depolarizing noise as 
\begin{equation}
\mathcal{D}_{1}^{(\varepsilon)}(\rho) = (1-\varepsilon) \rho + \frac{\varepsilon}{3} \big[X \rho X + Y \rho Y +  Z \rho Z \big],
\end{equation}
and refer to $\varepsilon$ as the noise rate.
We model $CNOT$ errors by two-qubit and single-qubit  depolarizing noise
$$
\mathcal{D}_{2}^{(\varepsilon_{ CNOT})}\circ\big(\mathcal{D}_{1}^{(\varepsilon_{I})} \otimes \mathcal{D}_{1}^{(\varepsilon_{I})}\big)(\rho)
$$
with  two-qubit depolarizing noise defined as
\begin{equation}
\mathcal{D}_{2}^{(\varepsilon)}(\rho) = (1-\varepsilon) \rho + \frac{\varepsilon}{15} \sum_{P \in \mathrm{P}_2\setminus\{{\id}^{\otimes 2}\}} P \rho P ,
\end{equation}
where the sum runs over all non-identity single and two-qubit Paulis. We have 
\begin{equation}
\varepsilon_{ CNOT} = 2.2\cdot10^{-3}.
\end{equation}
The idling error rate is chosen as $\frac{3}{2} \times (\text{infidelity per qubit per layer})$ estimated from interleaved transport randomized benchmarking~\cite{moses2023race}. 
That experiment measured the infidelity of transport required to execute layers  built of random pairs of 2-qubit gates.
Therefore, the resulting transport time is likely significantly higher than for circuits with
linear connectivity and large number of single-qubit gates which are considered in our numerical simulations.
Consequently, the idling rate estimated from that experiment is likely significantly higher than for circuits that we consider in our numerical experiments. 
Nevertheless, in our case we use the experimental idling rate, even if it is overtly pessimistic in our case, 
since, according to our best knowledge, circuit-specific idling rates are not publicly available.
The proportionality factor relates the single-qubit depolarizing noise infidelity to the channel error rate~\cite{nielsen2000quantum, nielsen2002simple}.
Analogously, we choose $\varepsilon_{1}$ as $\frac{3}{2} \times (\text{average infidelity per Clifford})$ from single-qubit randomized benchmarking~\cite{moses2023race} plus a contribution from memory errors caused by gate transport during a layer execution.

More  generally,  gates in a layer usually cannot be executed simultaneously for trapped-ion quantum computers.  Frequently, all qubits  need to be transported to execute a layer. Therefore, we include a transport error contribution for each operation during a circuit execution including gates, state preparation and measurement. For simplicity, we assume that, for each layer, each qubit is affected by tarnsport errors modeled by single-qubit depolarzing noise with the error rate $\varepsilon_I$ irrespective of the operations in the layer. As in a real-world case, the transport schedules can be optimized to reduce transport errors, so our assumption likely overestimates the real-hardware rate of the memory errors. 

The state preparation  (measurement) errors are modeled as single-qubit depolarizing channels following the state preparation (preceding the measurement). The non-transport contribution to this error rate is taken as $\frac{3}{4} \times (\text{average SPAM error per qubit})$. The prefactor takes into account that the SPAM experiment does not distinguish preparation and measurement errors and assumes that their rate is the same in both cases~\cite{moses2023race}.
For $CNOT$ gates, transport memory errors during layer execution are modeled by  $\mathcal{D}_{1}^{(\varepsilon_{I})}$ acting at each qubit and the non-transport part is modeled as two-qubit depolarizing noise with a rate chosen as $\frac{5}{4} \times (\text{average infidelity per two-qubit gate})$ estimated from two-qubit randomized benchmarking~\cite{moses2023race}. The prefactor relates the noise error rate to the channel infidelity, and we assume that a single native two-qubit gate is used per $CNOT$ gate~\cite{trout2018simulating}.

\section{Random Clifford circuits with a non-Pauli noise model}
\label{app:rand_Cliff_Our}

\subsection{Setup}
\label{sec:Cliff_setup_Our}
We simulate random mirrored Clifford circuits which compile to identity of a form of $VV^\dag$ (Eq.~(23) in main text). 
Here $V$ consists of $L$ layers of an ansatz built from a layer of randomly chosen single-qubit Clifford gates followed by two layers of  alternating nearest-neighbor $CNOT$s arranged according to a brickwall pattern. This choice closely resembles random circuits from main text Sec. ``Random Clifford Circuits: setup'' in the main text. The single-qubit gates are chosen from the gate set $\{H, X, Y, Z, S, S^{\dag} \}$, and the circuits are compiled for the qubits arranged into a line, 
with the first $n_c$ registers being clean and the remaining $n_d$ registers being noisy. The initial state is $\ket{\overline{0}}^{\otimes n_c} \ket{0}^{\otimes n_d}$, again the same as in main text Sec. ``Random Clifford circuits: setup'' in the main text. 

We simulate the circuits using a noise model informed by the gate set tomography of IBM's Ourense quantum computer~\cite{cincio2020machine}. The model assumes a  physical gate set consisting of $R_Z(\theta) = e^{-i (\theta/2) Z}$, $\sqrt{X}= e^{-i (\pi/4) X}$, and the $CNOT$ gate. These   gates are  by process matrices obtained from IBM's Ourense process tomography  and decreasing their infidelity  by a factor of $1/q=100$. 
Furthermore, we add an identity gate $I$ to the gate set to represent an idling noisy qubit.
To investigate the effects of idling noise in detail, we consider the process matrix of $I$ with the idling infidelity decreased with respect to the Ourense quantum computer by multiple factors  $1/q_I=100,200,400,800$, and without idling noise ($q_I=0$).  We provide details on the construction of the process matrices in SI Sec.~\ref{app:noise_eff}. This noise model captures error processes in a real-world quantum computer which typically goes beyond Pauli noise due to a variety of reasons like non-perfect Pauli-twirling and leakage detection~\cite{varbanov2020leakage}. 

In the case of gates acting on noisy registers only, we find their noisy process matrices by directly compiling them to the physical gates. For logical gates involving clean qubits, which are implemented with the Steane code, we use logical process tomography~\cite{suzuki2022quantum} to obtain their effective process matrices representing $\PC_c \circ \U_{\rm logical}$ of these gates. Process tomography is performed directly by simulating the action of noisy gates and noiseless QEC rounds on underlying physical qubits, assuming perfect noiseless logical state preparation and perfect logical measurement. We describe the procedure in detail in SI Sec.~\ref{app:noise_eff}. While this approach leads to an unrealistic treatment of errors happening during the error correction procedures, it enables us to simulate at scale effects of a noise that can not be simulated efficiently by methods based on the stabilizer formalism. Furthermore, it models accurately errors occurring at the noisy qubits and their propagation to the error-corrected qubits through $b$-type CNOTs. Such errors are specific to the partial error correction and understanding of their impact on the framework performance is the main focus of this work.

\subsubsection{Results}

We estimate the fidelity of the noisy and the noiseless state for the mirrored Clifford circuits with $n=10-30$,     $L=5-50$, and multiple choices of $0\le n_c\le n$.  
Numerical simulations are performed with a matrix-product-state~\cite{fannes1992finitely,orus2014practical} density matrix simulator utilizing the one-dimensional structure of the circuits. Additionally, the mirrored character of the circuits further decreases the simulation cost since the noiseless state has a trivial matrix-product-state representation.   For each $n$, $n_c$, and $L$, we compute the average fidelity as the mean of fidelities of $100$ random circuits. For all simulated system sizes and circuit depths, we observe that the fidelity at $n_c=1$ is lower than that at $n_c=0$, but then it subsequently improves approximately geometrically with increasing $n_c$. Such behavior reflects our analytical result in Theorem 2. The same as in the case of the depolarizing noise results from the main text, the fidelity improves upon $n_c=0$ at some threshold value $n_{\rm threshold}<n$. We find that $n_{\rm threshold}$ does not depend on the layer number and is determined by $n$ and the noise model. Moreover, in the presence of no idling noise, this is a constant with respect to increasing system size. This parallels the threshold behavior for the depolarizing noise considered in main text Sec. ``Random Clifford circuits: setup'' in the main text.  We show the fidelities for the largest simulated $n=30$ and $1/q_I=800$ in Fig.~\ref{fig:eff_res}.

\begin{figure}[t]
    \centering 
    \includegraphics[width=0.5\columnwidth]{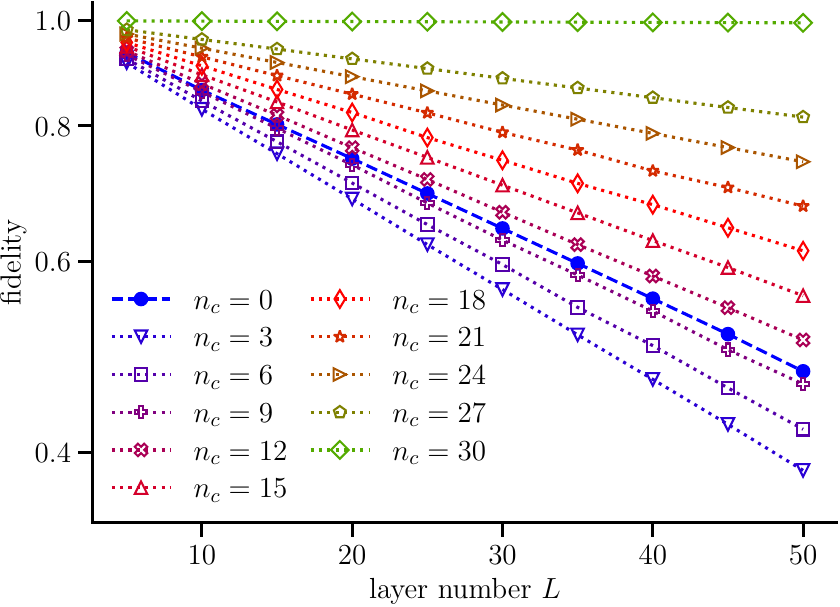}
    \caption{ {\bf Random Clifford circuits with a non-Pauli noise model.} Average fidelities of the noisy and noiseless states prepared by random mirrored Clifford circuits with $n_c$ clean qubits, $n=30$ qubits in total plotted versus the number of layers $L$.   We simulate the setup with a  noise model informed by real-device noise and taking into account idling noise ($q_I=1/800$)  described in detail in SI Sec.~\ref{app:noise_eff}. We use $100$ random circuits per data point. Note that the $y$ scale is logarithmic, indicating exponential decay of the fidelities.}
    \label{fig:eff_res}
\end{figure} 

We examine the dependence of $n_{\rm threshold}$ on $n$ and the idling noise in more detail in  Fig.~\ref{fig:eff_idling}, plotting it against $n$ for $q_I/q=0,0.125,0.25,0.5,1$ corresponding to a ratio of noisy $I$ and $CNOT$ gates infidelities  $\epsilon_I/\epsilon_{CNOT}=0,0.019,0.039,0.078,0.156$, respectively (see SI Sec.~\ref{app:noise_eff} for details). As for the depolarizing noise simulations in Fig.~7 (main text), we find that without idling noise $n_{\rm threshold}\approx5$ and does not depend on $n$. Otherwise, it grows linearly with $n$, and the growth slope increases with the idling noise strength.   

  \begin{figure}[t]
    \centering 
    \includegraphics[width=0.5\columnwidth]{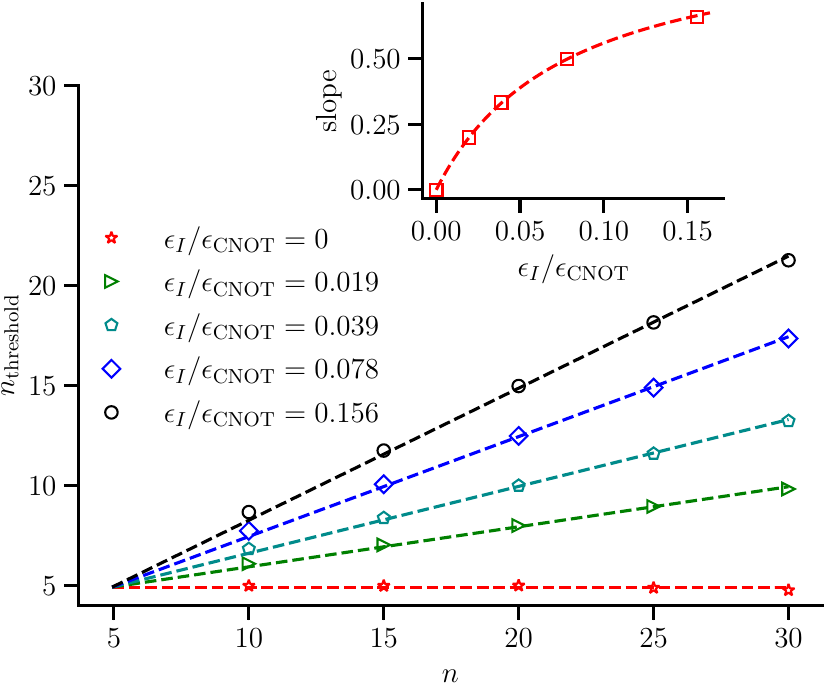}
   \caption{{\bf The clean qubit number required for an advantage over an all-noisy device.} Here we plot $n_{\rm threshold}$ for the random mirrored Clifford circuits with $n=10-30$ and several ratios of the noisy identity and $CNOT$ gates infidelities $\epsilon_I/\epsilon_ {CNOT} =0,0.019,\dots,0.156$, see SI Sec.~\ref{app:noise_eff} for details. The results are fitted by a linear ansatz of Eq.~(25) (main text -- dashed lines) with the slope dependent on the idling noise strength, unlike the intercept.  In the inset, we show the dependence of the slope on $ \epsilon_I/\epsilon_{CNOT} $ fitted by a heuristic phenomenological ansatz of Eq.~(26) (main text -- dashed line).} 
    \label{fig:eff_idling}
\end{figure}

We apply the heuristic predictions of the scaling in the main text Eq.(25) to our numerical results as shown in   Fig.~\ref{fig:eff_idling}. The same as in the depolarizing noise case, for each $\epsilon_I/\epsilon_{CNOT} $ we fit  $n_{\rm threshold}(n)$ with the linear ansatz in main text Eq.~(25) taking $a$ as the fit parameter and estimating $b$ by fitting the no idling noise case with the constant ansatz given in main text Eq.~(27). We obtain good agreement with the numerical results. Secondly, we fit $a$ versus $\epsilon_{I}/\epsilon_{CNOT} $ in more detail, by
a functional dependence of main text Eq.~(26). Again, we obtain good agreement with the numerical results (see Fig.~\ref{fig:eff_idling}). The qualitative agreement of the $n_{\rm threshold}$ scaling results with the depolarizing noise scaling suggests that the observed scaling is valid for a broad class of realistic noise cases.

\section{Non-Pauli noise model}
\label{app:noise_eff}

We consider a noise model with the following gates acting on physical qubits: a $CNOT$ gate, an arbitrary rotation around $z$ axis $R_Z(\theta) = e^{-i (\theta/2) Z}$, and $\sqrt{X}= e^{-i (\pi/4) X}$, where $Z$ and $X$ are Pauli matrices. Furthermore,  we have an identity gate $I$ corresponding to an idling physical qubit. The noisy physical gates are defined with process matrices obtained by gate set tomography of IBM's Ourense quantum computer $ {CNOT} ^{\rm Ourense}, \sqrt{X}^{\rm Ourense}, I^{\rm Ourense}$ and process matrices of the perfect (not affected by the noise) gates $ {CNOT} ^{\rm noiseless}, \sqrt{X}^{\rm noiseless}, I^{\rm noiseless}, R_Z^{\rm noiseless}(\theta)$. The Ourense process matrices can be found in Ref.~\cite{cincio2020machine}. In our numerical implementation, we consider the following maps: 
\begin{align}
 {CNOT}  &= (1-q) {CNOT} ^{\rm noiseless} + q  {CNOT} ^{\rm Ourense}\,, \\
 \sqrt{X} &= (1-q)\sqrt{X}^{\rm noiseless} + q \sqrt{X}^{\rm Ourense}\,, \\
 I &= (1-q_I)I^{\rm noiseless} + q_I I^{\rm Ourense}\,.
 \end{align}
To obtain the error rates enabling successful quantum error correction, we choose $q=0.01$, since the average gate infidelities of $ {CNOT} ^{\rm Ourense}$, $\sqrt{X}^{\rm Ourense}$  are $1.9 \times 10^{-2}$ and  $8.8 \times 10^{-4}$, respectively~\cite{cincio2020machine}. Such a choice corresponds roughly to a noise strength reduction by a factor of $100$. To examine in detail the effects of idling noise, we consider a range of $q_I=0-0.01$. We note that the average gate infidelity of $I^{\rm Ourense}$ is $2.8 \times 10^{-3}$~\cite{cincio2020machine}.  We assume that  $R_Z(\theta)$ is noiseless, as in the case of the IBM's quantum computers,   that is,
 \begin{equation}
R_Z(\theta) = R_Z^{\rm noiseless}(\theta)\,.
\end{equation}

Single-qubit Clifford gates $X,Y,Z,H,S,S^{\dag}$ acting at noisy registers are obtained by compilation to the physical gates resulting in 
\begin{align}
X &= \sqrt{X}\sqrt{X}\,, \\
Y &= R_Z(\pi) \sqrt{X} \sqrt{X}\,,  \\
Z &= R_Z(\pi)\,, \\
S &= R_Z(\pi/2)\,, \\
S^{\dag} &= R_Z(-\pi/2)\,, \\
H  &= R_Z(\pi/2) \sqrt{X} R_Z(\pi/2)\,,
\end{align} 
where we assume that the gates are executed from right to left. 

We assume that the clean registers are encoded with the Steane code. Therefore, we have
 \begin{align}
\overline{X} &= {\rm QEC} \circ \big(\prod_{i=0}^6 X_{i} \big)\,, \\
\overline{Y} &= {\rm QEC} \circ \big(\prod_{i=0}^6 Y_{i} \big)\,, \\
\overline{Z} &= \prod_{i=0}^6 R_{Z}(\pi)_i\,,   \label{eq:Z_impl}\\
\overline{S} &= \prod_{i=0}^6 R_{Z}(3\pi/2)_i \,, \\
\overline{S}^{\dag} &= \prod_{i=0}^6 R_{Z}(-3\pi/2)_i\,,\\ 
\overline{H} &= {\rm QEC} \circ \big(\prod_{i=0}^6 H_{i} \big)\,,
\end{align}
where $i \in \{0,1,\dots,6 \}$ index physical registers constituting one clean register. Here ${\rm QEC}$ is an error correction round, i.e., a syndrome measurement and recovery operation. We note that we do not follow $\overline{Z}$, $\overline{S}$, and $\overline{S}^{\dag}$ with error correction rounds as they are implemented by noiseless physical gates. 

We consider here an idealized error correction round defined by a quantum channel
\begin{equation}
{\rm QEC}(\rho) = \sum_s U_s P_s \rho P_s U_s^{\dag}\,,  
\end{equation}
where $s$ numbers outcomes of the syndrome measurements, $P_s$ are projectors at the eigenspace of the Steane code stabilizers for which $s$ is the syndrome outcome, and $U_s$ are corresponding noiseless recovery operation~\cite{nielsen2000quantum}. Such an idealized QEC round corresponding to noiseless syndrome measurement and recovery results in ${\rm QEC}(\rho)$ mapping $\rho$ into the code space.   

$CNOT$s between a noisy and a clean qubit controlled at a noisy and a clean register are implemented as 
\begin{align}
{\widetilde{CNOT}}_{cd} &=  {\rm QEC} \circ \big(\prod_{i=0}^6  {CNOT} _{i,7} \big)\,,  \label{eq:CNOT_cd_impl} \\
{\widetilde{CNOT}}_{dc} &= {\rm QEC} \circ \big(\prod_{i=0}^6  {CNOT} _{7,i} \big)\,,  \label{eq:CNOT_dc_impl} 
\end{align}
respectively.
Here, the clean register consists of physical $0,1,\dots,6$ registers, the noisy register is labeled by 7, and $ {CNOT} _{i,j}$ is a $CNOT$ controlled at a register $i$ with a target $j$.
Finally, a $CNOT$ in between two clean registers is implemented as  
\begin{equation}
    {\overline{CNOT}} = {\rm QEC} \circ \big(\prod_{i=0}^6  {CNOT} _{i,i+7} \big)\,, 
\end{equation}
where for convenience we assume that the clean registers are encoded with physical registers $0,1,\dots,6$ and $7,8,\dots,13$ (as shown in Fig.~\ref{fig:circ_CNOTcc}), and the QEC round is applied to both of them. Compiling ${\rm \widetilde{CNOT}}$ and ${\rm \overline{CNOT}}$, we take into account the idling noise by inserting gates at each idling physical register 
$I$ accordingly~\cite{cincio2020machine}. We supplement this set of gates with a logical identity gate corresponding to an idling clean qubit
\begin{equation}
\overline{I} = {\rm QEC} \circ \big(\prod_{i=0}^6 I_{i} \big)\,.
\label{eq:CNOTb_Our}
\end{equation}

We note that this particular choice of the logical gate implementation  is not  unique, as in the case of the Steane code one can also choose 
\begin{equation}
\overline{Z}=  R_{Z}(\pi)_1  R_{Z}(\pi)_3  R_{Z}(\pi)_5, 
\end{equation}
with the physical indices labeled as in (\ref{eq:Z_impl}).
This  alternative choice of $\overline{Z}$ in particular enables implementation of the noisy-error-corrected couplings as
\begin{align}
    {\widetilde{CNOT}}_{cd} &= {\rm QEC} \circ \big(\prod_{i=1,3,5}  {CNOT} _{i,7} \big)\,, \\
    {\widetilde{CNOT}}_{dc} &=  {\rm QEC} \circ  \big(  \prod_{i=1,3,5}  {CNOT} _{7,i}   \big) \label{eq:CNOT_cd_impl2}  \,,  
\end{align}
with the physical indices labeling as in (\ref{eq:CNOT_dc_impl}), and where we use main text Eq.~(10). Such an implementation is used in the depolarizing noise simulations, see SI Sec.~\ref{app:QEC_impl}, as it reduces the number of the physical $CNOT$ gates and the circuit depth in comparison to (\ref{eq:CNOT_dc_impl}-\ref{eq:CNOT_cd_impl}). Since in our numerical work utilizing the non-Pauli  we are interested in a qualitative comparison of the threshold behavior with the more realistic depolarizing noise results, we expect that the usage of the less efficient $b$-$CNOT$ compilation (\ref{eq:CNOT_dc_impl}-\ref{eq:CNOT_cd_impl}) does not affect conclusions of SI Sec.~\ref{app:rand_Cliff_Our}.

We assume perfect state preparation of the initial state $\rho_i = \ket{\vec{0}}\bra{\vec{0}},\,\ket{\vec{0}}=\ket{\overline{0}}^{\otimes n_c} \ket{0}^{\otimes n_d}$ and perfect QEC rounds that map an erroneous state to the code space. More precisely, we work within a reduced density matrix formalism  and simulate the QEC round as a channel acting a state of a logical qubit $\rho$ as
\begin{equation}
{\rm QEC} (\rho) = \sum_{s_1, s_2, \dots, s_6} R_{s_1 s_2, \dots, s_6 } P_{s_1 s_2, \dots, s_6 } \rho P_{s_1 s_2, \dots, s_6 }^{\dag} R_{s_1 s_2, \dots, s_6 }^{\dag}. 
\end{equation}
Here, $s_1, s_2, \dots, s_6$ are measurement outcomes of the Steane's code stabilizers
\begin{equation}
g_1 = Z_0 Z_1 Z_2 Z_3, \quad g_2 = Z_1 Z_2 Z_4 Z_5, \quad g_3 = Z_2 Z_3 Z_5 Z_6, \quad  g_4 = X_0 X_1 X_2 X_3, \quad g_5 = X_1 X_2 X_4 X_5, \quad g_6 = X_2 X_3 X_5 X_6, 
\end{equation}
with $s_i \in \{1,-1\}$ being an outcome of $g_i$ measurement. $ P_{s_1 s_2, \dots, s_6 }$ is a projector at a  the stabilizers' eigenspace such that $g_i$'s eigenvalue is $s_i$, i.e
\begin{equation}
P_{s_1 s_2, \dots, s_6 } = \prod_{i=1}^6 \bigg( \frac{1+s_i g_i}{2} \bigg).
\end{equation}
$R_{s_1 s_2, \dots, s_6 }$ is a recovery operation that maps the eigenspace, which $P_{s_1 s_2, \dots, s_6 }$ projects onto, to the code space. This recovery operation is a Pauli with the smallest weight that satisfies this condition and is identified algebraically using a check matrix formalism~\cite{nielsen2000quantum}. When all $s_i=1$, $R_{s_1 s_2, \dots, s_6 }$ reduces to an identity.

These assumption ensure that the state of each clean qubit before and after each gate action is in the code space of the Steane code spanned by  $\ket{\overline{0}}$ and $\ket{\overline{1}}$. We obtain process matrices of the gates acting on the clean qubits with quantum process tomography~\cite{nielsen2000quantum, suzuki2022quantum}. Furthermore, since under the above assumptions, the states of clean registers belong to the code space both before and after the gate action, we compute only the process matrix elements within the code space. In the case of a single-qubit gate $g$, we obtain its process matrix $\mathcal{E}_{g}$ by computing the gate action on basis states spanning the code space, namely $\mathcal{E}_{g}(\ket{\overline{n}}\bra{\overline{m}})$, $n,m\in\{0,1\}$. Analogously, we obtain process matrices of a two-qubit gate acting on one or two clean qubits computing its action on the basis states  $\ket{\overline{n^1}, n^2}\bra{\overline{m^1},m^2}$ or  $\ket{\overline{n^1}, \overline{n^2}}\bra{\overline{m^1},\overline{m^2}}$, respectively,  with  $n^1,n^2,m^1,m^2 \in \{0,1\}$.

We gather average infidelities of the gates computed by averaging over $10000$ random states sampled according to the Haar measure in Tab.~\ref{tab:gate_infid},  \ref{tab:gate_infid2}. In the tables, we also list the circuit depths of the logical gates when compiled to the physical gates that act on the physical registers. These depths are used to account for idling while compiling a circuit to the logical gates. The compilation assumes that a logical gate with the depth  $k$  takes as much time as $k$ logical idling gates. For simplicity, we do not account for the physical gates required to perform the QEC round when computing the logical gate depths. Furthermore, we do not count $R_Z$ gates in computations of the gate depths assuming that they can be performed virtually  without increasing the idling time, as for the IBM's quantum computers. 

\begin{table}
\begin{tabular}{|c|c|c|}
\hline
gate & average infidelity & depth \\
\hline
$X$ & $1.8\times10^{-5}$ & $2$ \\
$Y$ & $1.8\times10^{-5}$ & $2$ \\
$Z$ & $0$ & $0$ \\
$S$ & $0$ & $0$ \\
$S^{\dag}$ & $0$ & $0$ \\
$H$ & $8.7\times10^{-5}$ & $1$  \\
$\overline{X}$ & $7.3\times10^{-9}$ & $2$ \\
$\overline{Y}$ & $7.3\times10^{-9}$ & $2$ \\
$\overline{Z}$ & $0$ & $0$ \\
$\overline{S}$ & $0$ & $0$ \\
$\overline{S}^{\dag}$ & $0$ & $0$ \\
$\overline{H}$ & $1.8\times10^{-9}$ & $1$  \\
$I$ & $2.8\times10^{-5}$ & $1$ \\
$\overline{I}$ & $1.7\times10^{-8}$ & $1$ \\
$ {CNOT} $ &$1.9\times10^{-4}$ & $1$\\ 
${\rm \overline{CNOT}}$ &$6.5\times10^{-6}$ & $7$\\
${\rm \widetilde{CNOT}}_{cd}$ &$1.6\times10^{-3}$ & $7$\\
${\rm \widetilde{CNOT}}_{dc}$ &$1.6\times10^{-3}$ & $7$\\
\hline
\end{tabular}
\caption{{\bf Average infidelities and depths of the logical gates in the device-inspired noise model.} Here $q=0.01$, and $q_I=0.01$. Average infidelities are computed from $10000$ random states sampled according to the Haar measure. Gate depths are numbers of non-parallelizable layers of physical gates acting on the physcial qubits  required for the gate's implementation.} 
\label{tab:gate_infid}
\end{table}

\begin{table}
\begin{tabular}{|c|c|c|c|}
\hline
gate & $q_I$ & average infidelity  \\
\hline
$I$ & $0.005$ & $1.4\times10^{-5}$  \\
$\overline{I}$ & $0.005$ & $4.3\times10^{-9}$ \\
${\rm \overline{CNOT}}$ & $0.005$ & $3.0\times10^{-6}$ \\
${\rm \widetilde{CNOT}}_{cd}$ & $0.005$ & $1.3\times10^{-3}$ \\
${\rm \widetilde{CNOT}}_{dc}$ & $0.005$ & $1.4\times10^{-3}$ \\
$I$ & $0.0025$ & $7.0\times10^{-6}$  \\
$\overline{I}$ & $0.0025$ & $1.0\times10^{-9}$ \\
${\rm \overline{CNOT}}$ & $0.0025$ & $1.7\times10^{-6}$ \\
${\rm \widetilde{CNOT}}_{cd}$ & $0.0025$ & $1.1\times10^{-3}$ \\
${\rm \widetilde{CNOT}}_{dc}$ & $0.0025$ & $1.3\times10^{-3}$ \\
$I$ & $0.00125$ & $3.5\times10^{-6}$  \\
$\overline{I}$ & $0.00125$ & $2.7\times10^{-10}$ \\
${\rm \overline{CNOT}}$ & $0.00125$ & $1.2\times10^{-6}$ \\
${\rm \widetilde{CNOT}}_{cd}$ & $0.00125$ & $1.1\times10^{-3}$ \\
${\rm \widetilde{CNOT}}_{dc}$ & $0.00125$ & $1.2\times10^{-3}$ \\
$I$ & $0$ & $0$ \\
$\overline{I}$ & $0$ & $0$ \\
${\rm \overline{CNOT}}$ & $0.00125$ & $8.3\times10^{-7}$ \\
${\rm \widetilde{CNOT}}_{cd}$ & $0.00125$ & $1.0\times10^{-3}$ \\
${\rm \widetilde{CNOT}}_{dc}$ & $0.00125$ & $1.1\times10^{-3}$ \\
\hline
\end{tabular}
\caption{{\bf Average logical gate infidelities for smaller idling noise strengths $q_I=0-0.005$.} Here $q=0.01$.  The infidelities are computed from $10000$ random states sampled according to the Haar measure. We list only gates for which the infidelities depend on the idling noise strength. }
\label{tab:gate_infid2}
\end{table}

\end{document}